%% file: arXiv_v2.tex
\documentclass[12pt]{article}
\usepackage[english]{babel}

\usepackage{fullpage}
\usepackage{cite}

\input{macros}

\usepackage{authblk}
\usepackage[dvipsnames]{xcolor}

\renewcommand{\EE}[1]{\underset{\scaleobj{.8}{#1}}{\mathds{E}\,}}
\renewcommand{\PP}[1]{\underset{\scaleobj{.8}{#1}}{\mathds{P}\,}}

\allowdisplaybreaks 
\usepackage{graphicx,import,xcolor,transparent} 
\usepackage{centernot}

\newcommand{\errata}[1]{{\color{Black} #1}}

\usepackage{caption} 
\newcommand{\mycaption}[2]{\caption[#1]{%
{\centering #1 \par}\vspace{0.4em}
        \normalfont #2}
        }

\usepackage{subcaption}
\usepackage{tabularray}
\UseTblrLibrary{booktabs} 
\usepackage{pifont}

\begin{document}

\title{Quantum Shannon theory made robust:\\
a tale of three protocols for almost i.i.d.\ sources}

\author[1]{Filippo Girardi\thanks{filippo.girardi@sns.it}}
\author[2]{Nilanjana Datta}
\author[3]{\\Giacomo De Palma}
\author[1]{Ludovico~Lami}
\affil[1]{Scuola Normale Superiore, Piazza dei Cavalieri 7, 56126 Pisa, Italy}
\affil[2]{Department of Applied Mathematics and Theoretical Physics,

University of Cambridge, Cambridge CB3 0WA, United Kingdom}
\affil[3]{Department of Mathematics, University of Bologna, Piazza di Porta San Donato 5, 40126 Bologna, Italy}

\date{}
\setcounter{Maxaffil}{0}
\renewcommand\Affilfont{\itshape\small}

\maketitle
\begin{abstract}

The asymptotic rates of information-theoretic protocols — including error exponents, data-compression rates, and channel capacities — are traditionally derived under the idealised assumption that the underlying resources are independent and identically distributed (i.i.d.). Somewhat surprisingly, even slight departures from the exact i.i.d.\ structure can drastically alter the asymptotic behaviour predicted by the i.i.d.\ theory. If the precise nature of the perturbation is known, for instance in the case of a pointwise defect, one can design a bespoke protocol that compensates for it, e.g.\ by discarding the corrupted subsystem. In realistic physical settings, however, exact i.i.d.\ behaviour cannot be guaranteed, and deviations from the ideal regime cannot generally be identified precisely. This raises a fundamental question: which notions of almost i.i.d.\ structure are sufficiently robust to preserve the asymptotic predictions of quantum Shannon theory? We investigate this question for three central information-theoretic tasks: asymmetric hypothesis testing, classical and quantum data compression, and classical communication through quantum channels. Rather than designing protocols tailored to specific defects, we seek robust protocols that remain asymptotically optimal and that are universal within a broad class of almost i.i.d.\ resources whose precise deviations from the ideal regime are unknown. To this end, we study three inequivalent notions of almost i.i.d.\ structure, and determine which of them preserve the asymptotic rates and error exponents predicted by the i.i.d.\ theory. Along the way, we introduce the notion of an \emph{almost i.i.d.\ process} and a new distance measure between quantum channels — the \emph{club distance} — designed to capture stability under local perturbations. These notions may be of independent interest.
\end{abstract}

\newpage
\tableofcontents

\newpage

\section{Introduction}

A central paradigm of information theory is that resources\footnote{These include information sources, channels, and entanglement resources.} are independent and identically distributed (i.i.d.), that is, memoryless. In realistic physical settings, however, exact i.i.d.\ behaviour is rarely realised: experimental imperfections (e.g.~local defects), environmental interactions, and residual correlations inevitably lead to departures from the idealised tensor-power structure. This raises a fundamental question of robustness: when do information-theoretic protocols remain stable under perturbations of the i.i.d.\ assumption, and when can weak correlations within the resources give rise to qualitatively different behaviour?

This question has motivated a growing body of work on approximate notions of independence in quantum information theory. One of the earliest systematic approaches was introduced by Renner in the context of quantum cryptography, through the notion of ``almost-power'' states~\cite[Theorem~4.3.2]{RennerPhD}. Related ideas later appeared in work of Brand\~ao and Plenio on the generalised quantum Stein's lemma~\cite[Eq.~(66)--(67)]{Brandao2010}, where controlled deviations from tensor-power structure played a central role, and in subsequent works resolving the same problem~\cite{Lami_2025}. A detailed discussion of these notions can be found in~\cite{Mazzola_2026}. More recently, a systematic study of different notions of approximate i.i.d.\ structure and the relations between them was initiated in~\cite{almost_iid}.

An important theme underlying these developments is that the expression ``almost i.i.d.'' does not refer to a single well-defined concept. Rather, there are several inequivalent ways in which a correlated quantum source may approximately resemble an i.i.d.\ one, each capturing a different operational or physical intuition. In some approaches, only a small number of constituents are allowed to behave anomalously; in others, one requires merely that sufficiently small marginals approximately coincide with those of an i.i.d.\ source.

This distinction is operationally significant because different information-theoretic tasks probe correlations at different scales. Protocols based on local measurements may be largely insensitive to sparse long-range correlations, whereas tasks involving entropy concentration, data compression, channel coding, or hypothesis testing can depend delicately on global structure. Consequently, two notions of approximate i.i.d.\ structure that appear nearly indistinguishable from a local perspective may nevertheless lead to markedly different asymptotic behaviour. In particular, a source may reproduce the local statistics of an i.i.d.\ source while failing dramatically to preserve its collective entropic properties (see Section~\ref{sec:intro_states}).\smallskip

These observations naturally lead to a central question of the present line of research: which notions of approximate i.i.d.\ structure are sufficiently robust to preserve the operational predictions of the i.i.d.\ setting? Equivalently, under which classes of perturbations do fundamental information-theoretic quantities — such as data compression rates, channel capacities, and hypothesis-testing exponents — remain stable?

In this work, we address this question in a systematic way, while relaxing the assumption that the underlying sources or channels are perfectly known. Rather than designing protocols adapted to specific perturbations of the i.i.d.\ regime, we seek universal mechanisms that remain effective across broad classes of almost i.i.d.\ resources.
\smallskip

Let us consider a concrete example in the setting of classical communication theory. Take the noiseless channel $\pazocal{I}$ on a bit $\{0,1\}$. It would be possible to achieve its capacity $C(\pazocal{I})=1$ by leveraging this naive communication protocol at the $n$-copies level: the encoder $\pazocal{E}_n: \{0,1\}^n\to \{0,1\}^n$ and the decoder $\pazocal{D}_n: \{0,1\}^n\to \{0,1\}^n$ for $\pazocal{I}_2^{\times n}$ could simply be chosen to be the identity maps, yielding a zero-error code with rate 1. Suppose that, for a pointwise defect on your communication line, your actual channel is not $\pazocal{I}_2^{\times n}$ but $\tilde{\pazocal{I}}_2^{(n)} \coloneqq\pazocal{F}\times\pazocal{I}_2^{\times n-1}$, where $\pazocal{F}$ is the (deterministic) bit flip, i.e.\ $\pazocal{F}(y|x)=1$ if and only if $y=x\oplus 1$. This perturbed channel could legitimately be called \emph{almost i.i.d.}, at least in the asymptotic limit $n\to \infty$, as it just differs from the i.i.d.\ one in a single site over $n$. The previous code immediately becomes useless, as the error probability immediately turns to 1. Of course, if we had known that the defect was taking place in the first use of the channel, we could have designed a code ignoring the first use, again restoring the same asymptotic performance. Mathematically speaking, the capacity of the perturbed channel is defined as an optimisation over all codes, \emph{which are perfectly aware} of the nature of the defect. But, in a practical setting, the exact nature of the local defect is likely to be unknown to the parties involved in the communication protocol. The only information available is that they are communicating on an unknown channel $\tilde{\pazocal{I}}_2^{(n)}$ which behaves almost like $n$ i.i.d.\ copies of $\pazocal{I}_2$. The quest for robustness not only involves the following question
\begin{center}
    \emph{Could optimal asymptotic rates of information theoretic protocols remain invariant, or even increase, when the original i.i.d.\ setting is replaced by any almost i.i.d.\ one?}
\end{center}
but also this dichotomy between theoretical optimum and lack of knowledge of the actual deviation from the i.i.d.\ structure:
\begin{center}
    \emph{Are there protocols which can still achieve the optimal asymptotic rates when the actual source has small unknown deviations from the ideal i.i.d.\ source?}
\end{center}

The remainder of the paper is organised as follows. In Section~\ref{sec:notation} we recall the notation and the basic notions we are going to use throughout the paper; in Section~\ref{sec:intro_protocol}, we give a concise recap of hypothesis testing, data compression and communication in the i.i.d.\ setting; in Section~\ref{sec:intro_W} we provide a self-contained introduction on the quantum Wasserstein distance of order 1, \errata{which will play a central role in the definition of almost i.i.d.\ states and processes.} In Section~\ref{sec:almost} we review the different almost i.i.d.\ structures for states, and we introduce the notion of almost i.i.d.\ process when considering sequences of channels. In Section~\ref{sec:hp} we study the robustness of asymmetric quantum hypothesis testing. Then, in Section~\ref{sec:data_comp} we address the problem of compression of classical and quantum almost i.i.d.\ sources. Finally, in Section~\ref{sec:channels} \errata{we discuss the robustness of channel coding when transmitting classical messages via classical and quantum channels}.\smallskip

For the convenience of the reader, we include a brief recap of the results of this work in Table~\ref{table}.

\begin{table}[t]
    \centering 
    \begin{talltblr}[
     caption = {Overview of the main results of this paper.},
     label = table,
    ]{
        width = \textwidth, 
        colspec = {X[1.7,c] X[0.9, c] X[c] X[0.9, c] X[1.2, c]}, 
    }
        \toprule
        \SetCell[c=5]{c}{\textbf{Almost i.i.d.\ states}} & & & & \\
        \toprule
        & $\rho_n\xrightharpoonup{w}\rho$ & $\rho_n\xrightharpoonup{W_1}\rho$ & $\rho_n\xrightharpoonup{\rm MSR}\rho$ &  \\
        \midrule
        \SetCell[c=5]{c}{\textit{Asymmetric hypothesis testing: robustness of the Stein exponent}} & & & & \\
        \midrule
        null hypothesis    & \ding{51} & \ding{51} & \ding{51} & Theorem~\ref{Stein_aiid_thm} \\
        alternative hypothesis & \ding{55} & \ding{55} & \ding{51} & and Remark~\ref{rem:trace} \\
        \midrule
        \SetCell[c=5]{c}{\textit{Data compression: robustness of the compression rate}} & & & & \\   
        \midrule
        classical source   & \ding{51} & \ding{51} & \ding{51} & Theorem~\ref{thm:c-data-comp}  \\
        quantum source  & \ding{51} & \ding{51} & \ding{51} & Theorem~\ref{thm:q-data-comp}\\
        \toprule
        \SetCell[c=5]{c}{\textbf{Almost i.i.d.\ channels}} & & & & \\
        \toprule
         && $\pazocal{N}_n\xrightharpoonup{\clubsuit}\pazocal{N}$ && \\
        \midrule
        \SetCell[c=5]{c}{\textit{Channel coding: robustness of classical information transmission}} & & & & \\ 
        \midrule
        capacity &  & \ding{51} &  & Theorem~\ref{thm:channels} \\
        reliability function  &  & \ding{55} &  & Section~\ref{sec:reliability} \\
        \midrule
    \end{talltblr}
\end{table}

\subsection{Notation and preliminaries}\label{sec:notation}

Let $\XX$ be a finite alphabet. \errata{The set of probability distributions on $\XX$ is denoted by 
$\mathcal{P}(\XX)$.} For a given positive integer $n$, we will denote as $\XX^n$ the set of strings $x^n=(x_1,\ldots,x_n)$ of length $n$ composed of symbols from $\XX$. An \deff{$n$-type} on $\XX$ is a probability distribution $t:\XX\to [0,1]$ on $\XX$ such that $nt(x)$ is an integer for all $x\in \XX$. The symbol $\mathcal{T}_n$ stands for the set of $n$-types on $\XX$. 

The type associated with a string $x^n\in \XX^n$ is defined as $t_{x^n}(x) \coloneqq \frac{N(x|x^n)}{n}$, where $N(x|x^n)$ is the number of times the symbol $x\in \XX$ appears in $x^n$. For a given $t\in \mathcal{T}_n$, we write 
\bb
T_{n,t} \coloneqq \left\{ x^n\in \XX^n:\ t_{x^n}=t \right\}
\ee
for the associated \deff{type class}, i.e.\ the set of strings $x^n$ with type $t$. Let $p=\{p(x)\}_{x\in\mathcal X}$ and
$q=\{q(x)\}_{x\in\mathcal X}$ be probability distributions on a finite alphabet $\mathcal X$. The \deff{total variation distance} between $p$ and $q$ is defined as
$\|p- q\|_1 \coloneqq \frac12 \sum_{x\in\mathcal X} |p(x) - q(x)|$.

A quantum system is modelled by a Hilbert space $\HH$; all Hilbert spaces in this work are assumed to be finite-dimensional. The set of linear operators acting on $\HH$ is denoted by $\mathcal{L}(\mathcal{H})$. A state of the system is represented by a \deff{density operator} (or density matrix), i.e.\ a positive semi-definite operator on $\HH$ with unit trace. The set of density operators on $\HH$ is denoted by $\D(\HH)$. \deff{Pure states} are simply rank-one projectors $\psi = \ketbra{\psi}$, where $\ket{\psi}\in \HH$ has unit norm, i.e.\ $\braket{\psi|\psi} = \|\ket{\psi}\|^2 = 1$.

The \deff{trace distance} between any two states $\rho, \sigma\in \D(\HH)$ is given by $\frac{1}{2}\|\rho - \sigma\|_1$, where $\|A\|_1 \coloneqq \Tr\big[|A|\big] = \Tr\big[\sqrt{A^\dagger A}\big]$. The corresponding (Umegaki) \deff{relative entropy} is instead defined as~\cite{Umegaki1962}
\bb
D(\rho\|\sigma) \coloneqq \Tr \big[ \rho (\log \rho - \log \sigma) \big]
\label{Umegaki}
\ee
provided that $\supp(\rho) \subseteq \supp(\sigma)$ (and, in this case, the logarithms are calculated on the support only); if $\supp(\rho) \not\subseteq \supp(\sigma)$, we set instead $D(\rho\|\sigma) = +\infty$.

Multipartite quantum systems are modelled by the tensor products of the local Hilbert spaces. Given some positive integer $n$ and a Hilbert space $\HH$, the Hilbert space corresponding to $n$ copies of the quantum system modelled by $\HH$ is therefore $\HH^{\otimes n}$. A permutation $\pi\in S_n$ of the underlying quantum systems, where $S_n$ stands for the symmetric group, is naturally represented by a unitary $U_\pi$ on $\HH^{\otimes n}$. A state $\rho_n\in \D\big(\HH^{\otimes n}\big)$ is said to be \deff{permutationally symmetric} if 
\bb
U_\pi^{\vphantom{\dag}} \rho_n U_\pi^\dag = \rho_n \qquad \forall\ \pi\in S_n\, .
\ee

A quantum channel ${\pazocal E}: \D(\HH) \to \D({\mathcal K})$, where $\HH$ and ${\mathcal K}$ are finite-dimensional Hilbert spaces,  is a linear completely positive trace-preserving (CPTP) map. We denote the identity map by ${\rm id}: \D(\HH) \to \D(\HH)$, and the identity operator on $\HH$ by $\id$. The \deff{diamond distance} between two channels $\Lambda_1, \Lambda_2: \D(\HH) \to \D({\mathcal K})$ is given by 
\bb\label{eq:diamond}
\|\Lambda_1 -\Lambda_1\|_{\diamondsuit} \coloneqq \sup_{\nu \in \D({\mathcal H}_A \otimes {\mathcal H}_R) }\|({\rm Id}_{R} \otimes \Lambda_1)(\nu) - ({\rm Id}_{R} \otimes \Lambda_2)(\nu)\|_1,
\ee 
where the supremum can be restricted to a reference system $R$ whose Hilbert space is isomorphic to that of $A$.

For a state $\rho\in\mathcal D(\mathcal H)$, and a quantum channel 
$\Lambda:{\mathcal D}(\mathcal H)\to{\mathcal D}(\mathcal K)$, the \emph{entanglement fidelity} of $\Lambda$ with respect to $\rho$ is defined as
\bb\label{eq:ent-fid}
F_e(\rho,\Lambda)
\coloneqq
\bra{\psi_\rho}
({\rm Id}\otimes\Lambda)
\bigl(
\ket{\psi_\rho}\!\bra{\psi_\rho}
\bigr)
\ket{\psi_\rho},
\ee
where $\ket{\psi_\rho}\in\mathcal H_R\otimes\mathcal H$ is any purification of $\rho$, i.e.
$
\Tr_R
\bigl(
\ket{\psi_\rho}\!\bra{\psi_\rho}
\bigr)
=
\rho .
$
The quantity $F_e(\rho,\Lambda)$ is independent of the choice of purification. For a CPTP map $\Lambda$ with Kraus operators $\{A_k\}_k$ we have (see e.g.~\cite[Eq. (9.135)]{NC})
\bb\label{eq:NC}
F_e(\rho, \Lambda) = \sum_k |\Tr(\rho A_k)|^2,
\ee

\subsection{Three protocols in quantum Shannon theory}\label{sec:intro_protocol}

\subsubsection{Hypothesis testing}\label{sec:intro_hp}

One of the key primitives of quantum information theory is quantum hypothesis testing. Given $n$ identical copies of a quantum system, each of which may be in a state $\rho$ (null hypothesis) or $\sigma$ (alternative hypothesis), the goal is to design a test, modelled by a binary POVM $\{E,\id-E\}$, which can distinguish between these two options reliably. To capture this concept quantitatively, one needs to distinguish between the two types of error one can make in the process: a type I error, which consists in guessing the alternative hypothesis when the null hypothesis holds, and a type II error, which, vice versa,  consists in guessing the null hypothesis when the alternative hypothesis holds. For the simplest case $n=1$ and a given test $0\leq E\leq \id$, which corresponds to guessing $\rho$, the error probabilities take the form
\bb
\pr\{\text{type I error}\,|\, E\} = \Tr[\rho (\id-E)]\, ,\qquad \pr\{\text{type II error}\,|\, E\} = \Tr[\sigma E]\, .
\ee
Given a threshold $\e\in [0,1)$ on the former, the minimum of the latter quantity is captured by the \deff{hypothesis testing relative entropy}~\cite{Buscemi2010}
\bb
D_H^\e(\rho\|\sigma) \coloneqq -\log \min\left\{ \Tr \sigma E:\ 0\leq E\leq \id\, ,\ \Tr \rho E \geq 1-\e\right\} .
\ee

If $n$ copies of the unknown state are available, we expect the error probabilities to decay exponentially fast as a function of $n$. In the setting of asymmetric hypothesis testing, which is deeply connected with both classical~\cite{Feinstein1954, Blackwell1959, Verdu1994, PPV} and quantum~\cite{Ogawa2002, Ogawa2007, Cheng2023Nov} coding theory, the goal is to minimise the type II error probability for a fixed type I error probability. The corresponding asymptotic figure of merit is the \deff{Stein exponent}, given by
\bb
\stein(\rho\|\sigma) \coloneqq \lim_{\e\to 0^+} \liminf_{n\to\infty} \frac1n\, \rel{D_H^\e}{\rho^{\otimes n}}{\sigma^{\otimes n}}\, .
\ee
A foundational result of quantum information theory is the proof, due to Hiai and Petz~\cite{Hiai1991} and refined by Ogawa and Nagaoka~\cite{Ogawa2000} that the above quantity is exactly equal to the Umegaki relative entropy~\eqref{Umegaki}: formally, the (strong converse) quantum Stein's lemma states that
\bb
\lim_{n\to\infty} \frac1n\, \rel{D_H^\e}{\rho^{\otimes n}}{\sigma^{\otimes n}} = D(\rho\|\sigma) \qquad \forall\ \e\in (0,1)\, ,
\ee
ensuring in particular that $\stein(\rho\|\sigma) = D(\rho\|\sigma)$.

\subsubsection{Classical and quantum data compression for i.i.d.~sources}\label{sec:intro_data}
Data compression is one of the central problems of information theory. The modern theory of data compression was initiated by Claude Shannon in his seminal paper~\cite{Shannon}. We briefly recall the notions of data compression (or source coding) for classical and quantum i.i.d.~sources below.

A classical memoryless (or i.i.d.)\ source $X \sim P$ is modelled by a sequence of independent and identically distributed random variables \(X_1,X_2,\ldots\), each taking values in a finite alphabet \(\mathcal X\) according to a fixed probability distribution \(\{P(x)\}_{x \in {\mathcal X}}\). In $n$ uses, the source outputs sequences
\bb
x^{n}=(x_1,\ldots,x_n)\in\mathcal X^n
\quad
\text{\,with probability\,} \quad
P^{\times n}(x^{n})
=
\prod_{i=1}^n P(x_i).
\ee

A fixed-length source code of rate \(R\) consists of the following encoding and decoding maps:
\begin{align}
{\pazocal E}_n:
\mathcal X^n
&\to
\{0,1\}^{\lceil nR\rceil}
\quad ; \quad
{\pazocal D}_n:
\{0,1\}^{\lceil nR\rceil}
\to
\mathcal X^n.
\end{align}

An \((n,R,\varepsilon)\) source code is a fixed-length source code
\(
{\pazocal C}_n
=
({\pazocal E}_n,{\pazocal D}_n,R)
\)
with blocklength \(n\), rate \(R\), and average probability of error satisfying
\bb
    p_{\rm err}(\pazocal{C}_n,P^{\times n})\leq \epsilon,
\ee
where
\bb
 p_{\rm err}(\pazocal{C}_n,P^{\times n})
\coloneqq
\PP{X^n\sim P^{\times n}}\left[
{\pazocal D}_n({\pazocal E}_n(X^{n}))
\neq
X^{n}\right]
\ee

A sequence of source codes \(({\mathcal C}_n)_n\) is said to be \emph{reliable} if
$
\lim_{n\to\infty}
p_{\mathrm{err}}^{(n)}({\mathcal C}_n)
=
0.$
A rate \(R\) is said to be \emph{achievable} for the source $X \sim p$ if there exists a reliable sequence of source codes of rate \(R\). Equivalently, \(R\) is achievable for $X \sim P$ if there exists a $(n, R+ \delta, \varepsilon)$ source code for every \(\varepsilon>0\), $\delta >0$, and sufficiently large $n$. The data compression limit for $X \sim p$ is then given by the optimal achievable rate:
$$R^*(P) \coloneqq \inf \{R \,: \, R \text{\, is achievable for \,} X^{\times n} \sim P^{\times n}\}.$$
Shannon's source coding theorem~\cite{Shannon} established that 
\bb
R^*(P) = H(P),
\ee
where $H(P)
=
-\sum_{x\in\mathcal X} P(x)\log P(x),$
is the Shannon entropy of the source.\bigskip

Quantum data compression is the quantum analogue of classical source coding. In the finite-dimensional setting, a quantum information source with a finite-dimensional Hilbert space \(\mathcal H\) is specified by the state of a quantum system i.e.\ its density matrix $\rho \in {\mathcal D}({\mathcal H})$. Such a source is said to be memoryless (or i.i.d.) if $\rho^{\otimes n}\in\mathcal D(\mathcal H^{\otimes n})$ characterises $n$ uses of it.

A quantum source code of rate \(R\) consists of completely positive trace-preserving (CPTP) encoding and decoding maps:
\bb
{\pazocal E}_n:
\mathcal D(\mathcal H^{\otimes n})
\to
\mathcal D({\mathcal K}_n),
\qquad
{\pazocal{D}}_n:
\mathcal D({\mathcal K}_n)
\to
\mathcal D(\mathcal H^{\otimes n}),
\ee
        where the compressed Hilbert space 
        \({\mathcal K}_n \subseteq {\mathcal H}^{\otimes n}\) 
        satisfies
\bb
\dim {\mathcal K}_n
\le
2^{\lceil nR\rceil}.
\ee

Since the state \(\rho^{\otimes n}\) may be part of a larger pure state entangled with a reference system, the appropriate figure of merit for the compression scheme is the entanglement fidelity of the overall compression--decompression map.

An \((n,R,\varepsilon)\) quantum source code is given by the triple
\({\mathcal C}_n
=
({\mathcal E}_n,{\mathcal D}_n,R)
\)
such that
\[
F_e\!\left(
\rho^{\otimes n},
{\pazocal D}_n\circ{\pazocal E}_n
\right)
\ge
1-\varepsilon,
\]
where \(F_e(\omega,\Lambda)\) denotes the entanglement fidelity of the state \(\omega\) under the quantum channel (i.e.\ CPTP map)  \(\Lambda\) (defined in~\eqref{eq:ent-fid}).
A sequence of quantum source codes \(({\pazocal C}_n)_n\) is said to be \emph{reliable} if
\[
\lim_{n\to\infty}
F_e\!\left(
\rho^{\otimes n},
{\pazocal D}_n\circ{\pazocal E}_n
\right)
=
1.
\]

A rate \(R\ge0\) is said to be \emph{achievable} for the source \(\rho\) if there exists a reliable sequence of quantum source codes of rate \(R\), and the quantum data compression limit for \(\rho\) is defined as
\bb
R^*(\rho) \coloneqq \inf \{ R:\, R \text{ is achievable for } \rho \}\, .
\ee

Schumacher's quantum source coding theorem~\cite{Schumacher} states that
\bb
R^*(\rho)
=
S(\rho),
\ee
where
\bb
S(\rho)
=
-\Tr[\rho\log\rho]
\ee
is the von Neumann entropy of the source.

\subsubsection{Reliable communication of classical information}\label{sec:intro_reliable}

Let $n\geq 1$ and let $\pazocal{N}^{(n)}_{A^n\to B^n}$ be a quantum channel mapping the states of $\pazocal{H}_A^{\otimes n}$ into states of $\pazocal{H}_B^{\otimes n}$. For $M\geq 2$, a code $\pazocal{C}_n=(\pazocal{E}_n,\pazocal{D}_n)$ of size $M$ is given by an encoder
\bb
    \pazocal{E}_n:[M]\to \mathcal{D}(\mathcal{H}_A^{\otimes n})
\ee
and a decoder
\bb\label{eq:decoder_def}
    \pazocal{D}_n: \mathcal{D}(\mathcal{H}_B^{\otimes n})\to \errata{\mathcal{P}([M])},
\ee
When using the code $\pazocal{C}_n$ to communicate the set of messages $[M]$ via the channel $\pazocal{N}^{(n)}$, the (average) error probability is
\bb\label{eq:p_err}
    p_{\rm err}\big(\pazocal{C}_n,\pazocal{N}^{(n)}\big)\coloneqq\frac{1}{M}\sum_{m=1}^{ M}\mathbb{P}\big(m=m'\sim\pazocal{D}_n\circ\pazocal{N}^{(n)}\circ\pazocal{E}_n(m)\big),
\ee

The case that is commonly considered \errata{when computing~\eqref{eq:p_err}} is the i.i.d.\ one, i.e.\ $\pazocal{N}^{(n)}=\pazocal{N}^{\otimes n}$ for a fixed channel $\pazocal{N}_{A\to B}$. We can then define
\bb\label{eq:c_epsilon}
    C_\epsilon(\pazocal{N})\coloneqq\sup\Big\{r\geq 0 :&&\!\! \limsup_{n\to\infty}\inf_{\substack{\pazocal{C}_n \text{ code}\\ \text{of size } \lceil 2^{rn}\rceil}}p_{\rm err}\big(\pazocal{C}_n,\pazocal{N}^{\otimes n}\big)\leq \epsilon\Big\}
\ee
The classical capacity $C(\pazocal{N})$ of $\pazocal{N}$ is defined as
\bb\label{eq:c}
    C(\pazocal{N})\coloneqq \lim_{\epsilon\to 0}C_\epsilon(\pazocal{N}).
\ee
For all rates below the capacity of $\pazocal{N}$, namely $0\leq r\leq C(\pazocal{N)}$, we can quantify the exponential decay of the error probability as $n\to \infty$ by defining the reliability function -- i.e.\ the error exponent -- as
\bb\label{eq:reliability}
    E(r,\pazocal{N})\coloneqq\liminf_{n\to\infty}-\frac{1}{n}\log \inf_{\substack{\pazocal{C}_n \text{ code}\\ \text{of size } \lceil 2^{rn}\rceil}}p_{\rm err}\big(\pazocal{C}_n,\pazocal{N}^{\otimes n}\big).
\ee
Let us conclude by a brief overview of the main results regarding classical communication via classical and quantum channels. The capacity of a classical channel $\pazocal{W}$ was introduced and identified by Shannon~\cite{Shannon} as the maximal mutual information between random input and output of $\pazocal{W}$, namely
\bb
    C(\pazocal{W})=\max_{P_X\in\mathcal{P}(\mathcal{X})}I(X:Y)\qquad X\sim P_X\qquad Y\sim \pazocal{W}_{Y|X}P_X.
\ee
The quantum generalisation is given by the Holevo--Schumacher--Westmoreland theorem~\cite{H-Schumacher-Westmoreland, Holevo-S-W}: for a quantum channel $\pazocal{N}_{A\to B}$, we have
\bb
    C(\pazocal{N})=\chi^{\infty}(\pazocal{N})\coloneqq\lim_{n\to\infty}\frac 1n\chi(\pazocal{N}^{\otimes n}),
\ee
with $\chi(\pazocal{N})$ being the Holevo quantity of $\pazocal{N}$, defined as
\bb
\chi(\pazocal{N}) \coloneqq \sup_{\rho_{XA}} I(X:B)_{{\rho}'}, 
\ee
where the supremum is over classical-quantum states $\rho_{XA} = \sum_x P_X(x) \ketbra{x} \otimes \rho_x^A$, and $\rho_{XB}' \coloneqq ({\rm Id}_X\otimes\pazocal{N}_{A\to B})(\rho_{XA})$.

There is no closed form expression for the reliability function of arbitrary channels at all rates $0<r<C(\pazocal{W})$. However, some important upper and lower bounds have been identified. It is important to mention at least Gallager's random coding bound~\cite{Gallager1965, gallager1968information}, which provides an achievable error exponent and turns out to be tight for rates larger than a critical value (called \emph{critical rate}) due to a matching converse bound, known as \emph{sphere-packing bound}, by Shannon--Gallager--Berlekamp~\cite{SHANNON196765,SHANNON1967522}. In Section~\ref{sec:channels}, we will leverage the idea of random coding (see e.g.~\cite[Chapter~5]{gallager1968information}) and we will consider a modified version of the maximal likelihood decoder in order to identify a valid code for almost i.i.d.\ sequences of channels.

\subsection{Transportation distances and the quantum Wasserstein distance of order 1}\label{sec:intro_W}

Standard distinguishability measures in quantum information theory -- such as trace distance, fidelity, and relative entropy -- are unitarily invariant and therefore regard any pair of orthogonal states as maximally distinguishable. While natural in many settings, this feature becomes problematic when one aims to quantify robustness under local perturbations. In multipartite systems, for example, one would intuitively expect the product state $|0\rangle^{\otimes n}$ to be substantially closer to $|1\rangle\otimes |0\rangle^{\otimes (n-1)}$ than to $|1\rangle^{\otimes n}$, since the former differs from the reference state only on a single subsystem. More generally, one seeks notions of distance compatible with the geometry induced by the Hamming metric on product spaces and stable under local modifications. Such considerations are particularly important in almost i.i.d. quantum information theory, where the relevant states depart from an exact tensor-product structure only on a small subset of subsystems.

These issues are closely connected with the continuity properties of entropic quantities. The von Neumann entropy is intrinsically robust under local perturbations: modifying a single qubit can alter the entropy by at most a constant independent of the total number of subsystems. However, this robustness is not reflected by unitarily invariant distances, since even a local operation may transform a state into an orthogonal one, thereby yielding maximal distinguishability. This discrepancy motivates the introduction of alternative metrics that explicitly incorporate locality.

In the classical setting, Wasserstein distances from optimal transport theory provide a natural framework for this purpose~\cite{villani2008optimal}. Given a metric space $(\mathcal X,d)$, the Wasserstein distance of order $1$ -- also known as the $W_1$ distance, Monge–Kantorovich distance~\cite{monge1781memoire,kantorovich2006translocation}, or earth mover’s distance -- between two probability distributions $\mu$ and $\nu$ on $\mathcal X$ quantifies the minimal average transportation cost required to transform one distribution into the other, where transporting a unit mass from $x\in\mathcal X$ to $y\in\mathcal X$ incurs cost $d(x,y)$.
More formally, we have
\begin{equation}
    W_1(\mu,\nu) = \min_\pi\,\EE{(X,Y)\sim\pi}\,d(X,Y)\,,
\end{equation}
where the minimization is performed over all the couplings $\pi$ between $\mu$ and $\nu$, given by the probability distributions on $\mathcal{X}\times\mathcal{X}$ with marginals $\mu$ and $\nu$.

When $\mathcal X$ is the set of strings over a finite alphabet, the natural underlying metric is the Hamming distance, which counts the number of coordinates on which two strings differ. The corresponding Wasserstein distance, known as Ornstein’s $\bar d$-distance~\cite{ornstein1973application}, has found numerous applications in ergodic theory and information theory, particularly in settings involving weak dependence, coding with memory, and rate-distortion theory~\cite{gray2011entropy}. Its central feature is its sensitivity to local perturbations: probability distributions concentrated on strings differing only in a few coordinates remain close in $W_1$ distance even when they are perfectly distinguishable in total variation distance.

These ideas have motivated the development of quantum analogues of the $W_1$ distance capable of capturing the locality structure of multipartite quantum systems. Such metrics provide a natural framework for quantifying approximate tensor-product structures and for analyzing continuity phenomena in almost i.i.d. quantum information theory.

The quantum Wasserstein distance of order $1$, or quantum $W_1$ distance, was introduced in~\cite{De_Palma_2021,De_Palma_2024b} as a generalization of the Hamming distance from classical strings to states of multipartite quantum systems. The construction is based on the notion of neighboring states. Two states $\rho$ and $\sigma$ of the $n$-partite system $A_1\ldots A_n$ are said to be \emph{neighboring} if there exists a subsystem $A_i$ such that
\begin{equation}
\Tr_{A_i}\rho=\Tr_{A_i}\sigma\,,
\end{equation}
that is, the two states become identical upon discarding subsystem $A_i$. The quantum $W_1$ norm is then defined as the norm whose unit ball is the convex hull of differences between neighboring states, and the associated quantum $W_1$ distance is the metric induced by this norm.

More precisely, we have the following definition.

\begin{Def}[(Quantum $W_1$ distance)]
For any $\rho,\,\sigma\in \mathcal{D}(\mathcal{H}_A^{\otimes n})$, we define
\bb
    \|\rho-\sigma\|_{W_1} \coloneqq \min \Bigg\{ \;\sum_{i=1}^n c_i\quad  &\text{such that}\quad  && c_i\geq 0, \qquad  \rho-\sigma=\sum_{i=1}^n c_i\left(\tau^{(i)}-\eta^{(i)}\right) \\
    & \text{with}&& \tau^{(i)},\eta^{(i)}\in \mathcal{D}(\mathcal{H}^{\otimes n}_A), \quad \Tr_{A_i}\tau^{(i)}=\Tr_{A_i}\eta^{(i)}&\Bigg\}.
\ee
\end{Def}

The quantum $W_1$ distance admits a dual characterization in terms of a quantum analogue of the Lipschitz constant for observables.

\begin{Def}[(Quantum Lipschitz constant~\cite{De_Palma_2021})]
Let $X$ be an observable on the $n$-partite system $A_1\ldots A_n$. For each $i\in[n]$, define the dependence of $X$ on subsystem $A_i$ as
\begin{equation}
    \partial_i X = 2\min\left\{\left\|X - Y_{A_i^c}\right\| : Y_{A_i^c}\textnormal{ does not act on }A_i\right\}\,.
\end{equation}
The quantum Lipschitz constant of $X$ is then
\begin{equation}
    \|X\|_L = \max_{i\in[n]}\partial_i X\,.
\end{equation}
\end{Def}

The $W_1$ distance between two states can therefore be expressed as the maximal difference in expectation values over observables with unit Lipschitz constant.

\begin{prop}[{\cite{De_Palma_2021}}]
Let $\rho$ and $\sigma$ be states of the $n$-partite system $A_1\ldots A_n$. Then
\begin{equation}
        \left\|\rho-\sigma\right\|_{W_1} = \max_{\|X\|_L=1}\mathrm{Tr}\left[\left(\rho-\sigma\right)X\right]\,.
\end{equation}
\end{prop}

The relevance of the quantum $W_1$ distance stems from its ability to detect when two states differ only on a small fraction of subsystems. Indeed, while a perturbation on a single subsystem may already render two states orthogonal -- and therefore maximally distinguishable according to any unitarily invariant distance such as trace distance, fidelity, or relative entropy -- the quantum $W_1$ distance remains small whenever the discrepancy is localised. Owing to this feature, the quantum $W_1$ distance has found applications across several areas of quantum information theory, including quantum machine learning through quantum generative adversarial networks~\cite{Kiani_2022}, limitations of variational quantum algorithms~\cite{De_Palma_2023}, quantum differential privacy~\cite{hirche2023quantumdifferentialprivacyinformation}, equivalence of ensembles in quantum statistical mechanics~\cite{De_Palma_2022,De_Palma_2025}, learning of many-body quantum states~\cite{De_Palma_2024,Rouz__2024,Rouz__2024b}, and rapid thermalization for geometrically local Hamiltonians~\cite{Bardet_2024,Kochanowski_2025,bakshi2025dobrushinconditionquantummarkov}.

The quantum $W_1$ distance always lies between the trace distance and $n$ times the trace distance~\cite[Proposition~2]{De_Palma_2021}:
\bb\label{eq:trW}
    \frac 12 \|\rho-\sigma\|_1\leq \|\rho-\sigma\|_{W_1}\leq \frac{n}{2}\|\rho-\sigma\|_1\,,
\ee
\errata{and cannot increase under the action of a quantum channel acting on a single subsystem~\cite[Proposition 3]{De_Palma_2021}.}
One of the most relevant properties of the quantum $W_1$ distance is that the von Neumann entropy per subsystem is uniformly continuous with respect to the normalised $W_1$ distance.

\begin{thm}[{(Continuity of the von Neumann entropy~\cite[Theorem 9.1]{De_Palma_2023b})}]\label{thm:SW}
For any $n\in\mathbb{N}$ and any $\rho,\,\sigma\in\mathcal{D}\left(\mathcal{H}^{\otimes n}\right)$ we have
\begin{equation}
    \frac{1}{n}\left|S(\rho) - S(\sigma)\right| \le h_2\left(w\right) + w\ln\left(\left(\dim\mathcal{H}\right)^2-1\right)\,\qquad \text{where}\qquad w\coloneqq \frac1n \left\|\rho-\sigma\right\|_{W_1}
\end{equation}
and, {\errata{for any $x \in [0,1]$, }} $h_2(x) := -x\ln x - \left(1-x\right)\ln\left(1-x\right)$ is the binary entropy function.
\end{thm}

Furthermore, for states that are diagonal in the canonical basis, the quantum $W_1$ distance recovers the classical $W_1$ distance.

\begin{prop}[{\cite[Proposition 6]{De_Palma_2021}}]\label{prop:cl}
Let $p$ and $q$ be probability distributions on $[d]^n$, and let
\begin{equation}\label{eq:defpq}
\rho = \sum_{x\in[d]^n}p(x)\,|x\rangle\langle x|\,,\qquad \sigma = \sum_{y\in[d]^n}q(y)\,|y\rangle\langle y|\,.
\end{equation}
Then,
\begin{equation}\label{eq:CQ}
\left\|\rho-\sigma\right\|_{W_1} = W_1(p,q)\,.
\end{equation}
\end{prop}

\section{Almost i.i.d.\ states and channels}\label{sec:almost}

In this section we are going to review three notions of almost i.i.d.-ness, and we will introduce the family of almost i.i.d.\ channels (according to the quantum Wasserstein distance of order 1), which are metrised by a novel channel distance.

\subsection{Three families of almost i.i.d.\ states}\label{sec:intro_states}
\begin{figure}[t]
  \centering
  \def\svgwidth{0.8\linewidth}
  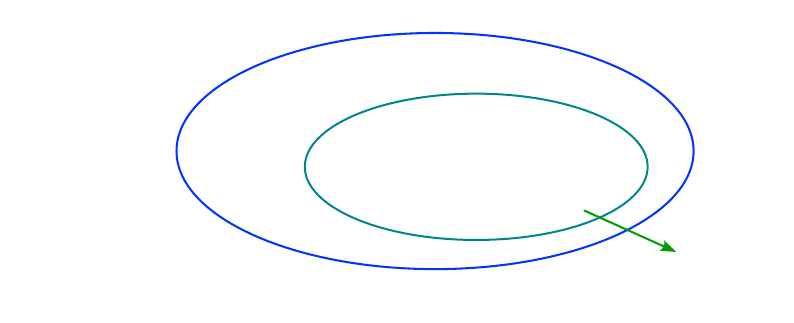
  \mycaption{Three different notions of almost i.i.d.\ processes.}{A pictorial representation of the hierarchical relation between the different notions of almost i.i.d.-ness. As discussed in~\cite{almost_iid}, the Mazzola--Sutter--Renner notion is the strictest, the one derived from Wasserstein distance is the intermediate one, and, 
  finally weak almost i.i.d.-ness is the broadest. The classical probability source $(\tilde p_n)_n$ and the pure-state source $\Psi_n$ are presented in~\cite[Section~3.3.1]{almost_iid}, while the source $(\xi_n)_n$ is studied in~\cite[Section~3.3.2]{almost_iid}. In this paper, we provide a complementary discussion on the notion of trace distance almost i.i.d.-ness, arguing that it is physically less interesting than the other three, and we discuss three new examples of sources, indicated in what follows by $(\tilde\sigma_n)_n$, $(\eta_n)_n$ and $(\gamma_n)_n$.} 
  \label{fig:almost}
\end{figure}

What does it mean, operationally, to claim that a quantum source is i.i.d.? The standard approach to answer this questions would be to assume that the $n$-copy state $\rho_n$ it produces is equal to a tensor product $\rho^{\otimes n}$. However, it is arguably impossible to guarantee that this is the case for arbitrarily large $n$. First of all, because we assume that $n$ is the total sample size, and we do not have access to \emph{multiple copies} of $\rho_n$, as one would need to carry out tomography on it. Secondly, we do not want to measure the state of the entire source, otherwise we would destroy its quantum properties. What we can guarantee, at best, is that for an arbitrarily large $n$ the relevant information about $\rho$ is preserved and can be efficiently tested or employed for compression or communication. There is no standard figure of merit to quantify how much information of the ideally i.i.d.\ source is present in the actual state once we take into account the action of the environmental noise. The naive guess that a source can robustly be treated as i.i.d.\ when the trace distance with respect to the ideal case asymptotically vanishes, i.e.\ $\|\rho_n-\rho^{\otimes n}\|_1\to 0$ as $n\to \infty$, not only, in general, is a too stringent and non realistic requirement (see Section~\ref{sec:intro_W}), but it also turns out that in some tasks trace distance does not even capture the properties of the source that are needed to ensure robustness (see Remark~\ref{rem:trace}). 
For this reason, in order to precisely identify the conditions for robustness in the protocols introduced in Section~\ref{sec:intro_protocol}, we are going to review different notions of almost i.i.d.-ness recently introduced in~\cite{Mazzola_2026, almost_iid}.

The weakest notion of almost i.i.d.\ source $(\rho_n)_n$ has the only guarantee that, for an arbitrarily large but fixed $k\in \N^+$, in the limit where $n\to\infty$ the $k$-body marginals of $\rho_n$, i.e.\ the reduced states of $\rho_n$ on the sub-systems of size $k$, will be close to $\rho^{\otimes k}$, at least \emph{on average}.
\begin{Def}[{(Weakly almost-i.i.d.\ source~\cite{almost_iid})}]
    Let $\rho\in\mathcal{D}(\mathcal{H})$ be a quantum state on a Hilbert space $\mathcal{H}$. A sequence $(\rho_n)_n$ of states $\rho_n\in\mathcal{D}(\mathcal{H}^{\otimes n})$ is said to be a \emph{weakly almost-i.i.d.\ source along $\rho$} if  
    \bb
        \lim_{n\to \infty} \EE{\substack{I\subseteq [n],\\ |I|=k}} \big\|(\rho_n)_I - \rho^{\otimes I}\big\|_1 = 0\qquad \forall k\geq 1
    \ee
    where the random variable $I$ has uniform distribution over the subsets of $[n]$ of size $k$. If this is the case, then we write
    \bb
        \rho_n\xrightharpoonup{w}\rho,
    \ee
    and we denote by $\pazocal{A}_\rho^w$ the set of all weakly almost i.i.d.\ sources along $\rho$:
    \bb
        \pazocal{A}_\rho^w &\coloneqq\{(\rho_n)_n: (\rho_n)_n \xrightharpoonup{w} \rho\}.
    \ee
\end{Def}

Under this notion, not even the normalised von Neumann entropy of an almost i.i.d.~source $\rho_n\xrightharpoonup{w}\rho$ is guaranteed to converge to the entropy of $\rho$~\cite[Section 3.3.1]{almost_iid}. However, as we already discuss in Section~\ref{sec:hp}, such a weak notion is sufficient to ensure some relevant robustness properties of quantum hypothesis testing.\bigskip

A stronger definition -- which, incidentally, ensures the asymptotic continuity of the von Neumann entropy -- naturally stems from and it is motivated by the discussion of Section~\ref{sec:intro_W}.

\begin{Def}[{(Wasserstein almost-i.i.d.\ source~\cite{almost_iid})}]
    Let $\rho\in\mathcal{D}(\mathcal{H})$ be a quantum state on a Hilbert space $\mathcal{H}$. A sequence $(\rho_n)_n$ of states $\rho_n\in\mathcal{D}(\mathcal{H}^{\otimes n})$ is called a \emph{Wassertein almost-i.i.d.\ source along $\rho$} if  
    \bb
        \lim_{n\to\infty}\frac 1 n \|\rho_n-\rho^{\otimes n}\|_{W_1}=0.
    \ee
    If this is the case, then we write
    \bb
        \rho_n\xrightharpoonup{W_1}\rho,
    \ee
    and we denote by $\pazocal{A}_\rho^{W_1}$ the set of all Wasserstein almost i.i.d.\ sources along $\rho$:
    \bb
        \pazocal{A}_\rho^{W_1} &\coloneqq\{(\rho_n)_n: (\rho_n)_n \xrightharpoonup{W_1} \rho\}.
    \ee
\end{Def}

Eventually, the strictest definition of almost i.i.d.-ness is the one inspired by the celebrated exponential quantum de Finetti theorem~\cite{RennerPhD, Renner2007} and formally introduced for the first time in~\cite{Mazzola_2026}. Here, we present the formulation in which the permutational invariance of the systems is relaxed.

\begin{Def}[{(Mazzola--Sutter--Renner (MSR) almost i.i.d.\ states~\cite{Mazzola_2026})}]\label{def:DeFinetti_almost_iid} 
Let $\rho \in \mathcal{D}(\mathcal{H}_A)$, and $n,r \in \mathbb{N}$ such that $r \leq n$. Let 
\bb
\pazocal{V}^n_r(\mathcal{H}_{AE}, \ket{\psi}):=\{U_\pi(\ket{\psi}^{\otimes n-r} \otimes \ket{\omega^{(r)}}): \pi \in S_n, \ket{\omega^{(r)}} \in \mathcal{H}_{AE}^{\otimes r} \}.
\ee
Then, $\rho_n \in \mathcal{D}(\mathcal{H}_A^{\otimes n})$ is called a $\binom{n}{r}$-\emph{almost i.i.d.\ state in} $\rho $ if there exists a purification $\ket{\psi_\rho}_{AE}$ of $\rho$ and an extension $\rho_n^{A^n E^n}$ of $\rho_n^{A^n}$ such that $\supp(\rho_n^{A^n E^n}) \subseteq  \mathrm{span}\,\pazocal{V}^n_r(\mathcal{H}_{AE},\ket{\psi_\rho}_{AE})$.
 Furthermore, we say that a sequence $(\rho_n)_n$ of states $\rho_n\in\mathcal{D}(\mathcal{H}^{\otimes n})$ is a \emph{MSR almost i.i.d.\ source along $\rho$} if, for $r_n=o(n)$, the state $\rho_n$ is $\binom{n}{r_n}$-almost i.i.d., and we write
 \bb
        \rho_n\xrightharpoonup{\rm MRS}\rho,
    \ee
    We denote by $\pazocal{A}_\rho^{W_1}$ the set of all MRS almost i.i.d.\ sources along $\rho$:
    \bb
        \pazocal{A}_\rho^{\rm MSR} &\coloneqq\{(\rho_n)_n: (\rho_n)_n \xrightharpoonup{\rm MSR} \rho\}.
    \ee
\end{Def}
The inclusion relation among these classes is represented in Figure~\ref{fig:almost}. We have also represented the family of \emph{trace distance almost i.i.d.\ sources}, which can naturally be defined as follows. 

\begin{Def}[(Trace distance almost-i.i.d.\ source)]
    Let $\rho\in\mathcal{D}(\mathcal{H})$ be a quantum state on a Hilbert space $\mathcal{H}$. A sequence $(\rho_n)_n$ of states $\rho_n\in\mathcal{D}(\mathcal{H}^{\otimes n})$ is called a \emph{trace distance almost-i.i.d.\ source along $\rho$} if  
    \bb
        \lim_{n\to\infty} \|\rho_n-\rho^{\otimes n}\|_1=0.
    \ee
\end{Def}

From~\eqref{eq:trW}, we immediately see that trace distance almost i.i.d.\ sources are included in the family of Wasserstein almost i.i.d.\ sources, but the inclusion is strict. We conclude this section by considering the following sequences of states, which are also represented in Figure~\ref{fig:almost}. 
\begin{itemize}
    \item Let $\gamma\in(\mathbb{C})^{\otimes 2}$ be the state $\gamma=\ketbra 0$. Consider the sequence $\gamma_n=\ketbra 1\otimes\ketbra 0 ^{\otimes n-1}$. Being $\rho_n$ orthogonal to $\gamma^{\otimes n}$, it cannot be a trace distance almost i.i.d.\ source along $\gamma$. However, it is not difficult to see that $\gamma_n\xrightharpoonup{W_1}\gamma$, and, more in particular, $\gamma_n\xrightharpoonup{\rm MSR}\gamma$. This source exemplifies the fragility of the trace distance request: even a local defect may hinder the convergence.
    \item A simple change in the previous example yields a source being both MSR almost i.i.d.\ and trace distance almost i.i.d.: let
\bb
    \eta_n=\left(\frac 1n \ketbra 1 + \frac {n-1}n\ketbra 0\right)\otimes\ketbra 0 ^{\otimes n-1}\qquad \eta = \ketbra 0.
\ee
For this source, it is easy to prove that $\|\eta_n-\eta^{\otimes n}\|_1\to 0$ as $n\to\infty$, and $\eta_n\xrightharpoonup{\rm MSR}\eta$.
\item Finally, the source
\bb
    \tilde \sigma_n = \left(1-\frac 1n\right)\sigma^{\otimes n}+\frac{1}{n}\rho^{\otimes n}
\ee
is trace distance almost i.i.d.\ along $\sigma$, but it is not MSR almost i.i.d.\ along $\sigma$ whenever $\rho\neq \sigma$. This second claim will be proved in Remark~\ref{rem:trace}.
\end{itemize}

\subsection{Almost i.i.d.\ processes}

The most frequently used distance measure between quantum channels is the diamond distance, defined in~\eqref{eq:diamond}.
Similarly to quantum states, when two quantum channels act on several subsystems, it may be natural to consider them as close if they differ only on a small fraction of the subsystems.
Since it is based on the trace distance, the diamond distance does not capture this property.
It seems then natural to define a notion of distance between quantum channels that is based on the quantum $W_1$ distance between their outputs.

\begin{Def}\label{def:club_norm} Let $\mathcal{H}_A$ be an arbitrary Hilbert space and, for $n\geq 1$, let $\mathcal{H}_{B^n}$ an $n$-partite Hilbert space. Given a linear map $\Delta\Phi_{A\to B^n}:\mathcal{L}(\mathcal{H}_A)\to\mathcal{L}(\mathcal{H}_{B^n})$ that is trace annihilating (i.e. $\mathrm{Tr}_{B^n}\circ\Delta\Phi=0$), we define its \emph{club norm} as
\bb
    \big\|\Delta\Phi\big\|_{\,\clubsuit}\coloneqq \sup_{\rho_{A}}\big\|\Delta\Phi(\rho_{A})\big\|_{W_1}.
\ee
\end{Def}

We mention that a different definition, 
involving the $W_1$ distance between the Choi states of the channels, was proposed in~\cite{Duvenhage_2023}. 

\errata{In this work we do not consider the \emph{completely bounded} extension of the club norm, i.e.\ its stabilised version, as we are interested in studying classical communication without entanglement assistance. For this purpose, Definition \ref{def:club_norm} is sufficient.}

\begin{prop}
    The club norm is a norm, and for any trace-annihilating map $\Delta\Phi:A\to B^n$ we have
    \begin{equation}
        \left\|\Delta\Phi\right\|_\clubsuit \le n\left\|\Delta\Phi\right\|_\diamond\,.
    \end{equation}
\end{prop}
\begin{proof}
    Nonnegativity and homogeneity with respect to scaling are obvious from the definition.
    The triangle inequality holds since
    \begin{align}
        \left\|\Delta\Phi_1 + \Delta\Phi_2\right\|_\clubsuit &= \sup_\rho\left\|\Delta\Phi_1(\rho) + \Delta\Phi_2(\rho)\right\|_{W_1} \le \sup_\rho\left\|\Delta\Phi_1(\rho)\right\|_{W_1} + \sup_\rho\left\|\Delta\Phi_2(\rho)\right\|_{W_1}\nonumber\\
        &= \left\|\Delta\Phi_1\right\|_\clubsuit + \left\|\Delta\Phi_2\right\|_\clubsuit\,.
    \end{align}
    Finally, $\big\|\Delta\Phi\big\|_{\,\clubsuit} = 0$ if and only if $\Delta\Phi(\rho_{A})=0$ for any input state $\rho_A$, which happens if and only if $\Delta\Phi=0$.
    For the second part, we have
    \begin{equation}
        \left\|\Delta\Phi\right\|_\clubsuit = \sup_\rho\left\|\Delta\Phi(\rho)\right\|_{W_1} \le \frac{n}{2}\sup_\rho\left\|\Delta\Phi(\rho)\right\|_1 \le n\left\|\Delta\Phi\right\|_\diamond\,,
    \end{equation}
    where we have used~\cite[Proposition 2]{De_Palma_2021}.
\end{proof}

\begin{figure}[t]
  \centering
  \def\svgwidth{\linewidth}
  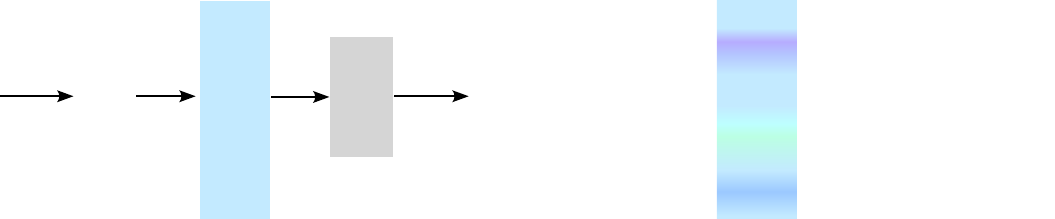
  \mycaption{From $n$ i.i.d.\ uses of a channel to an almost i.i.d.\ process.}{A common assumption when evaluating the transmission capabilities -- namely, error probability and communication rate -- of multiple uses of a given channel $\pazocal{N}$ is the possibility to access exact i.i.d.\ instances of the channel. As a consequence, coding theorems and strategies are intended for parallel uses of channels satisfying such idealised behaviour. However, in a realistic scenario, the action of the channel might fluctuate over time, and memory effects might induce correlations. For this reason, we introduce the notion of almost i.i.d.\ process, replacing the tensor product structure $\pazocal{N}^{\otimes n}$ with a more general channel $\tilde{\pazocal{N}}^{(n)}$ asymptotically satisfying the condition stated in Definition~\ref{def:channel}.} 
  \label{fig:channel}
\end{figure}

In Section~\ref{sec:intro_states} we have argued that it is not realistic to {\errata{consider}}
\errata{quantum sources} to be exactly i.i.d.
\errata{Similar considerations apply to quantum channels: the transformations they describe may be affected by small temporal fluctuations, and therefore need not be exactly memoryless.} This situation is schematically depicted in Figure~\ref{fig:channel} in the particular case of channel coding. This motivates the following definition.

\begin{Def}\label{def:channel} Let $\Lambda:\mathcal{D}(\mathcal{H}_A)\to \mathcal{D}(\mathcal{H}_B)$ be a quantum channel and, for every $n\geq 1$, let $\tilde \Lambda^{(n)}$ be a channel from $A^n$ to $B^n$.  The sequence of channels $\tilde \Lambda=\big(\tilde \Lambda^{(n)}\big)_{n\geq 1}$ is an \emph{almost i.i.d.\ process} along $\Lambda$ if
\bb
    \lim_{n\to\infty}\frac{1}{n}\big\|\tilde\Lambda^{(n)}-\Lambda^{\otimes n}\big\|_{\,\clubsuit}=0
\ee
Furthermore, if $(\Lambda_n)_n$ is an almost i.i.d.\ process along $\Lambda$, we write
\bb
    (\Lambda_n)_n \xrightharpoonup{\clubsuit} \Lambda.
\ee
\end{Def}

\subsection{Alternative quantum Wasserstein distances}
\errata{

Several inequivalent quantum analogues of classical Wasserstein distances have been proposed in the literature. In this work, we employ the quantum $W_1$ distance introduced in~\cite{De_Palma_2021}, which is particularly well suited to the locality structure relevant for almost i.i.d.\ systems. We briefly review some alternative approaches below.

One prominent line of work, developed by Carlen, Maas, Datta, and Rouz\'e~\cite{carlen2014analog,carlen2017gradient,carlen2020non,rouze2019concentration,datta2020relating,van2020geometrical}, defines a quantum Wasserstein distance of order~2 through a Riemannian metric on the space of quantum states, constructed from a quantum analogue of differential calculus. This framework establishes deep connections between Wasserstein geometry, entropy, and Fisher information~\cite{datta2020relating}, and has been used to characterise convergence rates for the quantum Ornstein--Uhlenbeck semigroup~\cite{carlen2017gradient,de2018conditional}. Building on the same quantum differential structure, Refs.~\cite{rouze2019concentration,carlen2020non,gao2020fisher} introduce quantum analogues of Lipschitz constants and of the corresponding \(W_1\)-type, or earth mover's, distance. Other approaches based on quantum differential structures appear in Refs.~\cite{chen2017matricial,ryu2018vector,chen2018matrix,chen2018wasserstein}, while Refs.~\cite{agredo2013wasserstein,agredo2016exponential,ikeda2020foundation} define quantum earth mover's distances using distances between vectors of the canonical basis.

A different line of work, initiated by Golse, Mouhot, Paul, and Caglioti~\cite{golse2016mean,caglioti2018towards,golse2018quantum,golse2017schrodinger,golse2018wave,caglioti2019quantum,Friedland_2022,bunth2026strongkantorovichdualityquantum}, emerged from the study of semiclassical limits in quantum mechanics and introduces quantum Wasserstein distances of order~2 based on quantum couplings. In this framework, a coupling between two quantum states \(\rho\) and \(\sigma\) on \(\mathcal H\) is a quantum state \(\Pi\) on \(\mathcal H^{\otimes 2}\) whose reduced states on the first and second subsystems are \(\rho\) and \(\sigma\), respectively. The transport cost associated with \(\Pi\) is defined as the expectation value, in the state \(\Pi\), of a positive semidefinite cost operator \(C\). Different choices of \(C\) lead to different distances, obtained by minimizing the transport cost over all admissible couplings.

The construction of Refs.~\cite{golse2016mean,caglioti2018towards,golse2018quantum,golse2017schrodinger,golse2018wave,caglioti2019quantum} has the feature that the distance between a quantum state and itself may be nonzero. Ref.~\cite{chakrabarti2019quantum} observed that this pathology disappears when the support of the cost operator lies entirely in the antisymmetric subspace of \(\mathcal H^{\otimes 2}\) under subsystem exchange. Motivated by this observation, Ref.~\cite{chakrabarti2019quantum} takes the orthogonal projector onto the antisymmetric subspace as the cost operator and uses the resulting distance as a loss function for quantum generative adversarial networks.

Ref.~\cite{de2019quantum} introduces another coupling-based quantum Wasserstein distance of order~2, in which each quantum coupling naturally determines a quantum channel. This proposal also has a nonzero self-distance. This issue is addressed in the modified construction of Ref.~\cite{De_Palma_2024b}, where the self-distance is subtracted. The modified distance has been proved to satisfy the triangle inequality~\cite{bunth2026wassersteindistancesdivergencesorder,Bunth_2024,wirth2025triangleinequalityquantumwasserstein}. Connections between quantum couplings and quantum channels in the setting of von Neumann algebras have also been investigated in Ref.~\cite{duvenhage2018balance}. More generally, quantum earth mover's distances based on quantum couplings have been studied in Ref.~\cite{agredo2017quantum}.

Finally, another approach defines the distance between two quantum states as the classical Wasserstein distance between the probability distributions generated by an informationally complete measurement. This perspective has been explored for Gaussian quantum systems using heterodyne measurements in Refs.~\cite{zyczkowski1998monge,zyczkowski2001monge,bengtsson2017geometry}.

Comprehensive reviews of quantum optimal transport and quantum Wasserstein distances can be found in Refs.~\cite{maas2024optimal,beatty2025wassersteindistancesquantumstructures}.

In this paper we employ the quantum \(W_1\) distance proposed in Ref.~\cite{De_Palma_2021}, since it is particularly well-suited to the locality structure of multipartite systems and recovers the classical \(W_1\) distance on probability distributions over bit strings for states diagonal in the computational basis.
}

\section{Quantum hypothesis testing with almost i.i.d.\ states}\label{sec:hp}

\begin{figure}[t]
  \centering
  \def\svgwidth{0.7\linewidth}
  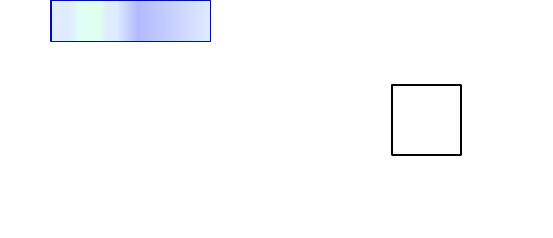
    \mycaption{Almost i.i.d.\ hypothesis testing.}{The operational task of quantum hypothesis testing with almost i.i.d.\ states can be illustrate in two instances. Suppose either one or both the i.i.d.\ hypotheses $H_0$ and $H_1$ are replaced with almost i.i.d.\ hypotheses $\tilde H_0$ and $\tilde H_1$. The first question is the following: suppose we are aware of the nature of the defects for the individual sources $\rho_n$ and $\sigma_n$; is it then possible to design a test which exploits such defects to enhance the distinguishability of the hypotheses, yielding a larger Stein exponent? The second instance, crucial for practical implementations, concerns the case in which the exact nature of the defects is not known: the only guarantee on the states that are tested is that they are almost i.i.d.\ sources along $\rho$ or $\sigma$, respectively. Is then possible to design a universal test which achieves the same performance of the i.i.d.\ setting for all the sources $\rho_n\xrightharpoonup{}\rho$ and $\sigma_n\xrightharpoonup{}\sigma$.}
  \label{fig:hp_test}
\end{figure}

What does robustness mean in asymmetric hypothesis testing? Imagine having $n$ copies of a state which obeys one of two hypotheses, say $\rho$ as a null hypothesis and $\sigma$ as an alternative hypothesis. As mentioned in Section~\ref{sec:intro_hp}, by the Hiai and Petz's quantum Stein's lemma~\cite{Hiai1991, Ogawa2000}, it is possible to design a sequence of tests $\{E_n\}_n$ such that the error of type II exponentially decays with exponent given by the Umegaki relative entropy $D(\rho\|\sigma)$, while the type I error remains bounded. In any realistic scenario, it is not possible to exclude that the ideally i.i.d.\ state to be tested -- i.e. either $\rho^{\otimes n}$ or $\sigma^{\otimes n}$ -- actually has an almost i.i.d.\ structure (see Figure~\ref{fig:hp_test}). And, more precisely, it is not possible to determine in advance which exact perturbation of the i.i.d.\ state we have to deal with. Hence, our problem can be phrased as follows:
\begin{itemize}
    \item Which notions of almost i.i.d.\ ($\rho_n \xrightharpoonup{?_0}\rho$ and $\sigma_n \xrightharpoonup{?_1}\sigma$), if any, ensure that tailor-made optimal protocols can achieve the (at least) the same performance of the i.i.d.\ setting when distinguishing the two hypotheses? Namely, is it possible to prove that for suitable almost i.i.d.\ perturbations we get
    \bb
    \liminf_{n\to\infty} \frac1n\, D_H^\e(\rho_n\|\sigma_n) \geqt{?} D(\rho\|\sigma) \qquad \forall\ \e\in (0,1)\,?
    \ee
    \item If the previous question can be answered affirmatively, are there universal protocols $\{E_n\}$, i.e.\ tests that asymptotically work for all almost i.i.d.\ sources according to such notions $\rho_n \xrightharpoonup{?_0}\rho$ and $\sigma_n \xrightharpoonup{?_1}\sigma$? Formally speaking, is it possible to prove that
    \begin{align}\nonumber
        {\rm Stein}_\epsilon\big(\pazocal{A}_\rho^{?_0}\|\pazocal{A}_\sigma^{?_1}\big)&\coloneqq \sup_{\substack{(E_n)_n\\ 0\leq E_n\leq \id_n}} \inf_{\substack{\rho_n \xrightharpoonup{?_0}\rho\\ \sigma_n \xrightharpoonup{?_1}\sigma}}\Bigg\{\errata{\liminf_{k\to\infty}}-\frac 1k\log\Tr[E_k\sigma_k]
        :\errata{\liminf_{k\to\infty}}\Tr[E_k\rho_k]\geq 1-\epsilon\Bigg\}\\
        &\geqt{?}D(\rho\|\sigma)\label{eq:stein_almost}
    \end{align}
    Note that, in the above setting, we first choose the sequence of tests, and then we evaluate the asymptotic performance of such tests with respect to almost i.i.d.\ perturbations of $\rho$ and $\sigma$.
\end{itemize}

The answer to both question is affirmative, provided that we consider the right notions of almost i.i.d.\ sources, as we prove in the following statement.

\begin{boxedthm}{}
\begin{thm}[(Quantum Stein's lemma for almost-i.i.d.\ sources)] \label{Stein_aiid_thm}
Let $\rho,\sigma\in \D(\HH)$ be two states on a finite-dimensional Hilbert space $\HH$. Let us consider
\begin{itemize}
    \item a be weakly almost-i.i.d.\ source $(\rho_n)_n$ along $\rho$,
    \item an MSR almost i.i.d.\ source $(\sigma_n)_n$ along $\sigma$.
\end{itemize} 
Then,
\bb\label{Stein_aiid_ineq_simple}
\liminf_{n\to\infty} \frac1n\, D_H^\e(\rho_n\|\sigma_n) \geq D(\rho\|\sigma) \qquad \forall\ \e\in (0,1)\, .
\ee
More generally, the relative entropy is the Stein exponent that can be achieved with a universal protocol depending only on the pair $(\rho,\sigma)$, i.e.
\bb\label{eq:A}
        {\rm Stein}_\epsilon\big(\pazocal{A}_\rho^w\|\pazocal{A}_\sigma^R\big)= D(\rho\|\sigma)\qquad \forall\ \e\in (0,1). 
\ee

\end{thm}
\end{boxedthm}
\begin{rem}
    Note that, in order to prove~\eqref{eq:A}, it is sufficient to show that
    \bb
        {\rm Stein}_\epsilon\big(\pazocal{A}_\rho^w\|\pazocal{A}_\sigma^R\big)\geq D(\rho\|\sigma)\qquad \forall\ \e\in (0,1), 
    \ee
    since the converse inequality immediately follows from the particular choice $\rho_n=\rho^{\otimes n}$ and $\sigma_n=\sigma^{\otimes n}$ in the minimisation appearing in~\eqref{eq:stein_almost}:
    \bb
        {\rm Stein}_\epsilon\big(\pazocal{A}_\rho^w\|\pazocal{A}_\sigma^R\big)\leq {\rm Stein}_\epsilon\big(\rho\|\sigma\big)=D(\rho\|\sigma).
    \ee
\end{rem}
\begin{proof}[Proof of Theorem~\ref{Stein_aiid_thm}]
The proof is based on three steps:
\begin{enumerate}
    \item first, we prove that~\eqref{Stein_aiid_ineq_simple} 
     holds in the classical setting, with a weakly almost i.i.d.\ null hypothesis $p_n$ and an i.i.d.\  alternative hypothesis $q^{\times n}$ (see Section~\ref{sec:almost_1});
    \item then, we lift the previous result to quantum states by a measurement argument (see Section~\ref{sec:almost_2});
    \item eventually, we conclude by introducing an MSR almost i.i.d.\ alternative hypothesis $\sigma_n$, so that we establish that~\eqref{Stein_aiid_ineq_simple} holds in full generality (see Section~\ref{sec:almost_3}).
\end{enumerate}
In every step, we show that the test can be chosen in a universal way. Such a strategy, which is developed in detail in the following sections for greater clarity, completes the proof of Theorem~\ref{Stein_aiid_thm}.
\end{proof}

\begin{rem}
The inequality in~\eqref{Stein_aiid_ineq_simple} can definitely be strict, even in the fully classical case. This can happen in several different ways:
\begin{itemize}
\item Consider for example the case where $\XX$ is a classical alphabet of size $|\XX|\geq 2$, $p=q\in \mathcal{P}(\XX)$ is some probability distribution such that $p(x_0) = 0$ for some $x_0\in \XX$. Take $q_n=q^{\times n}$ and 
\bb
p_n = \left(\left(1-\frac{1}{\sqrt{n}}\right) p + \frac{1}{\sqrt{n}}\, \delta_{x_0}\right)^{\otimes n},
\ee
with $\delta_{x_0}$ denoting the deterministic probability distribution concentrated on $x_0$, is a weakly almost-i.i.d.\ source along $p$. In this case, guessing the null hypothesis if and only if the symbol $x_0$ appears in the sequence achieves type-2 error probability exactly equal to zero and type-1 error probability that is asymptotically vanishing. Thus, $D_H^\e(p_n\|q^{\otimes n})=+\infty$ for all $\e\in (0,1)$, provided that $n$ is sufficiently large, while obviously $D(p\|q)=0$.

\item A different set of examples can be obtained with the construction described in~\cite[Section~3.3.1]{almost_iid}. More precisely, by considering sources $p_n$ that are not permutationally symmetric, one can better discriminate them from i.i.d.\ ones. For example, let $p_n$ be the probability distribution generated as follows. The symbols $x_j$ with odd index $j\in [n]$ are drawn in an i.i.d.\ fashion according to some distribution $p\in \mathcal{P}(\XX)$ (say, with full support); the symbols with even $j$, instead, are set to be equal to the preceding odd-$j$ symbol. For some $k\in \N^+$, a random subset $I \subseteq [n]$ of cardinality $k$ will not contain any pair of consecutive indices $\{j,j+1\}$, with $j$ odd, with asymptotically unit probability $\frac{n-2}{n} \frac{n-4}{n} \ldots \frac{n-2k}{n}$. When no such pair is present, the probability distribution of the symbols in $I$ is exactly $p^{\otimes I}$. Therefore, 
\bb
\E_{I\subseteq [n],\ |I|=k} \left\| (p_n)_I - p^{\otimes I} \right\|_1 \leq 2 \left( 1 - \frac{n-2}{n} \frac{n-4}{n} \ldots \frac{n-2k}{n} \right) \tendsn{} 0\, .
\ee
However, it is not difficult to verify that
\bb
\lim_{n\to\infty} \frac1n\, D_H^\e\big(p_n\|p^{\otimes n}\big) = \frac12\, D \left(\sumno_x p(x) \ketbra{xx} \Big\| \left(\sumno_x p(x) \ketbra{x}\right)^{\otimes 2} \right) = \frac{H(p)}{2}\, .
\ee
Therefore, $p_n$ can be distinguished from $p^{\otimes n}$ with a non-zero Stein exponent, implying that the inequality in~\eqref{Stein_aiid_ineq_simple} can be strict even in the case where $p$ has full support.
\end{itemize}
\end{rem}

\begin{rem}\label{rem:trace}
We cannot, in general, substitute $\sigma^{\otimes n}$ with a Wasserstein almost-i.i.d.\ source along $\sigma$ -- and not even with a trace distance almost i.i.d.\ source. Indeed, if that were possible then we could take $\tilde \sigma_n = \left(1-\frac1{n}\right) \sigma^{\otimes n} + \frac1{n}\, \rho^{\otimes n}$, which is an almost i.i.d.\ source along $\tilde \sigma^{\otimes n}$, but then
\bb\label{eq:0}
\frac1n\, D_H^\e(\rho^{\otimes n}\|\tilde \sigma_n) \leqt{(a)} - \frac1n \log \frac{1-\e}{n} \tendsn{} 0\, ,
\ee
i.e.\ the Stein exponent would be zero in this case. The inequality (a) can be directly obtained as follows:
\bb
     D_H^\e(\rho^{\otimes n}\|\sigma_n)&=-\log\min_{0\leq E_n\leq \id}\big\{\Tr[E_n\tilde \sigma_n]:\Tr[E_n\rho^{\otimes n}]\geq 1-\epsilon\big\}\\
     &=-\log\min_{0\leq E\leq \id}\Big\{\underbrace{\Big(1-\tfrac{1}{n}\Big)\Tr[E_n\sigma^{\otimes n}]}_{\geq 0}+\underbrace{\tfrac{1}{n}\Tr[E_n\rho^{\otimes n}]}_{\geq \frac{1-\epsilon}{n}}:\Tr[E_n\rho^{\otimes n}]\geq 1-\epsilon\Big\}.
\ee
This counterexample is particularly instructive, because it clearly shows that the trace distance, despite its restrictiveness, fails to capture the properties of the source needed to ensure robustness when perturbing the alternative hypothesis. Furthermore, since for $\rho\neq \sigma$ the source $(\tilde \sigma_n)_n$ cannot satisfy 
\eqref{Stein_aiid_ineq_simple}, we conclude that it cannot be a MSR almost i.i.d.\ source along $\sigma$. 
\end{rem}

\subsection{Classical case with i.i.d.\ alternative hypothesis}\label{sec:almost_1}

In this section, we are going to prove a classical version of Theorem~\ref{Stein_aiid_thm} where the alternative hypothesis is assumed to be i.i.d. This statement will turn out to be the core element of the proof of the next section.

\begin{lemma} \label{classical_Stein_aiid_lemma}
Let $p,q\in \mathcal{P}(\XX)$ be two probability distributions on a finite alphabet $\XX$, and let $(p_n)_n$ be a weakly almost-i.i.d.\ (classical) source along $p$. Then
\bb
\liminf_{n\to\infty} \frac1n\, D_H^\e(p_n\|q^{\otimes n}) \geq D(p\|q) \qquad \forall\ \e\in (0,1)\, .
\label{classical_Stein_aiid}
\ee
More precisely, there exists a \emph{universal} sequence of measurements -- that just depends on $p$ and $q$, not on the particular weakly almost-i.i.d.\ source $p_n$ -- which achieves a type II error exponent at least as large as $D(\rho\|\sigma)$.
\end{lemma}

\begin{proof}
To find a lower bound on $D_H^\e(p_n\|q^{\otimes n})$, we can assume without loss of generality that $p_n$ be permutationally symmetric. Operationally, this means that we can shuffle the input sequence with a random permutation. By data processing, this can only decrease the hypothesis testing relative entropy. Furthermore, the symmetrisation of a weakly almost-i.i.d.\ source is again weakly almost-i.i.d.\ along the same state. Thus, we can write
\bb
p_n = \sum_{t_n\in\mathcal{T}_n} \alpha_{n,t_n} u_{t_n}\, ,
\label{classical_Stein_aiid_proof_eq1}
\ee
where $\mathcal{T}_n$ is the set of $n$-types on $\XX$, and $u_{t_n}$ denotes the uniform probability distribution on the type class $T_{n,t_n}$ (i.e.\ the set of sequences $x^n\in \XX^n$ with type $t_n$). For every fixed $k\in \N^+$ and every $s_k\in \mathcal{T}_k$, we then have that the $k$-symbol marginal of $u_{t_n}$, denoted with $(u_{t_n})_{[k]}$, is given by the hypergeometric distribution. 
Therefore,
\begin{align}
(u_{t_n})_{[k]}\big(T_{k,s_k}\big) &= \binom{n}{k}^{-1} \prod_x \binom{nt_n(x)}{ks_k(x)} \nonumber \\
&\leq \frac{k!}{(n-k)^k}  \prod_x \frac{(nt_n(x))^{ks_k(x)}}{(ks_k(x))!} \nonumber \\
&= \left(1-\frac{k}{n}\right)^{-k} \binom{k}{ks_k} \prod_x t_n(x)^{ks_k(x)} \\
&=\left(1-\frac{k}{n}\right)^{-k} t_n^{\otimes k}\big(T_{k,s_k}\big)\, , \nonumber
\end{align}
meaning that
\bb
D_{\max}\big((u_{t_n})_{[k]}\, \big\|\, t_n^{\otimes k} \big) = \max_{s_k\in\mathcal{T}_k} \log \frac{(u_{t_n})_{[k]}\big(T_{k,s_k}\big)}{t_n^{\otimes k}\big(T_{k,s_k}\big)} \leq - k \log\left(1-\frac{k}{n}\right) \eqqcolon \log \lambda_{k,n}\, ,
\label{classical_Stein_aiid_proof_eq3}
\ee
where $\lambda_{k,n} \coloneqq \left(1-\frac{k}{n}\right)^{-k}$.

Now, consider any set $\pazocal{A} \subseteq \mathcal{P}(\XX)$. Then for all $n\geq k$ we have that
\begin{align*}
p^{\otimes k} \!\left( \bigcup\nolimits_{s_k\in\mathcal{T}_k \cap \pazocal{A}} T_{k,s_k} \right) &\leq (p_n)_{[k]} \!\left( \bigcup\nolimits_{s_k\in\mathcal{T}_k \cap \pazocal{A}} T_{k,s_k} \right) + \frac12\left\| p^{\otimes k} - (p_n)_{[k]}\right\|_1 \\
&= \sum_{t_n\in\mathcal{T}_n} \alpha_{n,t_n} (u_{t_n})_{[k]}\left( \bigcup\nolimits_{s_k\in\mathcal{T}_k \cap \pazocal{A}} T_{k,s_k} \right) + \frac12\left\| p^{\otimes k} - (p_n)_{[k]}\right\|_1 \\
&\leqt{(a)} \lambda_{k,n} \sum_{t_n\in\mathcal{T}_n} \alpha_{n,t_n} t_n^{\otimes k} \left( \bigcup\nolimits_{s_k\in\mathcal{T}_k \cap \pazocal{A}} T_{k,s_k} \right) + \frac12\left\| p^{\otimes k} - (p_n)_{[k]}\right\|_1 \\
&\leqt{(b)} \lambda_{k,n} \sum_{t_n\in\mathcal{T}_n} \alpha_{n,t_n} \min\left\{1,\, (k+1)^{|\XX|-1}\, 2^{-k D(\pazocal{A}\|t_n)}\right\} + \frac12\left\| p^{\otimes k} - (p_n)_{[k]}\right\|_1\, ,
\end{align*}
Here, we denoted with $(p_n)_{[k]}$ the $k$-symbol marginal of the permutationally symmetric $p_n$, in~(a) we employed~\eqref{classical_Stein_aiid_proof_eq3}, and~(b) is the standard estimate due to Sanov~\cite[Exercise~2.12]{CSISZAR-KOERNER}. 

Taking the limit $n\to\infty$ yields
\bb
p^{\otimes k} \!\left( \bigcup\nolimits_{s_k\in\mathcal{T}_k \cap \pazocal{A}} T_{k,s_k} \right) \leq \liminf_{n\to\infty} \sum_{t_n\in\mathcal{T}_n} \alpha_{n,t_n} \min\left\{1,\, (k+1)^{|\XX|-1}\, 2^{-k D(\pazocal{A}\|t_n)}\right\} .
\ee
If $\pazocal{A} = B_\infty(p,\delta) \coloneqq \left\{ r\in \mathcal{P}(\XX):\ \|r-p\|_\infty \leq \delta \right\}$ 
for some $\delta>0$, 
we get
\begin{align}
&p^{\otimes k} \!\left( \bigcup\nolimits_{s_k\in\mathcal{T}_k \cap B_\infty(p,\delta)} T_{k,s_k} \right) \nonumber \\
&\qquad \leq \liminf_{n\to\infty}\sum_{t_n\in\mathcal{T}_n} \alpha_{n,t_n} \min\left\{1,\, (k+1)^{|\XX|-1}\, 2^{-k D(B_\infty(p,\delta) \| t_n)}\right\} \nonumber \\
&\qquad \leq \liminf_{n\to\infty} \sum_{t_n\in\mathcal{T}_n \cap B_\infty(p,2\delta)} \alpha_{n,t_n} + (k+1)^{|\XX|-1} \limsup_{n\to\infty} \sum_{t_n\in\mathcal{T}_n \setminus B_\infty(p,2\delta)} \alpha_{n,t_n} 2^{-k D(B_\infty(p,\delta) \|t_n)} \\
&\qquad \leqt{(c)} \liminf_{n\to\infty} \sum_{t_n\in\mathcal{T}_n \cap B_\infty(p,2\delta)} \alpha_{n,t_n} + (k+1)^{|\XX|-1} \limsup_{n\to\infty} \sum_{t_n\in\mathcal{T}_n \setminus B_\infty(p,2\delta)} \alpha_{n,t_n} e^{-2k \delta^2} \nonumber \\
&\qquad \leq \liminf_{n\to\infty} \sum_{t_n\in\mathcal{T}_n \cap B_\infty(p,2\delta)} \alpha_{n,t_n} + (k+1)^{|\XX|-1} e^{-2k \delta^2} , \nonumber 
\end{align}
where in~(c) we used Pinsker's inequality. Finally, taking the limit $k\to\infty$ and using the law of lawrge numbers on the leftmost side shows that
\bb
\lim_{n\to\infty} \sum_{t_n\in\mathcal{T}_n \cap B_\infty(p,2\delta)} \alpha_{n,t_n} = 1\, ;
\ee
recalling~\eqref{classical_Stein_aiid_proof_eq1}, this is equivalent to stating that the test that consists in checking whether the sequence $x^n$ has type in $B_\infty(p,2\delta)$ achieves an asymptotically vanishing type I error, i.e. 
\bb
\lim_{n\to\infty} p_n\!\left( \bigcup\nolimits_{t_n\in\mathcal{T}_n \cap B_\infty(p,2\delta)} T_{n,s_n} \right) = 1
\ee
for all $\delta>0$. Such a test depends only on the reference distribution $p$, not on the particular weakly almost-i.i.d.\ source, hence it is \emph{universal}. As is well known, again via Sanov's theorem~\cite[Exercise~2.13]{CSISZAR-KOERNER} the type II error probability of this test can be shown to satisfy
\bb
-\frac1n \log q^{\otimes n} \left( \bigcup\nolimits_{t_n\in\mathcal{T}_n \cap B_\infty(p,2\delta)} T_{n,s_n} \right) &\geq - \frac{|\XX|-1}{n} \log(n+1) + D\big( B_\infty(p,2\delta)\, \big\|\, q\big) \\
&\tendsn{} D\big( B_\infty(p,2\delta)\, \big\|\, q\big)\, ,
\ee
thus implying that
\bb
\liminf_{n\to\infty} \frac1n\, D_H^\e(p_n\|q^{\otimes n}) \geq D\big( B_\infty(p,2\delta)\, \big\|\, q\big) \qquad \forall\ \delta > 0\, .
\ee
We can now take the limit $\delta\to 0^+$. Using the lower semi-continuity of the relative entropy, it is not difficult to verify that $\lim_{\delta\to 0^+} D\big( B_\infty(p,2\delta)\, \big\|\, q\big) = D(p\|q)$, which completes the proof of~\eqref{classical_Stein_aiid}.
\end{proof}

\subsection{Quantum case with i.i.d.\ alternative hypothesis}\label{sec:almost_2}

The above classical solution can be lifted immediately to the quantum setting via measuring. Intuitively, this will be the strategy:
\begin{enumerate}
    \item symmetrise the $n$-partite state by applying a random permutation;
    \item for a fixed $h\leq n$, divide the $n$ parties in batches of $h$ copies, and discard the remaining systems, if needed;
    \item measure each batch with the measurement $\pazocal{M}_h$ achieving the measured relative entropy between $\rho^{\otimes h}$ and $\sigma^{\otimes h}$ 
    \item apply the strategy of Lemma~\ref{classical_Stein_aiid_lemma}, i.e.\ verify whether the sequence of the outcomes is strongly typical for $\pazocal{M}_h(\rho^{\otimes h})$.
\end{enumerate}
All these step do not depend on the specific weak almost i.i.d.\ source $\rho_n$, but only on the reference states $\rho$ and $\sigma$. Therefore, for each fixed $h$, the strategy above defined is \emph{universal}. This is in analogy with the result known as `quantum Sanov theorem'~\cite{Bjelakovic2005, Noetzel_2014}, which applies to the (different) setting in which the null hypothesis is guaranteed to be i.i.d.\ but with unknown base state, and states that also in that case the test can be chosen universally for all base states. See also~\cite{generalised-Sanov, doubly-comp-quantum} for a further strengthening of this statement.

With this approach in mind, we are ready to formally prove Theorem~\ref{Stein_aiid_thm} for i.i.d.\ alterative hypotheses, by showing that such a procedure yields a type II error exponent asymptotically equal to (or larger than) $D(\rho\|\sigma)$ as $h\to \infty$. \bigskip

Fix some $h\in \N^+$, and let $\MM_h:\D\big(\HH^{\otimes h}\big) \to \mathcal{P}(\XX)$ be a measurement on $h$ copies of the quantum system under examination. Our test consists in applying a random permutation to the $n$ systems at our disposal, partitioning them into $m \coloneqq \floor{n/h}$ batches of $h$ systems each, measuring each batch with $\MM_h$, and applying the classical test constructed in Lemma~\ref{classical_Stein_aiid_lemma}. To this end, define the classical probability distributions $p_m$ given by $p_m \coloneqq \MM_h^{\otimes m} \big( \Tr_{n - mh} \widebar{\rho}_n\big)$, where $\widebar{\rho}_n$ is the symmetrised version of the state $\rho_n$, and $\Tr_{n-mh}$ denotes the partial trace over the last $n-mh$ systems. Set also $p \coloneqq \MM_h\big(\rho^{\otimes h}\big)$ and $q \coloneqq \MM_h\big(\sigma^{\otimes h}\big)$. The problem now consists in deciding between the two hypotheses $p_m$ and $q^{\otimes m}$. 

Let $k\in \N^+$. Pick a uniformly random subset $I\subseteq [m]$ of size $|I| = k$, and construct the subset $J\subseteq [n]$ of size $kh$ obtained by joining the batches indicated by $I$. A little thought reveals that $J$ is then a uniformly random subset of size $kh$.  
Hence, due to data processing
\bb
\E_{I\subseteq [m],\ |I|=k} \left\| (p_m)_I - p^{\otimes I} \right\|_1 &\leq \E_{J\subseteq [n],\ |J|=kh} \left\| (\widebar{\rho}_n)_J - \rho^{\otimes J} \right\|_1 \\
&= \left\| \E_{J\subseteq [n],\ |J|=kh} (\rho_n)_J - \rho^{\otimes J} \right\|_1 \\
&\leq \E_{J\subseteq [n],\ |J|=kh} \left\| (\rho_n)_J - \rho^{\otimes J} \right\|_1 \\
&\tendsn{} 0\, ,
\ee
implying that the source $(p_m)_m$ is almost i.i.d.\ along $p$. Therefore, Lemma~\ref{classical_Stein_aiid_lemma} ensures that for all $\e\in (0,1)$ it holds that
\bb
\liminf_{n\to\infty} \frac1n\, D_H^\e(\rho_n\|\sigma^{\otimes n}) &\geq \liminf_{n\to\infty} \frac1n\, D_H^\e(p_m\|q^{\otimes m}) \\
&= \frac1h \liminf_{m\to\infty} \frac1m\, D_H^\e(p_m\|q^{\otimes m}) \\
&\geq \frac1h\, D(p\|q) \\
&= \frac1h\, D\Big(\MM_h\big(\rho^{\otimes h}\big)\Big\|\MM_h\big(\sigma^{\otimes h}\big)\Big)\, .
\ee
Taking the supremum over all measurements $\MM_h$ yields
\bb
\liminf_{n\to\infty} \frac1n\, D_H^\e(\rho_n\|\sigma^{\otimes n}) \geq \frac1h \rel{D^M}{\rho^{\otimes h}}{\sigma^{\otimes h}} \geq \frac1h \rel{D^M}{\rho^{\otimes h}}{\sigma^{\otimes h}} - \frac{d}{h} \log (h+1)\, ,
\ee
where $D^M(\rho\|\sigma) \coloneqq \sup_\MM \rel{D}{\MM(\rho)}{\MM(\sigma)}$ is the measured relative entropy~\cite{Donald1986, Berta2017}, and the last inequality has been established by Hiai and Petz~\cite{Hiai1991} (see also~\cite[Lemma~2.4]{berta_composite}). Now, evaluating the limit $h\to\infty$ 
proves~\eqref{Stein_aiid_ineq_simple}. 

Note that the type II error exponent $D(\rho\|\sigma)$ can be achieved with a \emph{universal} sequence of tests, as: (a)~the optimal measurement $\MM$ that achieves the above measured relative entropy only depends on $\rho$ and $\sigma$, and not on the particular almost i.i.d.\ source at hand; and (b)~the same is true of the classical decision rule applied to the resulting classical probability distributions, which only depends on $p = \MM\big(\rho^{\otimes h}\big)$.

\subsection{Quantum case in full generality}\label{sec:almost_3}

Leveraging the result of the previous section, we can eventually prove robustness when the alternative hypothesis is MSR almost i.i.d. We need to recall the following technical lemma.

\begin{lemma}[{\cite[Lemma~4.3]{Mazzola_2026}}]\label{lem:giulia}
    Let $\sigma_{A^n}$ be a $\binom{n}{r}$-almost i.i.d.\ state along $\sigma_A$. Then, there exist a purification $\ket{\sigma}_{AE}$ of $\sigma_{A}$ and an extension $\sigma_{A^nE^n}$ of $\sigma_{A^n}$ that can be written as
\begin{align}
\sigma_{A^nE^n} = \sum_{t,t' \in \pazocal{T}} \beta_{t,t'} \ket{\Psi_t} \bra{\Psi_{t'}}_{A^nE^n} \qquad \beta_{t,t'} \in \mathbb{C},
\end{align}
for a family $\{\ket{\Psi_t}_{A^nE^n}\}_{t \in \pazocal{T}}$ of orthonormal vectors from $\pazocal{V}^n_r(\mathcal{H}_{AE}, \ket{\sigma}_{AE})$, with
\begin{align}\label{eq:upper_T}
\log |\pazocal{T}| \leq n \, h\Big( \frac{r}{n} \Big) + r \log d_{AE} .
\end{align}
Furthermore, calling 
\begin{align}\label{eq:sig}
\tilde \sigma_{A^n E^n T} \coloneqq \sum_{t \in \pazocal{T}} \beta_{t,t} \ketbra{\Psi_t}_{A^n E^n} \otimes \ketbra{t}_T,
\end{align}
we have $\sigma_{A^n E^n} \leq |\pazocal{T}| \tilde \sigma_{A^n E^n}$, whence $\sigma_{A^n} \leq |\pazocal{T}| \tilde \sigma_{A^n}$.
\end{lemma}

For some $\delta\in (0,1)$ to be specified later, set $\sigma_A(\delta)\coloneqq (1-\delta) \sigma_A + \delta \frac{\id}{d}$, where $d \coloneqq \dim \HH_A$. Using the notation of Lemma~\ref{lem:giulia}, let $\ket{\Psi_t}_{A^nE^n}=U_{\pi_t}\left(\ket{\sigma}^{\otimes n-r}_{AE}\otimes \ket{\omega_t}_{A^rE^r}\right)$ for a suitable permutation $\pi_t \in S_n$. \errata{Note that, without loss of generality, we can assume $d_{AE}=d_A^2$; indeed, up to an isometry $V_{E\to E'}$, we can reduce $\ket{\sigma}_{AE}$ to its standard purification $V_{E\to E'}\ket{\sigma}_{AE}$, with $d_{E'}=d_A$; then, we can act with $V_{E\to E'}^{\otimes n}$ on $\sigma_{A^nE^n}$ and $\tilde\sigma_{A^nE^nT}$ without changing the claim of Lemma~\ref{lem:giulia}. Now,} calling $(\omega_t)_{A^r}=\Tr_{E^r}\ketbra{\omega_t}$, we can write 
\bb
\widetilde\sigma_{A^n}
 =
 \sum_{t\in \T}
 \beta_{t,t}\,
 U_{\pi_t}
 \left(
 \sigma_A^{\otimes (n-r)}
 \otimes
 (\omega_t)_{A^r}
 \right)
 U_{\pi_t}^\dagger ,
\ee
where \((\omega_t)_{A^r}\in\mathcal D(\mathcal H_A^{\otimes r})\). Since $(\omega_t)_{A^r}\le \id_{A^r}
\le
\errata{d_A^r}\delta^{-r} \sigma_A(\delta)^{\otimes r},
$
we obtain, for each \(t\in \T\),
\bb
U_{\pi_t}
\left(
\sigma_A^{\otimes (n-r)}
\otimes
(\omega_t)_{A^r}
\right)
U_{\pi_t}^\dagger
\le
\errata{d_A^r}(1-\delta)^{-n+r} \delta^{-r}
U_{\pi_t}
\left(
\sigma_A(\delta)^{\otimes (n-r)}
\otimes
\sigma_A(\delta)^{\otimes r}
\right)
U_{\pi_t}^\dagger .
\ee
However, since \(\sigma_A(\delta)^{\otimes n}\) is permutation invariant, 
$
U_{\pi_t}
\left(
\sigma_A(\delta)^{\otimes (n-r)}
\otimes
\sigma_A(\delta)^{\otimes r}
\right)
U_{\pi_t}^\dagger
=
\sigma_A(\delta)^{\otimes n}.
$
Hence,
\bb
\widetilde\sigma_{A^n}
\le
\errata{d_A^r}(1-\delta)^{-n+r}\delta^{-r}
\sum_{t\in {\mathcal{T}}}
\beta_{t,t}
\sigma_A(\delta)^{\otimes n}.
\ee
Applying the trace to both sides of the identity~\eqref{eq:sig}, defining
$\widetilde\sigma_{A^nE^nT}$, and using linearity and multiplicativity of the trace over tensor products, yields
$
1=\sum_{t\in \T}\beta_{t,t}.
$
Hence,
\bb
\widetilde\sigma_{A^n}
\le
\errata{d_A^r}(1-\delta)^{-n+r}\delta^{-r} \sigma_A(\delta)^{\otimes n},
\ee
and by Lemma~\ref{lem:giulia},
\bb
\sigma_{A^n}\le |{\T}|\,\widetilde\sigma_{A^n}
\le
\frac{|\T|\errata{d_A^r}}{(1-\delta)^{n-r}\delta^{r}}\,
\sigma_A(\delta)^{\otimes n}.
\ee
Thus, for every positive semidefinite operator \(0\le E_n\le \id_{A^n}\),
\bb
\Tr[E_n\sigma_{A^n}]
\le
\frac{|\T|\errata{d_A^r}}{(1-\delta)^{n-r}\delta^{r}}
\Tr[E_n\sigma_A(\delta)^{\otimes n}].
\ee
Equivalently,
\bb
-\frac1n\log \Tr[E_n\sigma_{A^n}]
&\ge
-\frac1n\log \Tr\big[E_n\sigma_A(\delta)^{\otimes n}\big]\\
&\qquad -
\frac1n\log |\T| \errata{-\frac rn \log d_A}
-
\frac rn \log \frac1{\delta} - \frac{n-r}{n}\log\frac{1}{1-\delta} .
\ee
Using the upper bound from Lemma~\ref{lem:giulia} on $\log |{\T}|$, we obtain
\bb
&-\frac1n\log \Tr[E_n\sigma_{A^n}] \\
&\quad \ge
-\frac1n\log \Tr\big[E_n\sigma_A(\delta)^{\otimes n}\big]
-
h\left(\frac rn\right)
-
\errata{\frac {3r}n\log d_{A}}
-
\frac rn \log \frac1{\delta} - \frac{n-r}{n}\log\frac{1}{1-\delta} .
\ee
In particular, if \(r=o(n)\), then \(r/n\to0\) and \(h(r/n)\to0\) as \(n \to \infty\), so that
\bb
-\frac1n\log \Tr[E_n\sigma_{A^n}]
\ge
-\frac1n\log \Tr\big[E_n\sigma_A(\delta)^{\otimes n}\big]
- \log \frac{1}{1-\delta} - o(1).
\ee
As a consequence,
\bb
&\sup_{\substack{(E_n)_{n\ge1}\\0\le E_n\le \id}}
\inf_{\substack{
\rho_n\xrightharpoonup{w}\rho\\
\sigma_{A^n}\xrightharpoonup{\rm \errata{MSR}}\sigma}}\left\{
\liminf_{k\to\infty}
\left(
-\frac1k\log\Tr[E_k\sigma_{A^k}]
\right): \liminf_{k\to\infty}\,\Tr[E_k\rho_k]\ge 1-\epsilon
\right\}
\\
&\qquad\ge
\sup_{\substack{(E_n)_{n\ge1}\\0\le E_n\le \id}}
\inf_{
\rho_n\xrightharpoonup{w}\rho}
\left\{\liminf_{k\to\infty}
\left(
-\frac1k\log\Tr\big[E_k\sigma_A^{\otimes k}(\delta)\big]
\right):\liminf_{k\to\infty}\Tr[E_k\rho_k]\ge 1-\epsilon\right\} - \log \frac{1}{1-\delta} \\
&\qquad = D(\rho\|\sigma\errata{(\delta)}) - \log \frac{1}{1-\delta},
\ee
with the convention $\inf \emptyset = 0$. Taking the \errata{limit $\delta\to 0^+$} completes the proof of Theorem~\ref{Stein_aiid_thm} in full generality\errata{, due to the lower semicontinuity of the quantum relative entropy}.

\begin{rem}
    Suppose that $H_0$ is composite and $H_1$ is MSR almost i.i.d.; then, the identical proof strategy shows that
    \bb
        {\rm Stein}_\epsilon\big(\pazocal{F}\|\pazocal{A}_\sigma^R\big)\geqt{(a)} {\rm Stein}_\epsilon(\pazocal{F}\|\sigma)\eqt{(b)} D(\pazocal{F}\|\sigma),
    \ee
    where (a) holds for any arbitrary family of states $\pazocal{F}$, and (b) holds when the generalised quantum Sanov theorem applies.
\end{rem}

\section{Robustness of data compression}\label{sec:data_comp}

\textbf{One-shot setting for arbitrary sources.} Given a possibly non-i.i.d. source $(\rho_n)_n$ source, namely a sequence of states $\rho_n\in\mathcal{D}(\mathcal{H})$, the one-shot approach to data compression can be generalised as follows. For any fixed $n$, we define a (fixed-lenght) $(n,\epsilon)$-code $\pazocal{C}_n=(\pazocal{E}_n,\pazocal{D}_n)$ with compressed space $\mathcal{K}_n$ of dimension $M(n,\epsilon)$ as a pair of an encoding channel 
$\pazocal E_n:\mathcal D(\mathcal H^{\otimes n})\to \mathcal D(\mathcal K_n)$ and a decoding channel $\pazocal D_n:\mathcal D(\mathcal K_n)\to\mathcal D(\mathcal H^{\otimes n})$ such that
\bb
   \dim\mathcal K_n \leq M(n,\epsilon) \qquad \text{and} \qquad F_e\!\left(\rho_n,\pazocal D_n\circ\mathcal E_n\right)\geq 1-\epsilon.
\ee

\noindent\textbf{Asymptotic setting for arbitrary sources.} Let $(\pazocal{C}_n)_n$ be a sequence of $(n,\epsilon_n)$-codes for $(\rho_n)_n$ with compressed spaces of dimension $M(n,\epsilon_n)$. Then, if $\epsilon_n\to 0$ as $n\to \infty$, the compression rate
\bb
    R\coloneqq \limsup_{n\to\infty} \frac 1n \log M(n,\epsilon)
\ee
is said to be achievable. The optimal compression rate of the source
\((\rho_n)_n\) is defined by
\bb 
R^\ast\big((\rho_n)_n\big)\coloneqq \inf \left\{ R:
R \text{ is achievable for }
(\rho_n)_n \right\}.
\ee

\noindent\textbf{Almost i.i.d.\ sources.} The previous definition is operationally meaningful if we exactly know the source. In the almost i.i.d.\ setting, as a particular case, we could define in principle $(n,\epsilon)$-codes for individual sources, and consequently also achievable and optimal rates. However, in the physical scenario we are interested in, the exact almost i.i.d.\ source along $\rho$ is not known, but we only have guarantees on the asymptotic behaviour $\rho_n\rightharpoonup \rho$. On the one hand, it is still possible to define sequence of codes $\pazocal{C}_n=(\pazocal{E}_n,\pazocal{D}_n)$ with compressed spaces $\mathcal{K}_n$, but we cannot upper bound the error at finite $n$ \emph{for all} almost i.i.d.\ sources along a fixed $\rho$ (unless there is no compression at all). However, it is still possible to require universal asymptotic guarantees as follows.
\begin{Def}[(Data compression for all weakly almost i.i.d.\ sources along $\rho$)]\label{def:data_almost_iid}
Let $\rho\in\mathcal D(\mathcal H)$ be a fixed reference state. For any fixed $\epsilon\in (0,1)$, a rate \(R_\epsilon\ge0\) is said to be \emph{universally achievable} for all weakly almost i.i.d.\ sources along $\rho$ with asymptotic entanglement fidelity at least $1-\epsilon$ if there exist Hilbert spaces \(\mathcal K_n\), an
encoding channel (compression) 
\bb
\pazocal E_n:
\mathcal D(\mathcal H^{\otimes n})
\to
\mathcal D(\mathcal K_n),
\ee
and a decoding channel (decompression) 
\bb
\pazocal D_n:
\mathcal D(\mathcal K_n)
\to
\mathcal D(\mathcal H^{\otimes n}),
\ee
such that
\bb
\limsup_{n\to\infty}
\frac1n\log\dim\mathcal K_n
\le R,
\ee
and, for all $\rho_n\xrightharpoonup{w}\rho$,
\bb\label{eq:ent_fid_bound}
\liminf_{n\to\infty}F_e\!\left(
\rho_n,
\pazocal D_n\circ\pazocal E_n
\right)
\geq 1-\epsilon.
\ee
namely the rate is achievable by every individual almost i.i.d.\ source along $\rho_n$ using \emph{the same} sequence of codes. The optimal compression rate for all almost i.i.d.\ sources along $\rho$ is then 
is defined by
\bb\label{eq:min_rate}
R^{\epsilon,\ast}_{w}(\rho)\coloneqq\inf \big\{
R_\epsilon:
R_\epsilon &\text{ is universally achievable by weakly almost i.i.d.\ sources}\\
& \text{ along $\rho$ with asymptotic entanglement fidelity at least $1-\epsilon$}
\big\}.
\ee
\end{Def}
Clearly, we have
\bb
    R^\ast\big((\rho_n)_n\big)\leq \lim_{\epsilon\to 0^+}R^{\epsilon,\ast}_{w}(\rho)
\ee
for all $\rho_n\xrightharpoonup{w}\rho$.

\subsection{Robustness of classical data compression}

\begin{figure}[t]
  \centering
  \def\svgwidth{0.85\linewidth}
  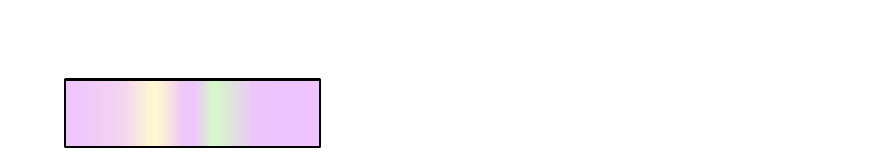
  \mycaption{Almost i.i.d.\ classical data compression.}{Let $P\in\mathcal{P}(\mathcal{X})$ be a classical probability distribution. In the i.i.d.\ setting (a), the aim of data-compression is to design a sequence of codes $(\pazocal{E}_n,\pazocal{D}_n)$ such that each sequence of symbols $X^n=(X_1,\dots, X_n)$ gets compressed to a message $m\in [M]$ with $M\ll |\mathcal{X}|^n$; then $m$ gets decompressed to the original sequence with vanishing error probability when $X^n\sim P^{\times n}$. Such a sequence of codes is \emph{robust} if, replacing $P^{\times n}$ with any arbitrary almost i.i.d.\ $P_n\xrightharpoonup{w}P$ -- as in (b) -- the error probability is still small in the limit $n \to \infty$.}
  \label{fig:data_comp}
\end{figure}

Before tackling the quantum case, let us discuss the compression of almost i.i.d.\ classical sources. More precisely, the aim of this section is to address the following question: is it possible to design a classical sequence of compression codes $(\pazocal{E}_n,\pazocal{D}_n)$ which are robust under arbitrary weakly almost i.i.d.\ perturbations of the source which is sampled (see Figure~\ref{fig:data_comp})?
\smallskip

Let $P$ denote a probability distribution on a finite alphabet $\mathcal{X}$, and for any $n \in {\mathbb{N}}_+$ let $P_n$ denote a probability distribution on ${\mathcal{X}}^n.$ Then a sequence of random variables $(X^n)_n$, with $X^n \sim P_n$, denotes a weakly almost i.i.d.~classical information source (along $P$) if $n$ uses of the source produce a sequence $x^n = (x_1, x_2, \ldots, x_n) \in \mathcal{X}^n$ with probability $P_n(x^n)$, 
and  \bb \label{cl-w-aiid}
        \limsup_{n\to \infty} \EE{\substack{I\subseteq [n],\\ |I|=k}}\|(P_n)_I - P^{\otimes I}\|_1= 0\qquad \forall k \in \mathbb{N}_+,
    \ee
    where the expectation is over uniformly random subsets $I \subseteq [n]$ of size $k$, and $(P_n)_I$ denotes the marginal of $P_n$ on $I$. We want to find a bound on the optimal rate of data compression (or source coding) for such a source.
    
    Recall that for any $\varepsilon \in (0,1)$ a fixed-length code $(n, \varepsilon)$-code, $\pazocal{C}_n$, of size $M$ for a classical information source $X^n\sim P_n$ consists of an encoding map $\pazocal{E}_n : \mathcal{X}^n \rightarrow \{1,2,\ldots, {M }\}$, and a decoding map $\pazocal{D}_n :  \{1,2,\ldots,  M \} \rightarrow {\mathcal{X}}^n$, such that the probability of error
    \bb
    p_{\rm err}(\pazocal{C}_n,{P_n})\coloneqq {\PP{X_n\sim P_n}}(X^n \neq {\pazocal{D}}_n({\pazocal{E}}_n(X^n))) \leq \varepsilon.
    \ee
    Clearly, $M \equiv M(n, \varepsilon)$. For any fixed values of  $n$ and $ \varepsilon$, let $ M^*(n, \varepsilon)$ denote the minimal value of $M(n, \varepsilon)$ over all such $(n, \varepsilon)$ codes. Then the optimal rate of data compression for the source is defined as
    \begin{align}\label{opt-rate}
        R^\ast\big((P_n)_n\big)\coloneq \lim_{\varepsilon \to 0} \limsup_{n \to \infty} \frac{\log M^*(n, \varepsilon)}{n}.
    \end{align}
     More generally, for any fixed  
     but arbitrary probability distribution $P\in\mathcal{P}(\mathcal{X})$ and $\epsilon\in(0,1)$, if there exists a sequence of codes $\pazocal{C}_n=(\pazocal{E}_n,\pazocal{D}_n)$ of size $M_n$
        \bb\label{eq:code0}
             \pazocal{E}_n : \mathcal{X}^n \to \{1,\dots, M_n \} \qquad \pazocal{D}_n :  \{1,\dots,  M_n \} \to \mathcal{X}^n\qquad n\geq 1,
        \ee
        satisfying
        \bb\label{eq:rate0}
            \limsup_{n\to\infty}\frac{\log M_n}{n}\leq R 
        \ee
        and        
        \bb\label{eq:error0}
            \limsup_{n\to\infty}p_{\rm err}(\pazocal{C}_n,P_n)\leq \epsilon,
        \ee
         for any arbitrary $P_n\xrightharpoonup{w}P$,
    we say that the rate $R_\epsilon$ is universally achievable for weakly almost i.i.d.\ sources along $P$ with error probability at most $\epsilon$ asymptotically.
    Similarly to~\eqref{eq:min_rate}, the infimum of these rates will be denoted by $R^{\epsilon,\ast}_{w}(P)$, and again
    \bb\label{eq:converse_rate}
    R^\ast\big((P_n)_n\big)\leq \lim_{\epsilon\to 0^+}R^{\epsilon,\ast}_{w}(P)
\ee
for all $P_n\xrightharpoonup{w}P$.
\smallskip
    
    The following statement establishes the robustness of data compression for weakly almost i.i.d.~classical information sources.
    \begin{boxedthm}{}
    \begin{thm}[(Robustness of classical data compression)]\label{thm:c-data-comp}
        Let $P\in \mathcal{P}(\mathcal{X})$ be fixed but arbitrary probability distribution on a finite set $\mathcal{X}$. For an individual weakly almost i.i.d.~classical information source $P_n\xrightharpoonup{w}P$, the optimal rate of data-compression satisfies
        \bb\label{eq:optimal}
        R^*\big((P_n)_n\big) \leq H(P),
        \ee
        where $H(P)$ is the Shannon entropy of $P$. Furthermore, for any arbitrary $\epsilon\in (0,1)$, there exists a sequence of codes $(\pazocal{E}_n,\pazocal{D}_n)$ achieving the optimal rate
        \bb
            R^{\epsilon,\ast}_{w}(P)=H(P)
        \ee
        with asymptotic error probability at most $\epsilon$ for all $P_n\xrightharpoonup{w}P$.
    \end{thm}
     \end{boxedthm}  
       
       \begin{proof}
       Let $n\geq 1$. Note that a fixed-length code $(\pazocal{E}_n,\pazocal{D}_n)$ of size $M_n$ can decode at most $M_n$ sequences without error. Namely, calling
       \begin{align}
           A_n \coloneqq \{ x^n \in {\cal{X}}^n \, ;  {\pazocal{D}}_n({\pazocal{E}}_n(x^n))= x^n\}.
       \end{align}
       we have $|A_n| \leq  M(n , \varepsilon)$. Conversely, any set
       $A_n \subseteq {\mathcal{X}}^n$ identifies a code $\pazocal{C}_n$ of size $|A_n|$
        by encoding elements of  $A_n$ injectively, decoding them perfectly, and mapping all elements $x^n \not\in A_n$ arbitrarily. Suppose the source $X^n$ has distribution $P_n$; then, the error probability for $\pazocal{C}_n$ is
       \bb
            p_{\rm err}(\pazocal{C}_n,{P_n})\coloneqq {\PP{X_n\sim P_n}}(X^n \notin A_n).
    \ee
       We construct a \emph{universal} code~\eqref{eq:code0}, i.e. one for which~\eqref{eq:error0} holds for all $P_n\xrightharpoonup{w}P$, whose rate satisfies $R=H(P)$. This immediately implies the upper bound~\eqref{eq:optimal} on the optimal rate of data compression by~\eqref{eq:converse_rate} for any individual weakly almost i.i.d.~source along $P$. The key idea behind constructing this code is to cast the data compression problem into the asymmetric binary hypothesis test stated below, and then leverage the robust acceptance function of Lemma~\ref{classical_Stein_aiid_lemma} to identify $A_n$. 
\begin{itemize}
    \item {\bf{Null hypothesis}}  $H_0$: the source $X^n$ is i.i.d.\ with distribution $P^{\times n}$ or, more in general, weakly almost i.i.d.\ along $P$ with unknown distribution $P_n$.
    \item 
       {\bf{Alternative hypothesis}}  $H_1$: the source is the uniform i.i.d.~source $X^n \sim Q^{\times n}$, where $Q(x)= \frac{1}{|\mathcal{X}|}$ for all $x \in {\mathcal{X}}$.
\end{itemize}
    Let us call $A_n^\epsilon$ the set of sequences that yield the acceptance of the null hypothesis according to the robust and universal test provided by Lemma~\ref{classical_Stein_aiid_lemma}, when the type I error is asymptotically constrained to be smaller than $\epsilon$. Let $e:A_n^\epsilon\to \{1,\dots, |A_n^\epsilon|\}$ be an injective enumeration of the sequences in the set $A_n^\epsilon$. Setting $M_n\coloneqq |A_n|$, we now define the following code:
    \bb
        \pazocal{E}_n(x^n)&\coloneqq \begin{cases}
            e(x^n) &  x^n\in A_n^\epsilon\\
            1 & x^n \notin A_n^\epsilon
        \end{cases}\qquad \text{and}\qquad
        \pazocal{D}_n(m)\coloneqq 
            e^{-1}(m) \quad \text{for all} \; m \in [M_n].
    \ee
     Now, suppose $P_n\xrightharpoonup{w}P$. Then,
    \bb
    p_{\rm err}(\pazocal{C}_n,{P_n})&\coloneqq \PP{X_n\sim P_n}\big(X^n \neq {\pazocal{D}}_n({\pazocal{E}}_n(X^n))\big)= \PP{X_n\sim P_n}\big(X^n \notin A_n^\epsilon\big)= \PP{X_n\sim P_n}\big(\text{type I error}\big),
    \ee
    whence, by Lemma~\ref{classical_Stein_aiid_lemma}, we have~\eqref{eq:error0} for any arbitrary $P_n\xrightharpoonup{w}P$. Let us now bound the compression rate for this sequence of codes: 
    \bb
        \frac{\log M_n}{n}&=\frac 1n \log |A_n^\epsilon|=\frac 1n \log \sum_{x^n\in A_n^\epsilon}\frac 1 {|\mathcal{X}|^n}+\log|\mathcal{X}|\\
        &=\frac 1n \log \PP{X^n \sim Q^{\times n}}\big(X^n\in A_n^\epsilon\big)+\log|\mathcal{X}|=\frac 1n \log \PP{X^n \sim Q^{\times n}}\big(\text{type II error}\big)+\log|\mathcal{X}|
    \ee
    Again leveraging Lemma~\ref{classical_Stein_aiid_lemma}, we conclude that
    \bb
        \limsup_{n\to\infty}\frac{\log M_n}{n}&\geq -D(P\|Q)+\log|\mathcal{X}| = H(P),
    \ee
    By the strong converse of Schumacher's data compression for i.i.d.\ information sources (see e.g.~\cite{Sharma2014}) we conclude that the code we have constructed is optimal:
    \bb
        H(P)=R^{\epsilon,\ast}(P)\leq R^{\epsilon,\ast}_{w}(P)\leq H(P),
    \ee
    where $R^{\epsilon,\ast}(P)$ is the optimal compression rate for the i.i.d.\ source $P^{\times n}$ with the constraint that the error probability is upper bounded by $\epsilon$.
    \end{proof}

       \subsection{Robustness of quantum data compression}

\begin{figure}[t]
  \centering
  \def\svgwidth{\linewidth}
  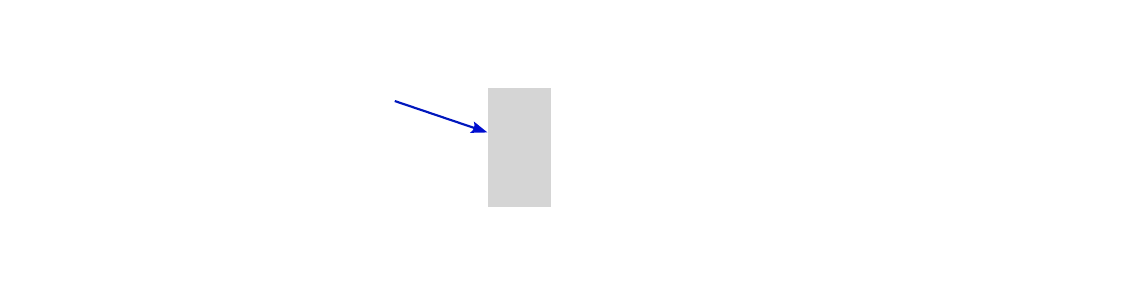
  \mycaption{Almost i.i.d.\ quantum data compression.}{Let $\rho\in\mathcal{D}(\mathcal{H}_A)$ be a state. In the i.i.d.\ setting (a), a data compression code $(\pazocal{E}_n,\pazocal{D}_n)$ acts on the subsystem $A^n$ of an arbitrary -- possibly entangled -- state $\Psi_{A^nR}$ with marginal $\Tr_R \Psi_{A^nR}=\rho^{\otimes n}$ by mapping it into a space $\mathcal{K}_n$; then, for a good code, the decompressed state $\sigma_{A^nR}=\big((\pazocal{D}_n\circ\pazocal{E}_n)\otimes {\rm Id}_R\big)(\Psi_{A^nR})$ should have high fidelity with the original state $\Psi_{A^nR}$. A sequence of codes is \emph{robust} if, replacing $\Psi_{A^nR}$ with a state $\tilde{\Psi}_{A^nR}$ having a weakly almost i.i.d.\ marginal $\rho_n\xrightharpoonup{w}\rho$ -- as in (b) -- the fidelity between $\tilde \Psi_{A^nR}$ and the new decompressed state $\tilde \sigma_{A^nR}$ is still asymptotically large.} 
  \label{fig:q_data_comp}
\end{figure}

We are now ready to generalise the idea discussed in the classical setting to quantum states, keeping in mind the framework introduced at the beginning of Section~\ref{sec:data_comp}.

The following theorem establishes the robustness of quantum data compression for weakly almost-i.i.d.~quantum information sources.
\begin{boxedthm}{}
\begin{thm}[(Robustness of quantum data compression)]\label{thm:q-data-comp}
Let $\rho\in \mathcal{D}(\mathcal{H})$ be an arbitrary fixed quantum state. For an individual weakly almost i.i.d.~quantum information source $\rho_n\xrightharpoonup{w}\rho$, the optimal rate of data-compression satisfies
        \bb\label{eq:optimal}
        R^*\big((\rho_n)_n\big) \leq S(\rho),
        \ee
        where $S(\rho)$ is the von Neumann entropy of $\rho$. Furthermore, for any arbitrary $\epsilon\in (0,1)$, there exists a sequence of codes $(\pazocal{E}_n,\pazocal{D}_n)$ achieving the optimal rate
        \bb
            R^{\epsilon,\ast}_{w}(\rho)=S(\rho)
        \ee
        with asymptotic entanglement fidelity at least $1-\epsilon$ for all $\rho_n\xrightharpoonup{w}\rho$, namely
       \bb\label{eq:error}
    \liminf_{n\to\infty}\, F_e\!\left(\rho_{n}, \pazocal{D}_n \circ \pazocal{E}_n \right) \geq 1-\epsilon .
       \ee
    \end{thm}    
\end{boxedthm}
       
\begin{proof}
Given any arbitrary sequence of subspaces $\mathcal{K}_n \subseteq \mathcal{H}^{\otimes n}$, with $n\geq 1$, let $\Pi_n $ be the orthogonal projector onto $\mathcal{K}_n$, namely
\begin{align}
\mathcal K_n = \operatorname{supp}\Pi_n \subseteq \mathcal H^{\otimes n}.
\end{align}
Let us define an encoding map corresponding to the chosen subspace $\mathcal{K}_n$ as follows: for all $\omega \in {\mathcal{D}}({\mathcal{H}}^{\otimes n})$,
\begin{align}
{\pazocal{E}}_n(\omega)
\coloneqq
\Pi_n \omega \Pi_n
+
\Tr[(\id-\Pi_n)\omega]\,\sigma_{{\mathcal{K}}_n},
\end{align}
where $\sigma_{{\mathcal{K}}_n}$ is a fixed state with ${\rm{supp\,}}\sigma_{{\mathcal{K}}_n}\subseteq {\mathcal{K}}_n$. Choose then the decoding map ${\pazocal{D}}_n$ to be the inclusion map induced by the embedding ${\mathcal{K}}_n \subseteq {\mathcal{H}}^{\otimes n}$,  i.e.~$\forall$ $\nu \in {\mathcal{D}}({\mathcal{K}}_n)$,
\begin{align}
{\pazocal{D}}_n(\nu)\coloneqq\nu\,\oplus\, 0_{\mathcal{K}_n^\perp}.
\end{align}
Calling $\Lambda_n (\omega) \coloneqq {\pazocal{D}}_n \circ {\pazocal{E}}_n$, we then have 
\begin{align}
\Lambda_n (\omega) = {\pazocal{D}}_n \circ {\pazocal{E}}_n(\omega)
=
\Pi_n \omega \Pi_n
+
\Tr [(\id-\Pi_n)\omega]\,\sigma_{\mathcal K_n}.
\end{align}
Suppose the eigenvalue decomposition of $\sigma_{\mathcal K_n}$ be given by 
\begin{align}
\sigma_{\mathcal K_n} = \sum_j s_j |s_j\rangle\langle s_j|,
\end{align}
and let $\{|f_\alpha\rangle\}_\alpha$ be an orthonormal basis of ${\rm{ran}}(\id-\Pi_n)$. Then the operators
\bb
K_0 &\coloneqq \Pi_n,\,\,\\
K_{j,\alpha} &\coloneqq \sqrt{s_j}\, |s_j\rangle\langle f_\alpha|, &\quad j= 1,\ldots , M_n, \quad\alpha =1,\ldots, {\rm{dim}}({\rm{ran}}(I - \Pi_n)).
\ee
form a set of Kraus operators of $\Lambda_n$.
Recall that for a CPTP map $\Lambda$ with Kraus operators $\{A_k\}_k$ we have (see~\eqref{eq:NC})
\bb
F_e(\rho, \Lambda) = \sum_k |\Tr(\rho A_k)|^2,
\ee
whence
\begin{align}
    F_e(\rho_{n}, {\Lambda}_n) &\geq \left[ \Tr(\rho_n K_0)\right]^2 = [\Tr(\rho_n \Pi_n)]^2.
\end{align}
Suppose now that we identify a sequence of subspaces $\{\mathcal{K}_n\}$ such that
\begin{align}\label{cond1}\liminf_{n\to\infty}\Tr[\rho_n \Pi_n] \geq 1 - \frac \varepsilon 2 \quad \text{for all}\quad \rho_n\xrightharpoonup{w}\rho.
\end{align}
In that case, the code $(\pazocal{E}_n,\pazocal{D}_n)$ satisfies the required bound~\eqref{eq:ent_fid_bound} on the entanglement fidelity, since 
\begin{align}\label{eq:ent-fid-bd}
\liminf_{n\to\infty}F_e (\rho_n , \Lambda_n) 
\geq (1 - \varepsilon^\prime)^2
= (1 - {\varepsilon}/{2})^2 \geq 1 - \varepsilon.
\end{align}
Thus, a sequence of projectors $\{\Pi_n\}_n$ that satisfies~\eqref{cond1}
identifies a valid protocol achieving the rate
\bb
    R_\epsilon=\limsup_{n\to\infty} \frac{\log \dim\mathcal{K}_n}{n}.
\ee
As in the classical case (see the proof of Theorem~\ref{thm:c-data-comp}), we are going to cast the data compression problem into the
asymmetric binary hypothesis testing problem stated below; by leveraging the robust sequence binary POVM of Theorem~\ref{Stein_aiid_thm}, we will identify a valid sequence of subspace $\{\mathcal{K}_n\}_n$ for our purpose.

\begin{itemize}
    \item {\bf{Null hypothesis}} $H_0:$ the source is the i.i.d.\  sequence of states $(\rho^{\otimes n})_n$ or any other arbitrary (unknown) weakly almost i.i.d.\ source $\rho_n$ along $\rho$.
    \item {\bf{Alternative hypothesis}} $H_1:$ the source is a uniform i.i.d.~source $\tau^{\otimes n}$,
where $ \tau \coloneqq {\id}/{d} \in {\pazocal{D}}({\mathcal{H}})$ denotes the completely mixed state.
\end{itemize}

Let us call $\{E_n\}_n$ the robust sequence of POVMs given by Theorem~\ref{Stein_aiid_thm} with type I error bounded by $\epsilon/2$; more precisely,  for all $\rho_n\xrightharpoonup{w}\rho$,
\bb\label{eq:errors}
\limsup_{n\to\infty}\PP{\rho_n}({\hbox{type I error}}) &= \limsup_{n\to\infty}\Tr[(\id - E_n)\rho_n]\leq \epsilon/2,\\
\liminf_{n\to\infty}-\frac 1n\log\PP{\tau^{\otimes n}} ({\hbox{type II error}}) &= \liminf_{n\to\infty}-\frac 1n\log\Tr[E_n\tau^{\otimes n}] \geq D(\rho\|\tau).
\ee
Now we want to extract a projective measurement $\{T_n, \id-T_n\}$ (i.e.\ $T_n^\dagger = T_n$, $T_n^2 = T_n$) from our POVM $\{E_n, \id- E_n\}$. We proceed in the following manner. Define
\begin{align}
    T_n \coloneqq \big\{E_n \geq \tfrac{1}{n} \id\big\},
\end{align}
which implies that $E_n \leq T_n + \frac 1n \id$, \errata{as}
\bb
    E_n-\tfrac{1}{n} \id \leq \big(E_n-\tfrac{1}{n} \id\big)_+ \leq\big\{E_n\geq \tfrac{1}{n} \id\big\},
\ee
where in the last inequality we have used the fact that the spectrum of $E_n-\alpha\id$ is upper bounded by $1$. \errata{Since $E_n \leq T_n + \frac 1n \id$, we have
\begin{equation}\label{eq:T_epsilon}
    \Tr[T_n\,\rho_n] \ge \Tr[E_n\,\rho_n] - \frac{1}{n}
\end{equation}
i.e.\ the probability of type I error is at most $\frac{1}{n} + \Tr[(\id - E_n)\rho_n]$.}
Let us now compare the probabilities of type II error for the POVM $\{E_n, \id- E_n\}$ and the projective measurement $\{T_n, \id - T_n\}$. Since every eigenvalue of $E_n$ on ${\rm{ran\,}} T_n$ is at least $1/n$,
\bb
\Tr E_n \geq \frac 1n \Tr T_n,
\ee
and hence
\begin{align}\label{eq:ranT}
    {\rm rk}\, T_n &=\Tr T_n\leq n \Tr E_n  = nd^n \Tr [E_n \tau^{\otimes n}]
\end{align}
Now, we claim that $(T_n)_n$ is exactly the sequence of projectors we are looking for, namely, if we set $\Pi_n\coloneqq T_n$ we get~\eqref{cond1} and the exponential growth of $\dim \mathcal{K}_n={\rm rk}\, T_n$ is upper bounded by $S(\rho)$. Indeed, by~\eqref{eq:T_epsilon},
\bb
    \liminf_{n\to\infty}\Tr [T_n \rho_n] \geq 1 - \frac{\epsilon}{2},
\ee
and by~\eqref{eq:ranT} combined with~\eqref{eq:errors}, we get
\bb
    R_\epsilon =\limsup_{n\to\infty}\frac{\log\dim\mathcal{K}_n}{n} \leq \log d+\limsup_{n\to\infty}\frac{1}{n}\log \Tr[E_n\tau^{\otimes n}] \leq \log d -D(\rho\|\tau)=S(\rho).\label{eq:dc-ach}
\ee
This means that $S(\rho)$ is universally achievable by weakly almost i.i.d.\ sources along $\rho$, with asymptotic fidelity at least $1-\epsilon$. For i.i.d.\ quantum sources, Schumacher compression is known to satisfy a strong converse:
compression below the von Neumann entropy rate forces the entanglement fidelity to vanish asymptotically. Equivalently, for an i.i.d.\ quantum source given by $\rho$, 
$R^{\epsilon,\ast}(\rho)\geq S(\rho)$ for any $\varepsilon \in (0,1)$, where $R^{\epsilon,\ast}(\rho)$ denotes the minimal rate of data compression that can be achieved with asymptotic entanglement fidelity of at least $1-\epsilon$ (see e.g.~\cite{Sharma2014} and references therein). This also ensures that 
\bb R^{\epsilon,\ast}_{w} (\rho)\geq S(\rho),
\ee
where $R^{\epsilon,\ast}_{w} (\rho)$ is defined through~\eqref{eq:min_rate}.
This is because the i.i.d.\ source $\rho$ (i.e.\ the sequence of i.i.d.\ states $(\rho^{\otimes n})_n$) itself belongs to the class of weakly almost i.i.d.\ sources along $\rho$, and the compression schemes considered are universal for the entire class.
Hence, together with~\eqref{eq:dc-ach} this implies that 
\bb
R^{\epsilon,\ast}_{w}(\rho)=S(\rho),
\ee
which concludes the proof.
\end{proof}

\section{Robustness of classical information transmission with quantum channels}\label{sec:channels}

In order to study the robustness of classical information communication via quantum channels, we need to extend the definition of Section~\ref{sec:intro_reliable} to sequences of possibly non-i.i.d.\ channels. For $n\geq 1$, let $\tilde{\pazocal{N}}=(\tilde{\pazocal{N}}^{(n)})_n$ be a sequence of quantum channels $\tilde{\pazocal{N}}^{(n)}_{A^nB^n}$ mapping the states on 
$\pazocal{H}_A^{\otimes n}$ into states 
on $\pazocal{H}_B^{\otimes n}$. The definition given in~\eqref{eq:c_epsilon} can now be generalised,
{\errata{for any fixed $\varepsilon \in (0,1)$, to}}
\bb
    C_\epsilon(\tilde{\pazocal{N}})\coloneqq\sup\Big\{r\geq 0 :&&\!\! \limsup_{n\to\infty}\inf_{\substack{\pazocal{C}_n \text{ code}\\ \text{of size } \lceil 2^{rn}\rceil}}p_{\rm err}\big(\pazocal{C}_n,\tilde{\pazocal{N}}^{(n)}\big)\leq \epsilon\Big\}
\ee
We then define the classical capacity $C(\tilde{\pazocal{N}})$ of the sequence of channels $\tilde{\pazocal{N}}$ as
\bb
    C(\tilde{\pazocal{N}})\coloneqq \lim_{\epsilon\to 0}C_\epsilon(\tilde{\pazocal{N}}),
\ee
generalising the i.i.d.\ case of~\eqref{eq:c}.
Finally, for all rates below $C(\tilde{\pazocal{N}})$, the reliability function can be defined as
\bb
    E(r,\pazocal{N})\coloneqq\liminf_{n\to\infty}-\frac{1}{n}\log \inf_{\substack{\pazocal{C}_n \text{ code}\\ \text{of size } \lceil 2^{rn}\rceil}}p_{\rm err}\big(\pazocal{C}_n,\tilde{\pazocal{N}}^{(n)}\big).
\ee

The following theorem establishes the robustness of the classical capacity for almost i.i.d.\ processes.

\begin{boxedthm}{}
\begin{thm}[(Robustness of the classical capacity)]\label{thm:channels} Let us suppose that a sequence of quantum channels 
\bb
\tilde{\pazocal{N}}=(\tilde{\pazocal{N}}^{(n)})_{n\geq 1},\qquad \text{with}\qquad \tilde{\pazocal{N}}^{(n)}:\mathcal{D}(\mathcal{H}_A^{\otimes n})\to \mathcal{D}(\mathcal{H}_B^{\otimes n}),
\ee
is an almost i.i.d.\ process along a quantum channel \mbox{$\pazocal{N}:\mathcal{D}(\mathcal{H}_A)\to \mathcal{D}(\mathcal{H}_B)$}, i.e. 
\bb
\big(\tilde{\pazocal{N}}^{(n)}\big)_n \xrightharpoonup{\clubsuit} \pazocal{N}.
\ee
Then, the unassisted classical capacity of $\tilde{\pazocal{N}}$ exactly equals the unassisted classical capacity of $\pazocal{N}$.
    \bb
        C(\tilde{\pazocal{N}})= C(\pazocal{N}).
    \ee

\errata{
In general,} given any arbitrary quantum channel $\pazocal{N}_{A\to B}$, there exists a universal sequence of codes $(\pazocal{E}_n,\pazocal{D}_n)_{n\geq 1}$ {\errata{using which one can transmit information at a rate $C(\pazocal{N})$ with asymptotically vanishing error, for any arbitrary almost i.i.d.\ process $\tilde{\pazocal{N}}$ along $\pazocal{N}$.}}
\end{thm}
\end{boxedthm}

\begin{proof}
\errata{In Section~\ref{sec:converse} we prove the converse inequality $C(
\tilde{\pazocal{N}}) \leq C(\pazocal{N})$. In Section~\ref{sec:achivability_quantum}, we prove the matching achievability result by constructing codes whose communication rates converge to $C(\pazocal{N})$ while the error probability vanishes asymptotically. The construction relies on a preliminary analysis of the fully classical case given in Section~\ref{sec:achievability_classical}. This concludes the proof.}
\end{proof}

\subsection{The entropic converse}\label{sec:converse}

\begin{prop}\label{prop:converse} Let $\pazocal{N}$ be a classical-quantum channel and let $\tilde{\pazocal{N}}=\big(\tilde{\pazocal{N}}^{(n)}\big)_n$ be an almost i.i.d.\ process along $\pazocal{N}$. Then,
    \bb
        C(\tilde{\pazocal{N}})\leq C(\pazocal{N}).
    \ee
\end{prop}

\begin{proof}
Recall that for a quantum channel $\Lambda_{A\to B}$, the Holevo information is defined as 
\bb
\chi(\Lambda) \coloneqq \sup_{\rho_{XA}} I(X:B)_{{\rho}'}, 
\ee
where the supremum is over classical-quantum states $\rho_{XA} = \sum_x P_X(x) \ketbra{x} \otimes \rho_x^A$, and
\bb
\rho_{XB}' \coloneqq ({\rm Id}_X\otimes\Lambda_{A\to B})(\rho_{XA}) = \sum_x P_X(x) \ketbra{x} \otimes \Lambda_{A\to B}\big(\rho_x^A\big) .
\ee
It is well known (see e.g.~\cite[Eq.~(20.56)--(20.65)]{MARK}) that, for any fixed $\varepsilon\in (0,1)$, the number of classical messages $M$ that can be transmitted via the channel $\Lambda$ with average error probability $\epsilon$ is upper bounded as
\bb
    \log M \leq\frac{\chi(\Lambda)+g(\epsilon)}{1-\epsilon},
\ee
where $g(\epsilon)\coloneqq (\epsilon+1)\log(\epsilon+1)-\epsilon\log\epsilon$. Then, we have
\bb\label{eq:152}
    C_\epsilon(\pazocal{\tilde{\pazocal{N}}})\leq \frac{1}{1-\epsilon}\liminf_{n\to\infty} \frac 1n \chi(\tilde{\pazocal{N}}^{(n)})\eqt{(a)} \errata{\liminf_{n\to\infty}}\frac 1n \chi(\pazocal{N}^{\otimes n})=C(\pazocal{N}),
\ee
where (a) follows from Theorem~\ref{thm:SW}. Indeed, calling
\bb
    \tilde\rho_{X^nB^n}'&\coloneqq \sum_{x^n}p(x^n)\ketbra{x^n}\otimes\tilde{\pazocal{N}}^{(n)}\big(\rho^{A^n}_{x^n}\big),\\
    \rho_{X^nB^n}'&\coloneqq \sum_{x^n}p(x^n)\ketbra{x^n}\otimes \pazocal{N}^{\otimes n}\big(\rho^{A^n}_{x^n}\big),
\ee
\errata{which satisfy $\tilde \rho_{X^n}'=\rho_{X^n}'$,}
we have
\bb
    I(X^n:B^n)_{\tilde \rho'} &= S(\tilde \rho_{B^n}')+S(\tilde \rho_{X^n}')-S(\tilde\rho_{X^nB^n}')\\
    &\leq S(\rho_{B^n}')+S(\rho_{X^n}')-S(\rho_{X^nB^n}')+2n f_d\left(\frac{w_n}{n}\right)=I(X^n:B^n)_{\rho'}+o(n),
\ee
with \errata{$f_d(x)\coloneqq -x\log x-(1-x)\log(1-x) + x\ln\left(\left(\dim\mathcal{H}_{XB}\right)^2-1\right)$} and
\bb
    \|\tilde \rho_{B^n}' - \rho_{B^n}'\|_{W_1}\leq \|\tilde \rho_{X^nB^n}' - \rho_{X^nB^n}'\|_{W_1} \leq \big\|\tilde{\pazocal{W}}^{(n)}-\pazocal{W}\big\|_{\,\clubsuit}\eqqcolon w_n.
\ee
Taking the limit $\epsilon\to 0^+$ in~\eqref{eq:152} completes the proof.
\end{proof}

\errata{
\begin{rem}
[(No matching converse bound for weakly converging output)] Having in mind the notion of weakly almost i.i.d. sources, we may be interested in its extension to channels, i.e.\ we would want to require
\bb \lim_{n\to\infty}\sup_{\rho_n}\EE{\substack{I\subseteq[n]\\|I|=k}}\Big\|\Tr_{I^c}\big[(\tilde{\pazocal{N}}^{(n)}(\rho_n)\big]-\Tr_{I^c}\big[\pazocal{N}^{\otimes n}(\rho_n)\big]\Big\|_1=0\qquad \forall k \geq 1.
\ee
One might wonder whether the converse bound of Proposition~\ref{prop:converse} still holds under the above assumption. The answer is negative, as shown by the following example.
\end{rem}}
\begin{prop}
    For any $n\in\mathbb{N}$, let us consider the classical channel that maps the input bit string $x\in\{0,1\}^{2n}$ to the following random bit string $Y\in\{0,1\}^{2n}$:
    \begin{equation}\label{eq:159}
        Y_{2i-1} = x_{2i-1} + Z_i\,,\qquad Y_{2i} = x_{2i} + Z_i\,,\qquad i\in[n]\,,
    \end{equation}
    where {\errata{ $Z\coloneqq (Z_1, \ldots, Z_n)$}} is a uniformly distributed string in $\{0,1\}^n$.
    Then, the channel is weakly almost-i.i.d. along the completely depolarizing channel, which has zero capacity, but its capacity is at least $\frac{1}{2}\log2$.
\end{prop}

\begin{proof}
Let us fix $k\in\mathbb{N}$.
The marginal of the output of the channel on a uniformly random subset of $[2n]$ of size $k$ is the uniform distribution on $\{0,1\}^k$ provided that the subset does not contain any 
pair of bits of the form $\{2i-1,\,2i\}$.
This happens with probability
\begin{equation}
    \frac{2\,n}{2\,n}\,\frac{2\left(n-1\right)}{2\,n-1}\ldots\frac{2\left(n-k+1\right)}{2\,n-k+1}\,,
\end{equation}
which tends to $1$ for $n\to\infty$.
\errata{Therefore, the sequence of channels defined through~\eqref{eq:159} constitutes a weakly almost i.i.d.\ process along the completely depolarizing channel.}

Let us show that the channel can send a message $w\in\{0,1\}^n$ with zero error probability.
Let us encode $w$ in $x\in\{0,1\}^{\errata{2n}}$ with
\begin{equation}
    x_{2i-1} = w_i\,,\qquad x_{2i} = 0\,,\qquad i\in[n]\,.
\end{equation}
We have
    \begin{equation}
        Y_{2i-1} = w_i + Z_i\,,\qquad Y_{2i} = Z_i\,,\qquad i\in[n]\,,
    \end{equation}
    hence the message can be recovered with
    \begin{equation}
        w_i = Y_{2i-1} + Y_{2i}\,.
    \end{equation}
    It follows that the capacity of the channel is at least $\frac{1}{2}\log2$.
\end{proof}

\subsection{Achievability in the classical case}\label{sec:achievability_classical}
Let $\pazocal{W}$ be a discrete memoryless channel from $\mathcal{X}$ to $\mathcal{Y}$. Let $\tilde{\pazocal{W}}$ be the sequence of channels
\bb
    \tilde{\pazocal{W}}\coloneqq\big(\tilde{\pazocal{W}}^{(n)}\big)_{n\geq 1},
\ee
where $\tilde{\pazocal{W}}^{(n)}$ is a discrete channel from $\mathcal{X}^n$ to $\mathcal{Y}^n$. We say that $(\pazocal{E},\pazocal{D})$ is a $(M,\epsilon)$-code for $\pazocal{W}$ if
\bb
    \frac{1}{M}\sum_{m=1}^{\lceil M\rceil}\mathbb{P}\big(m=\pazocal{D}\circ\pazocal{W}\circ\pazocal{E}(m)\big)\leq \epsilon.
\ee
The classical capacity $C(\pazocal{W})$ of $\pazocal{W}$ is defined as
\bb\label{eq:capacity_def}
    C(\pazocal{W})\coloneqq\sup\Big\{r\geq 0 :&&\!\!  \forall n \geq 1\; \exists\; (2^{nr}, \epsilon_n) \text{-code for } \pazocal{W}^{\times n},\;\;\\
    &&\text{ such that } \epsilon_n\to 0 \text{ as } n \to \infty\Big\}
\ee
The classical capacity $C(\tilde{\pazocal{W}})$ of the sequence of channels $\tilde{\pazocal{W}}$ is defined as
\bb\label{eq:capacity_def_2}
    C(\tilde{\pazocal{W}})\coloneqq\sup\Big\{r\geq 0 : &&\!\! \forall n \geq 1\; \exists\; ( 2^{nr}, \epsilon_n) \text{-code for } \tilde{\pazocal{W}}^{(n)},\;\;\\
    &&\text{ such that } \epsilon_n\to 0 \text{ as } n \to \infty\Big\}
\ee
Note that~\eqref{eq:capacity_def} is a particular case of~\eqref{eq:capacity_def_2} when considering -- with some abuse of notation -- the sequence $\pazocal{W}=\big(\pazocal{W}^{\times n}\big)_n$.

\subsubsection{A pictorial interpretation of almost i.i.d.\ classical processes}
The club distance for classical channels can be interpreted according to optimal transport formulation for the classical Wasserstein distance of order 1 based on the Hamming distance (see Section~\ref{sec:intro_W}):
\bb\label{eq:club_plan}
    \big\|\tilde{\pazocal{W}}^{(n)}-\pazocal{W}^{\times n}\big\|_{\, \clubsuit}=\max_{x^n}\min_{\pi_{x^n}}\EE{(\tilde Y^n,Y^n)\sim \pi_{x^n}}\big[d_H(\tilde Y^n,Y^n)\big]
\ee
where $d_H(\tilde y^n,y^n)$ represents the Hamming distance between $\tilde y^n$ and $y^n$, the maximisation runs over any $x^n\in\mathcal{X}^n$, and the minimum is taken over all the couplings $\pi_{x^n}:\mathcal{Y}^n\times \mathcal{Y}^n\to [0,1]$ of $\tilde{\pazocal{W}}^{(n)}(\,\cdot\,|x^n)$ and $\pazocal{W}^{\times n}(\,\cdot\,|x^n)$, namely
\bb\label{eq:marginals}
    \sum_{y^n\in\mathcal{Y}^n}\pi_{x^n}(\tilde y^n,y^n)=\tilde{\pazocal{W}}^{(n)}(\tilde y^n\,|x^n),\qquad \sum_{\tilde y^n\in\mathcal{Y}^n}\pi_{x^n}(\tilde y^n,y^n)=\pazocal{W}^{\times n}(y^n\,|x^n).
\ee
We can equivalently write the coupling in terms of a conditional distribution $\Phi_{x^n}$ as follows
\bb\label{eq:noise}
    \pi_{x^n}(\tilde y^n, y^n)=\Phi_{x^n}(\tilde y^n|y^n)\pazocal{W}^{\times n}(y^n\,|x^n),
\ee
i.e.
\bb\label{eq:phi_def}
    \Phi_{x^n}(\tilde y^n|y^n)\coloneqq \begin{cases}
        \frac{\pi_{x^n}(\tilde y^n, y^n)}{\pazocal{W}^{\times n}(y^n\,|x^n)}& \pazocal{W}^{\times n}(y^n\,|x^n)\neq 0,\\
        0 & \text{otherwise}.
    \end{cases}
\ee
Let $\pi^\ast_{x^n}$ be an optimal coupling in~\eqref{eq:club_plan}, yielding the optimal conditional distribution $\Phi_{x^n}^\ast$. Then, combining~\eqref{eq:marginals} with~\eqref{eq:noise}, we get
\bb\label{eq:interpr_noise}
    \tilde{\pazocal{W}}^{(n)}(\tilde y^n\,|x^n)=\sum_{y^n\in\mathcal{Y}^n} \Phi_{x^n}^\ast(\tilde y^n|y^n)\pazocal{W}^{\times n}(y^n\,|x^n)
\ee
and
\bb
    \EE{(\tilde Y^n,Y^n)\sim \pi^\ast_{x^n}}\big[d_H(\tilde Y^n,Y^n)\big]\leq \big\|\tilde{\pazocal{W}}^{(n)}-\pazocal{W}^{\times n}\big\|_{\, \clubsuit}\qquad \forall x^n\in\mathcal{X}^n.
\ee

\begin{figure}[t]
  \centering
  \def\svgwidth{0.8\linewidth}
  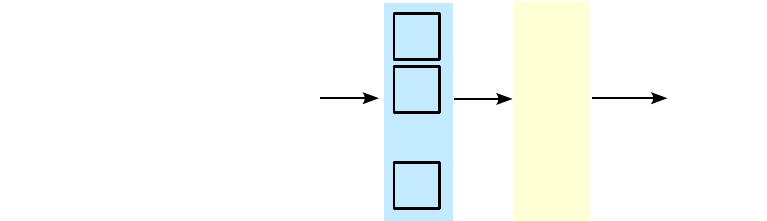
  \mycaption{Schematic interpretation of almost i.i.d.\ processes.}{Writing the definition of the club distance between an i.i.d.\ channel $\pazocal{W}^{\times n}$ and a non i.i.d.\ channel $\tilde{\pazocal{W}}^{(n)}$ in terms of couplings as in~\eqref{eq:club_plan}, it is possible to interpret the transformation of a random input $X^n$ to a random output \errata{$\tilde Y^n$} induced by $\tilde{\pazocal{W}}^{(n)}$ as a two-step transformation of $X^n$: first, the i.i.d.\ channel $\pazocal{W}^{\times n}$ is applied to $X^n$, yielding an output sequence $Y^n$; then, a further noisy channel $\Phi_{X^n}^\ast$, possibly depending on the input $X^n$, transforms $Y^n$ into $\tilde Y^n$ by modifying, on average, a number of bits that is upper bounded by the club distance between $\pazocal{W}^{\times n}$ and $\tilde{\pazocal{W}}^{(n)}$.} 
  \label{fig:coupling}
\end{figure}

The action of the channel $\tilde{\pazocal{W}}^{(n)}$ can therefore be interpreted as the composition of two channels: the use of the i.i.d.\ local channel $\pazocal{W}^{\times n}$ followed by the global noise $\Phi_{x^n}^\ast$ spoiling, on average, at most $\big\|\tilde{\pazocal{W}}^{(n)}-\pazocal{W}^{\times n}\big\|_{\, \clubsuit}$ symbols of the outcome of $\pazocal{W}^{\times n}$. In Figure~\ref{fig:coupling} we represent
\bb\label{eq:YY}
    X^n &:&& \text{random sequence in input, having distribution } P_{X^n};\\[0.3em]
    Y^n &:&& \text{random output of } \pazocal{W}^{\times n} \text{, when fed with } X^n;\\[0.3em]
    \tilde Y^n &:&& \text{random output of } \Phi_{X^n} \text{, when fed with } Y^n,\\
    & && \text{which has distribution } \tilde {\pazocal{W}}^{(n)}_{\tilde Y^n|X^n}P_{X^n};\\[0.3em]
     d_H(\tilde Y^n,Y^n) &:&& \text{number of symbols that get spoilt by } \Phi_{X^n}^\ast,\\
    & &&\text{which is upper bounded by } \big\|\tilde{\pazocal{W}}^{(n)}-\pazocal{W}^{\times n}\big\|_{\, \clubsuit} \text{ on average.}
\ee

\subsubsection[The $\delta$-smoothed maximum likelihood decoder]{The $\boldsymbol \delta$-smoothed maximum likelihood decoder}\label{sec:smoothed}

The aim of this section is to define a regularisation of the maximum likelihood decoder, and to prove that the supremum of the reliable communication rates that can be achieved with regularised codes is continuous as the smoothing parameter $\delta$ goes to zero. As a consequence, since for $\delta=0$ we retrieve the ordinary maximum likelihood decoder, the capacity of the channel can asymptotically be achieved by the regularised sequence of codes by taking the smoothing parameter to be arbitrarily small.
\begin{Def}
    Let $\pazocal{E}:[M]\to \mathcal{X}^n$ be any fixed deterministic encoder, and let \errata{$\pazocal{E}(m)= x^n_m=(x_{m}^{(1)},\dots, x_{m}^{(n)}$)}. For any $\delta\geq0$, we define the $\delta$\emph{-smoothed maximum likelihood decoder} for $\pazocal{W}^{\times n}$ corresponding to the encoding scheme $\pazocal{E}$ as 
    \bb\label{eq:arg_max}
        \pazocal{D}(y^n)\coloneqq \arg\max_{m\in[M]} \prod_{i=1}^{n}\pazocal{W}_\delta\big(y^{(i)}|x_m^{(i)}\big)
    \ee
    for all $y^n=(y^{(1)},\dots, y^{(n)})\in\mathcal{Y}^n$, where $\pazocal{W}_\delta(y|x)\coloneqq \pazocal{W}(y|x)+\delta$. If the maximiser in~\eqref{eq:arg_max} is not unique, then we arbitrarily define $\pazocal{D}(y^n)$ as any of those maximisers.
\end{Def}
The particular case $\delta=0$ corresponds to the ordinary maximum likelihood decoder. Choosing any discrete memoryless channel $\pazocal{W}$, we are now going to prove that the largest communication rate (with asymptotically vanishing error probability) that can be achieved using the $\delta$-smoothed maximum likelihood decoder gets arbitrarily close to the capacity of $C(\pazocal{W})$ as $\delta$ goes to zero.\\

Given $n\geq 1$ and a decoder $\pazocal{D}:\mathcal{Y}^n\to [M]$, we can partition $\mathcal{Y}^n$ into $M$ disjoint subsets $\{Y_m\}_{m\in [M]}$, called \emph{decoding regions}, whose items are the output sequences that are decoded according to the label of the subset, namely
\bb
    Y_m\coloneqq \pazocal{D}^{-1}(m)
\ee
and
\bb
    \mathcal{Y}^n=\bigsqcup_{m\in[M]}Y_m.
\ee

\begin{lemma}\label{lem:R_delta} Let $\pazocal{W}$ be a discrete memoryless channel from $\mathcal{X}$ to $\mathcal{Y}$. Let $r>0$. For every $n\geq 1$, we call $M_n\coloneqq \lceil 2^{rn}\rceil$. Then, for any fixed $\delta\geq 0$, there exists an sequence of encoders $(\pazocal{E}^\ast_n)_n$ such that
\bb\label{eq:R_delta}
        \lim_{\delta\to 0}R_\delta= C(\pazocal{W}),
    \ee
where
\begin{itemize}
    \item $R_\delta$ is the smallest rate at which the reliability function
\bb
    E_\delta(r,\pazocal{W})\coloneqq \liminf_{n\to \infty} -\frac{1}{n}\log\Bigg(\frac{1}{ M_n}\sum_{m=1}^{M_n} \mathbb{P}\big(m\neq \pazocal{D}_n\circ\pazocal{W}^{\times n}\circ\pazocal{E}^\ast_n(m)\big)\Bigg)
\ee
 for the sequence of codes $\pazocal{C}^\ast_n=(\pazocal{E}^\ast_n,\pazocal{D}_n)$ vanishes, i.e. $R_\delta\coloneqq \min\{r\geq 0: E_\delta(r,\pazocal{W})=0\}$;
 \item  $\pazocal{D}_n$ is the $\delta$-smoothed maximum likelihood decoder for $\pazocal{W}^{\times n}$ with encoder $\pazocal{E}^\ast_n$.
\end{itemize}

\end{lemma}

\begin{proof}
    The core idea of the proof is to adapt the strategy of the random coding bound to the reliability function. Let us fix \errata{$\delta> 0$} and $n\geq 1$; then, let $\{\pazocal{C}_n\}=\{(\pazocal{E}_n,\pazocal{D}_n)\}$ be the family of all encoding functions $\pazocal{E}_n:[M]\to \mathcal{X}^n$ together with their respective $\delta$-smoothed maximum likelihood decoder. Given any arbitrary $P_X\in\mathcal{P}(\mathcal{X})$, we define a corresponding probability distribution $Q_P$ on $\pazocal{C}_n=(\pazocal{E}_n,\pazocal{D}_n)$ given by
    \bb\label{eq:above}
        Q_P(\pazocal{C}_n)\coloneqq \prod_{m\in[M]}P_X^{\times n}\big(\pazocal{E}(m)\big)=\prod_{\substack{m\in[M]\\1\leq i \leq n}}P_X(x_m^{(i)}),
    \ee
    where $x_m^n=(x_m^{(1)},\dots, x_m^{(n)})\coloneqq\pazocal{E}_n(m)$. For every code $\pazocal{C}_n$, we call $\{Y_m^{\pazocal{C}_n}\}_{m\in [M]}$ the partition of $\mathcal{Y}^n$ into decoding regions according to the $\delta$-smoothed maximum likelihood decoder $\pazocal{D}_n$ for $\pazocal{W}^{\times n}$ associated with the encoding $\pazocal{E}_n$, namely $Y_m\coloneqq \pazocal{D}_n^{-1}(m)$ for all $m\in [M]$. When transmitting a message $m$ with a code $\pazocal{C}_n$ using a codeword $x_m^n=\pazocal{E}_n(m)$, the probability $p_{m\to m'|\,y^n,\,\pazocal{C}_n}$ of decoding $m'\neq m$ after observing $y^n$ as the output of the channel $\pazocal{W}^{\times n}$ is either $0$ (if $y^n\notin Y_{m'}^{\pazocal{C}_n}$) or $1$ (if $y^n\in Y_{m'}^{\pazocal{C}_n}$), since the decoder is a deterministic function. Hence, we can upper bound
        \bb\label{eq:upper_bound_MLD}
        p_{m\to m'\,|\,y^n,\,\pazocal{C}_n}=\errata{\id_{Y_{m'}^{\pazocal{C}_n}}}(y^n)\leq \left(\frac{\pazocal{W}_\delta^{\times n}(y^{n}|x_{m'}^n)}{\pazocal{W}_\delta^{\times n}(y^{n}|x_{m}^n)}\right)^s,
    \ee
    where $\id_S(\,\cdot\,)$ denotes the characteristic function of the set $S$, and the last upper bound holds for any arbitrary $s\geq 0$. Fo our purposes, it is sufficient to consider $s\in[0,1]$. The last inequality is a simple consequence of the very definition of the $\delta$-smoothed maximum likelihood decoder. The overall error probability for a fixed code when sending the message $m$ with a code $\pazocal{C}_n$ is
    \bb\label{eq:p_err_iid}
        p_{m\to {\rm err}\,|\,\pazocal{C}_n}\coloneqq \sum_{y^n}\pazocal{W}^{\times n}(y^n|x^n_m)p_{m\to {\rm err} \,|\,x_m^n,y^n}
    \ee
    where
    \bb
        p_{m\to {\rm err} \,|\,y^n,\, \pazocal{C}_n}\coloneqq \sum_{m'\neq m}p_{m\to m' |\,y^n,\, \pazocal{C}_n}
    \ee
    is the total error probability when transmitting a message $m$ and observing $y^n$ after the use of the channel. Its expectation value over the family of codes $\pazocal{C}_n$ and conditioned on the particular encoding $m\mapsto x^n_m=\pazocal{E}_n(m)$ of the message $m$ is upper bounded as
   \bb\label{eq:average}
        \EE{\pazocal{C}_n\sim Q_P}\big[p_{m\to {\rm err}\,|\, x_m^n,y^n}\,\big|\,\pazocal{E}_n(m)=x_m^n\big]&=\sum_{m'\neq m}\EE{\pazocal{C}_n\sim Q_P}\big[p_{m\to m'\,|\,x_m^n,y^n}\,\big|\,\pazocal{E}_n(m)=x_m^n\big]\\
        &\leqt{(i)} \sum_{m'\neq m}\underbrace{\sum_{x_{m'}^n\coloneqq \pazocal{E}_n(m')}P^{\times n}\big(x_{m'}^n\big)\left(\frac{\pazocal{W}_\delta^{\times n}(y^{n}|x_{m'}^n)}{\pazocal{W}_\delta^{\times n}(y^{n}|x_{m}^n)}\right)^s}_{\rm (a)}\\[0.5em]
        &\eqt{(ii)}(M-1)\sum_{x^n}P^{\times n}\big(x^n\big)\left(\frac{\pazocal{W}_\delta^{\times n}(y^{n}|x^n)}{\pazocal{W}_\delta^{\times n}(y^{n}|x_{m}^n)}\right)^s,
    \ee
    where in (i) we first have leveraged the upper bound~\eqref{eq:upper_bound_MLD} and then used the particular form of $Q_P$ as in~\eqref{eq:above}; in (ii) we have noticed that (a) does not depend on the particular message $m'$, but only on its encoding $x_{m'}^n$. Hence, the overall error probability when transmitting $m$ averaged \errata{over} the family of codes $\{\pazocal{C}_n\}$ is
    \begin{align*}
        \EE{\pazocal{C}_n\sim Q_P}\big[p_{m\to {\rm err}}\big]&=\sum_{y^n}\sum_{x^n_m}\PP{\pazocal{C}_n\sim Q_P}\big(\pazocal{E}_n(m)=x_m^n\big)\pazocal{W}^{\times n}(y^n|x^n_m)\underbrace{\EE{\pazocal{C}_n\sim Q_P}\big[p_{m\to {\rm err} \,|\,x_m^n,y^n}\,\big|\,\pazocal{E}_n(m)=x_m^n\big]}_{\eqcolon \xi}\\
        &\leqt{(iii)}\sum_{y^n}\sum_{x^n_m}P^{\times n}(x_m^n)\pazocal{W}^{\times n}(y^n|x^n_m)\left(\EE{\pazocal{E}_n\sim Q_P}\big[p_{m\to {\rm err}\,|\, x_m^n,y^n}\,\big|\,\pazocal{E}_n(m)=x_m^n\big]\right)^{\rho}\\
        &\leqt{(iv)} (M-1)^{\rho}\sum_{y^n}\left(\sum_{x^n_m}P^{\times n}(x_m^n)\pazocal{W}^{\times n}(y^n|x^n_m)\pazocal{W}_\delta^{\times n}(y^{n}|x^n_m)^{-\rho s}\right)\\
        &\qquad\qquad \qquad \qquad\times\left(\sum_{x^n}P^{\times n}\big(x^n\big)\pazocal{W}_\delta^{\times n}(y^{n}|x^n)^s\right)^\rho\\
        &\leqt{(v)} (M-1)^{\rho}\sum_{y^n}\left(\sum_{x^n_m}P^{\times n}(x_m^n)\pazocal{W}_\delta^{\times n}(y^n|x^n_m)^{\errata{1-\rho s}}\right)\left(\sum_{x^n}P^{\times n}\big(x^n\big)\pazocal{W}_\delta^{\times n}(y^{n}|x^n)^s\right)^\rho\\
        &\eqt{(vi)} (M-1)^{\rho}\sum_{y^n}\left(\sum_{x^n}P^{\times n}\big(x^n\big)\pazocal{W}_\delta(y^{n}|x^n)^{\frac{1}{1+\rho}}\right)^{1+\rho},
    \end{align*}
    where in
    (iii) we have used the elementary inequality $\xi\leq \xi^\rho$ for any arbitrary $0\leq\xi\leq 1$ and $0\leq \rho\leq 1$, with the usual convention $0^0=1$; in
    (iv) we have used~\eqref{eq:average}; in
    (v) we have noticed that $\pazocal{W}(y|x)\leq \pazocal{W}_\delta(y|x)$ for all $x,y$; in
    (vi) we have chosen $s=\frac{1}{1+\rho}\in[0,1]$.
    By the product structure of the upper bound, and setting $M=\lceil 2^{rn}\rceil$, we get
    \bb\label{eq:last}
    	\EE{\pazocal{C}_n\sim Q_P}\big[p_{m\to {\rm err}}\big]&\leq
    	2^{\rho r n}\prod_{i=1}^n\sum_{y^{(i)}}\left(\sum_{x^{(i)}}P\big(x^{(i)}\big){\errata{\pazocal{W}_\delta}}\big(y^{(i)}\big|x^{(i)})\big)^{\frac{1}{1+\rho}}\right)^{1+\rho}\\
    	&= 2^{\rho r n}\left(\sum_{y}\left(\sum_{x}P(x)\pazocal{W}_\delta(y|x)^{\frac{1}{1+\rho}}\right)^{1+\rho}\right)^n.
    \ee
    By the previous upper bound, there exists at least one code $\pazocal{C}_n^\ast=(\pazocal{E}_n^\ast,\pazocal{D}_n)$ which has (average) error probability
    \bb\label{eq:random_encoder}
    	p_{{\rm err},\delta}^\ast =\frac 1M\sum_{m=1}^Mp_{m\to {\rm err}}^\ast&\leq 2^{\rho r n}\left(\sum_{y}\left(\sum_{x}P(x){\errata{\pazocal{W}_\delta}}(y|x)^{\frac{1}{1+\rho}}\right)^{1+\rho}\right)^n\eqqcolon \bar p_{{\rm err},\delta}^\ast.
    \ee
    Therefore, given any rate $r>0$, the exponent $E_\delta(r,\pazocal W)$ at which the error probability of the sequence of codes $(\pazocal E_n^\ast, \pazocal D_n)$ decays is lower bounded by
    \bb\label{eq:modified_gallager}
    	E_\delta(r,\pazocal W)\geq \bar E_\delta(r,\pazocal W)\coloneqq \max_{0\leq \rho\leq 1}\max_{P_X}\Bigg(\underbrace{-\log \sum_{y}\left(\sum_{x}P(x)\pazocal{W}_\delta(y|x)^{\frac{1}{1+\rho}}\right)^{1+\rho}}_{\eqqcolon\; E_G(\rho,P,\delta)}\;-\;\rho r\Bigg),
    \ee
    where we call $E_G(\rho,P,\delta)$ the \emph{modified Gallager function}. Defining, for $\delta\geq 0 $,
    \bb
        R_\delta&\coloneqq \min\{r\geq 0: E_\delta(r,\pazocal{W})=0\},\\
        \bar R_\delta&\coloneqq \errata{\min\{r\geq 0: \bar E_\delta(r,\pazocal{W})\leq0\},}
    \ee
    we have
    \bb\label{eq:ineq_RR}
        \bar R_\delta\leq R_\delta \leq C(\pazocal{W})=R_0,
    \ee
    where the last inequality follows from the sphere packing bound, which coincides with the ordinary random coding exponent (i.e. $\delta=0$) for rates above a critical value. 
\errata{
Now, since \(\pazocal W_\delta(y|x)\to \pazocal W(y|x)\) as  \(\delta\to0\) for every \(x,y\), the modified Gallager function converges pointwise:
\[
E_G(\rho,P,\delta)\to E_G(\rho,P,0) \quad \text{as} \quad \delta \to 0.
\]
Moreover, since the alphabets are finite and the parameter sets \([0,1]\) and \(\mathcal P(\mathcal X)\) are compact, this convergence passes through the maximization, yielding
\[
\lim_{\delta\to0}\bar E_\delta(r,\pazocal W)=\bar E_0(r,\pazocal W).
\]

Let \(L\coloneqq\lim_{\delta\to0}\bar R_\delta\). If \(r>L\), then for all sufficiently small \(\delta>0\) one has \(r>\bar R_\delta\), and hence \(\bar E_\delta(r,\pazocal W)\le0\). Passing to the limit \(\delta\to0\) gives \(\bar E_0(r,\pazocal W)\le0\), which implies \(r\ge R_0\). Since this holds for every \(r>L\), we conclude that \(L\ge R_0\).
Combined with the previously established inequality \(L\le R_0\), this yields \(L=R_0\).}
\end{proof}

\subsubsection{The protocol for achievability}

\begin{prop}\label{prop:achiev_cl} Let $\pazocal{W}$ be a classical channel and let $\tilde{\pazocal{W}}=\big(\tilde{\pazocal{W}}^{(n)}\big)_n$ be an almost i.i.d.\ process along $\pazocal{W}$. Then,
    \bb
        C(\tilde{\pazocal{W}})\geq C(\pazocal{W}),
    \ee
    where the sequence of codes used to achieve the communication rate $C(\pazocal{W})$ via $\tilde{\pazocal{W}}$ only depends on the channel 
    $\pazocal{W}$, not on the particular almost i.i.d.\ process $\tilde{\pazocal{W}}$.
\end{prop}

\begin{proof}
    Let us fix $\delta>0$. For every $n\geq 1$, choosing any arbitrary encoder $\pazocal{E}_n:[M]\to \mathcal{X}^n$ let us consider the code $(\pazocal{E}_n,\pazocal{D}_n)$, where $\pazocal{D}_n$ is the $\delta$-smoothed maximum likelihood decoder for $\pazocal{W}^{\times n}$ with encoding given by $\pazocal{E}_n$.
    We use the code $(\pazocal{E}_n, \pazocal{D}_n)$ in order to communicate with the almost i.i.d.\ channel $\tilde{\pazocal{W}}^{(n)}$. Let us call $\tilde p_{m\to {\rm err}\,|\,x_m^n,\tilde y^n}$ the probability of decoding any $m'\neq m$ when the input message is $m$, assuming that it gets encoded in $x_m^n=\pazocal{E}_n(m)$ and it gets transformed into $\tilde y^n$ by the channel $\tilde{\pazocal{W}}$:
    \begin{align}\nonumber
        \tilde p_{m\to m'\,|\,x_m^n,\tilde y^n}=\id_{Y_{m'}}(\tilde y^n)
        &\leqt{(i)}\left(\frac{\pazocal{W}_\delta^{\times n}( \tilde y^{n}|x_{m'}^n)}{\pazocal{W}_\delta^{\times n}(\tilde  y^{n}|x_{m}^n)}\right)^s\\ \nonumber
        &=
        \left(\frac{\pazocal{W}_\delta^{\times n}( y^{n}|x_{m'}^n)}{\pazocal{W}_\delta^{\times n}( y^{n}|x_{m}^n)}\right)^s\times \left(\prod_{i:y_i\neq \tilde y_i}\frac{\pazocal{W}_\delta(\tilde y^{(i)}|x_{m'}^{(i)})}{\pazocal{W}_\delta( y^{(i)}|{x_{m'}^{(i)}}
        )}\cdot \frac{\pazocal{W}_\delta( y^{(i)}|x_{m}^{(i)})}{\pazocal{W}_\delta(\tilde y^{(i)}|x_{m}^{(i)})}\right)^s\\ \nonumber
        &\leqt{(ii)} \left(\frac{\pazocal{W}_\delta^{\times n}( y^{n}|x_{m'}^n)}{\pazocal{W}_\delta^{\times n}( y^{n}|x_{m}^n)}\right)^s \times \left(\frac{(1+\delta)^2}{\delta^2}\right)^{sd_H(\tilde y,y)}\\ \label{eq:p_upper}
        &\leq \left(\frac{\pazocal{W}_\delta^{\times n}( y^{n}|x_{m'}^n)}{\pazocal{W}_\delta^{\times n}( y^{n}|x_{m}^n)}\right)^s e^{2d_H(\tilde y,y)\log\frac{1+\delta}{\delta}},
    \end{align}
    where $y^n$ is any arbitrary sequence in $\mathcal{Y}^n$, and $s\in [0,1]$. In particular, in (i) we have noticed that, for the $\delta$-smoothed maximum likelihood decoder for $\pazocal{W}^{\times n}$ -- we stress that it is designed for $\pazocal{W}^{\times n}$, not for the particular almost i.i.d.\ channel $\tilde{\pazocal{W}}^{(n)}$ -- we can proceed as in~\eqref{eq:above}, and in (ii) we have used the bound $\delta\leq \pazocal{W}_\delta\leq 1+\delta$. 
    \errata{For every input string \(x^n\in\mathcal X^n\), let us denote by \(\pi_{x^n}^\ast\) the optimal coupling between the output distributions \(\tilde{\pazocal W}^{(n)}(\,\cdot\,|x^n)\) and \(\pazocal W^{\times n}(\,\cdot\,|x^n)\) achieving the club distance. The corresponding conditional distribution \(\Phi^\ast_{x^n}(\tilde y^n|y^n)\) is then defined from this coupling via \(\pi^\ast_{x^n}(\tilde y^n,y^n)=\Phi^\ast_{x^n}(\tilde y^n|y^n)\pazocal W^{\times n}(y^n|x^n)\) as in~\eqref{eq:phi_def}.}
    Calling $p_{m\to {\rm err}\,|\,x_m^n,\tilde y^n}\coloneqq\sum_{m'\neq m}\tilde p_{m\to m'\,|\,x_m^n,\tilde y^n}$, the overall error probability when sending the message $m$ is 
    \begin{align*}
        &\EE{\pazocal{E}_n\sim Q_P}\big[\tilde p_{m\to {\rm err}}\big]\\
        &\quad= \sum_{\tilde y^n}\sum_{x^n_m}\PP{\pazocal{E}_n\sim Q_P}\big(\pazocal{E}_n(m)=x_m^n\big)\tilde{ \pazocal{W}}^{(n)}(\tilde y^n|x^n_m)\EE{\pazocal{E}_n\sim Q_P}\big[\tilde p_{m\to {\rm err}\,|\,x_m^n,\tilde y^n}\,\big|\,\pazocal{E}_n(m)=x_m^n\big]\\
        &\quad\eqt{(iii)} \sum_{\tilde y^n,y^n}\sum_{x^n_m}P^{\times n}(x_m^n)\;\Phi^\ast_{x^n_m}(\tilde y^n|y^n)\pazocal{W}^{\times n}(y^n\,|x_m^n)\underbrace{\EE{\pazocal{E}_n\sim Q_P}\big[\tilde p_{m\to {\rm err}\,|\,x_m^n,\tilde y^n}\,\big|\,\pazocal{E}_n(m)=x_m^n\big]}_{\eqcolon \xi\in[0,1]}\\
        &\quad\leqt{(iv)} \sum_{d_H(\tilde y^n,y^n)\leq n\eta}\sum_{x^n_m}P^{\times n}(x_m^n)\;\Phi^\ast_{x^n_m}(\tilde y^n|y^n)\pazocal{W}^{\times n}(y^n\,|x_m^n)\Big(\EE{\pazocal{E}_n\sim Q_P}\big[\tilde p_{m\to {\rm err}\,|\,x_m^n,\tilde y^n}\,\big|\,\pazocal{E}_n(m)=x_m^n\big]\Big)^\rho\\
        &\quad\qquad + \PP{(\tilde Y^n,Y^n)\sim \pi^\ast_{x^n_m}}\big(d_H(\tilde Y^n,Y^n)>n\eta\big)\\
        &\quad\leqt{(v)} (M-1)^\rho\sum_{d_H(\tilde y^n,y^n)\leq n\eta}\sum_{x^n_m}P^{\times n}(x_m^n)\;\Phi^\ast_{x^n_m}(\tilde y^n|y^n) e^{2\rho\,d_H(\tilde y,y)\log \frac{1+\delta}{\delta}}\\
        &\quad\phantom{\leqt{(v)} (M-1)^\rho\sum_{d_H(\tilde y^n,y^n)\leq n\eta}\sum_{x^n_m,x^n}}\quad \times 
        \pazocal{W}^{\times n}(y^n\,|x_m^n)\left(\sum_{x^n}P^{\times n}(x^n)\left(\frac{\pazocal{W}_\delta^{\times n}( y^{n}|x^n)}{\pazocal{W}_\delta^{\times n}( y^{n}|x_{m}^n)}\right)^{s}\right)^{\rho} + p_\eta\\
        &\quad\leqt{(vi)} e^{2\rho n\eta \log \frac{1+\delta}{\delta}}(M-1)^\rho \\
        &\quad\qquad\times \sum_{y^n}\sum_{x^n_m}P^{\times n}(x_m^n)\Big(\sum_{\tilde y^n} \Phi^\ast_{x^n_m}(\tilde y^n|y^n)\Big)\pazocal{W}_\delta^{\times n}(y^n\,|x_m^n)^{1-\rho s}\left(\sum_{x^n}P^{\times n}(x^n)\pazocal{W}_\delta^{\times n}( y^{n}|x^n)^{s}\right)^\rho  + p_\eta\\
        &\quad=e^{2\rho n\eta \log\frac{1+\delta}{\delta}}(M-1)^\rho \exp\big(-nE_G(\rho,P,\delta)\big)  + p_\eta\\
    \end{align*}
    where
    \bb
        \errata{p_\eta = \max_{x^n}\PP{(\tilde Y^n,Y^n)\sim \pi^\ast_{x^n}}\big(d_H(\tilde Y^n,Y^n)>n\eta\big).}
    \ee
    \errata{In particular,} in (iii) we have used~\eqref{eq:interpr_noise}, in (iv) we have leveraged two elementary upper bounds 
    \bb
        \xi \leq \begin{cases}
            \xi^{\rho}  \quad \text{when summing on $\tilde y^n$ and $y^n$ such that} & d_H(\tilde y^n,y^n)\leq n\eta,\\
            1  \qquad ''\qquad\qquad  ''\qquad\qquad ''\qquad\qquad  '' & d_H(\tilde y^n,y^n)> n\eta,
        \end{cases}
    \ee
    where $\rho$ is an arbitrary parameter between $0$ and $1$; in (v) we have used the upper bound~\eqref{eq:p_upper} and we have proceeded similarly to~\eqref{eq:average}; finally, in (vi) we have relaxed the summation on $\tilde y^n,y^n$ such that $d_H(\tilde y^n,y^n)\leq n\eta $ to a sum on all $\tilde y^n$ and $y^n$; the last equality is analogous to~\eqref{eq:last}. The probabilistic argument over all the encoders $\{\pazocal{E}_n\}$ used in Lemma~\ref{lem:R_delta} can be identically applied here, ensuring the existence of a \emph{universal} encoder $\pazocal{E}_n^\ast$ -- i.e. $\pazocal{E}_n^\ast$ does only depend on $\pazocal{W}^{\times n}$, not on the specific almost i.i.d.\ channel $\tilde{\pazocal{W}}^{(n)}$ -- which satisfies
    \bb
        (M-1)^\rho\sum_{y^n}\pazocal{W}^{\times n}(y^n\,|x^n)\left(\frac{\pazocal{W}_\delta^{\times n}( y^{n}|x_{m'}^n)}{\pazocal{W}_\delta^{\times n}( y^{n}|x_{m}^n)}\right)^s\leq \bar p_{\rm err}^\ast,
    \ee
    where $ \bar p_{\rm err}^\ast$ is the upper bound on the error probability for $\pazocal{W}^{\times n}$ given in~\eqref{eq:random_encoder}. Hence, the overall error probability $\tilde p_{{\rm err},\delta}^\ast$ when communicating $M=\lceil 2^{rn}\rceil$ messages over $\tilde{\pazocal{W}}^{(n)}$ with the code $(\pazocal{E}_n^\ast,\pazocal{D}_n)$ is upper bounded as
    \bb\label{eq:p_tilde}
        \tilde p_{{\rm err},\delta}^\ast\leq \underbrace{e^{2\rho n\eta\log \frac{1+\delta}{\delta}}\bar p_{{\rm err},\delta}^\ast}_{\rm (a)}+\underbrace{p_\eta}_{\rm (b)}.
    \ee
    By Markov's inequality, we can upper bound (b) as
    \bb
        p_\eta &=\errata{\max_{x^n}} \PP{(\tilde Y^n,Y^n)\sim \pi^\ast_{x^n}}\big(d_H(\tilde Y^n,Y^n)>n\eta\big) \\
        &\leq \frac{1}{\eta}\cdot\frac{1}{n}\errata{\max_{x^n}}\EE{(\tilde Y^n,Y^n)\sim \pi^\ast_{x^n}}\big[d_H(\tilde Y^n,Y^n)\big]\leq \frac{1}{\eta}\cdot\frac{1}{n}\big\|\tilde{\pazocal{W}}^{(n)}-\pazocal{W}^{\times n}\big\|_{\,\clubsuit}.
    \ee
    By choosing $\eta=\sqrt{\tfrac{1}{n}\big\|\tilde{\pazocal{W}}^{(n)}-\pazocal{W}^{\times n}\big\|_{\,\clubsuit}}$, the asymptotic decay rate of (a) is
    \bb
        \liminf_{n\to\infty}-\frac{1}{n}\log\Big(e^{2\rho n\eta\log \frac{1+\delta}{\delta}}\bar p_{{\rm err},\delta}^\ast\Big)&=\bar E_\delta(r,\pazocal{W})-2\log\tfrac{1+\delta}{\delta}\lim_{n\to\infty}\sqrt{\tfrac{1}{n}\big\|\tilde{\pazocal{W}}^{(n)}-\pazocal{W}^{\times n}\big\|_{\,\clubsuit}}\\
        &= \bar E_\delta(r,\pazocal{W}).
    \ee
    Since for all \errata{$r<\bar R_\delta\coloneqq \min\{r\geq 0: \bar E_\delta(r,\pazocal{W})\leq 0\}$} both (a) and (b) are vanishing asymptotically in $n$, we conclude that
    \bb
        C(\tilde {\pazocal{W}})\geq \bar R_\delta.
    \ee
    Since this holds for any arbitrary $\delta>0$, by Lemma~\ref{lem:R_delta} -- more precisely, by~\eqref{eq:R_delta} -- we get
    \bb
        C(\tilde {\pazocal{W}})\geq \lim_{\delta\to 0^+}\bar R_\delta=C(\pazocal{W}),
    \ee
    which completes the proof. 
\end{proof}

\errata{
\begin{rem} 
The random-coding argument used in the proof of Lemma~\ref{lem:R_delta} ensures the existence of a deterministic sequence of encoders \(\{\pazocal E_n^\ast\}_n\), together with the associated \(\delta\)-smoothed maximum-likelihood decoders \(\{\pazocal D_n\}_n\), depending only on the reference channel \(\pazocal W\), such that
$
\bar p_{{\rm err},\delta}^\ast
$
satisfies the bound in~\eqref{eq:random_encoder}.
Importantly, the encoder and decoder depend only on the reference channel \(\pazocal W\). The particular almost i.i.d.\ process \(\tilde{\pazocal W}\) enters the error estimate only through the term \(p_\eta\), which is controlled by the club distance between \(\tilde{\pazocal W}^{(n)}\) and \(\pazocal W^{\times n}\). Since \(\tilde{\pazocal W}\) is almost i.i.d.\ along \(\pazocal W\), this term vanishes asymptotically uniformly along the chosen sequence of codes.
\end{rem}
}

\subsection{Achievability in the quantum case}\label{sec:achivability_quantum}

\errata{The quantum Wasserstein distance of order 1 depends on a choice of the partition of the global quantum system into subsystems.
In this section, we will consider cases where such a choice is not unique.
Indeed, let $n = k\,r$, and let $\mathcal{H}_n=\mathcal{H}^{\otimes n}$ be the Hilbert space of the quantum system made by $n$ copies of the quantum system with Hilbert space $\mathcal{H}$.
Let us associate each copy with an element of the set $[n]$.
Then, each partition of $[n]$ corresponds to a partition of the global system into subsystems.
We denote with $W_1^n$ the $W_1$ distance on $\mathcal{D}(\mathcal{H}_n)$ associated with the partition of $[n]$ into $n$ subsystems of one element each:
\bb
    \|\rho-\sigma\|_{W_1^n} \coloneqq \min \left\{\sum_{i=1}^n c_i\right.\quad  &\text{such that}\quad  && c_i\geq 0, \qquad  \rho-\sigma=\sum_{i=1}^n c_i\left(\tau^{(i)}-\eta^{(i)}\right)\\
    & \text{with}&& \left.\tau^{(i)},\eta^{(i)}\in \mathcal{D}(\mathcal{H}_n), \quad \Tr_i\tau^{(i)}=\Tr_i\eta^{(i)}\right\}.
\ee
We denote with $W_1^r$ the $W_1$ distance on $\mathcal{D}(\mathcal{H}_{kr})$ associated to the partition of $[k\,r]$ into the subsets
    \bb 
    S_1=\{1,\,\ldots,\,k\} \quad \cdots \quad S_r = \{kr - k + 1,\,\ldots,kr\}
    \ee
of $k$ elements each:
\bb
    \|\rho-\sigma\|_{W_1^{kr}} \coloneqq \min \left\{\sum_{j=1}^r c_j\right.\quad  &\text{such that}\quad  && c_j\geq 0, \qquad  \rho-\sigma=\sum_{j=1}^r c_j\left(\tau^{(j)}-\eta^{(j)}\right)\\
    & \text{with}&& \left.\tau^{(j)},\eta^{(j)}\in \mathcal{D}(\mathcal{H}_n), \quad \Tr_{S_j}\tau^{(j)}=\Tr_{S_j}\eta^{(j)}\right\}.
\ee
We stress that the superscript of $W_1$ always denotes the number of subsystems of the partition.

\errata{
\begin{lemma}\label{lem:kr}
The norms $W_1^r$ and $W_1^{kr}$ satisfy
    \begin{equation}
        \left\|\,\cdot\,\right\|_{W_1^r} \le \left\|\,\cdot\,\right\|_{W_1^{kr}} \le 2\,k\left\|\,\cdot\,\right\|_{W_1^r}\,.
    \end{equation}
\end{lemma}

\begin{proof}
    The unit ball of $\left\|\,\cdot\,\right\|_{W_1^{kr}}$ is the convex hull of the differences between states that differ in a single $A$ system, and is contained in the unit ball of $\left\|\,\cdot\,\right\|_{W_1^{r}}$, which is the convex hull of the differences between states that differ in a single block of $k$ copies of $A$.
    Therefore, $\left\|\,\cdot\,\right\|_{W_1^r} \le \left\|\,\cdot\,\right\|_{W_1^{kr}}$.

    Let $X$ be a self-adjoint traceless operator acting on $\mathcal{H}_A^{\otimes kr}$, and let
    \begin{equation}
        X = \sum_{i=1}^r X^{(i)}\,,\qquad \mathrm{Tr}_{S_i}X^{(i)} = 0\,,\qquad \left\|X\right\|_{W_1^r} = \frac{1}{2}\sum_{i=1}^r\left\|X^{(i)}\right\|_1\,.
    \end{equation}
    Let $\omega$ be the maximally mixed state on one copy of $A$, and let
    \begin{equation}
        X^{(1,1)} = X^{(1)} - \omega\otimes\mathrm{Tr}_1 X^{(1)}\quad \cdots\quad X^{(1,k)} = \omega^{\otimes\left(k-1\right)}\otimes\mathrm{Tr}_{1\ldots k-1}X^{(1)} - \omega^{\otimes k}\otimes\mathrm{Tr}_{1\ldots k}X^{(1)}\,,
    \end{equation}
    such that
    \begin{equation}
        X^{(1)} = X^{(1,1)} + \ldots + X^{(1,k)}\,,\qquad \mathrm{Tr}_j X^{(1,j)} = 0\,,\qquad \left\|X^{(1,j)}\right\|_1 \le 2\left\|X^{(1)}\right\|_1\,.
    \end{equation}
    Analogously, for any $i\in[r]$ and any $j\in[k]$ we build $X^{(i,j)}$ such that
    \begin{equation}
        X^{(i)} = X^{(i,1)} + \ldots + X^{(i,k)}\,,\qquad \mathrm{Tr}_j X^{(i,j)} = 0\,,\qquad \left\|X^{(i,j)}\right\|_1 \le 2\left\|X^{(i)}\right\|_1\,.
    \end{equation}
    We then have
    \begin{equation}
        X = \sum_{i=1}^r\sum_{j=1}^k X^{(i,j)}\,,\qquad \left\|X\right\|_{W_1^{kr}} \le \frac{1}{2}\sum_{i=1}^r\sum_{j=1}^k \left\|X^{(i,j)}\right\|_1 \le k\sum_{i=1}^r\left\|X^{(i)}\right\|_1 = 2\,k\left\|X\right\|_{W_1^r}\,.
    \end{equation}
    The claim follows.
\end{proof}
}

Similarly, we will use a more explicit notation to emphasise the partition underlying the output system appearing in the the club norm.
Namely, we denote with $\left\|\cdot\right\|_{\,\clubsuit, r}$ the norm obtained employing the norm $W_1^r$ in Definition~\ref{def:club_norm}:
\bb
    \left\|\Delta\Phi\right\|_{\,\clubsuit, r}\coloneqq \sup_{\rho}\big\|\Delta\Phi(\rho)\big\|_{W_1^r}.
\ee
}

\subsubsection{Extension to the classical-quantum case}
\errata{
\begin{prop}\label{prop:achiev_cq}
Let $\pazocal M$ be a classical-quantum channel, and let
$
\tilde{\pazocal M}
=
(\tilde{\pazocal M}^{(n)})_n
$
be an almost i.i.d.\ process along $\pazocal M$. Then,
\bb
C(\tilde{\pazocal M})
\ge
C(\pazocal M).
\ee
Moreover, for every rate $r<C(\pazocal M)$, there exists a sequence of codes achieving the rate $r$ over $\tilde{\pazocal M}$ that depends only on the reference channel $\pazocal M$, and not on the particular almost i.i.d.\ process $\tilde{\pazocal M}$.
\end{prop}}
\begin{proof}
Let us consider an almost i.i.d.\ classical-quantum process $\tilde{\pazocal{M}}=\big(\tilde{\pazocal{M}}^{(n)}\big)_{n}$ along a given channel $\pazocal{M}:\mathcal{X}\to \mathcal{D}(\mathcal{H}_B)$. \errata{Fix $k\ge1$. For each $r\ge1$, regard $kr$ uses of the channel as $r$ blocks of length $k$:}
\bb
    \mathbf{x}^r\coloneqq \begin{pmatrix}
        x_1^k\\
        \vdots\\
        x_r^k
    \end{pmatrix}\in(\mathcal{X}^{k})^{\times r},\qquad \text{where}\qquad x_i^k=(x_i^{(1)},\dots, x_i^{(k)})\in\mathcal{X}^k, \quad i=1,\dots, r.
\ee
Let
\bb
\Lambda_k:\mathcal D(\mathcal H_B^{\otimes k})\to\mathcal{P}(\mathcal Y_k)
\ee
be an arbitrary measurement channel. We then define the induced classical channel
\bb
\tilde{\pazocal W}_k^{(r)}(\mathbf x^r)
\coloneqq
\Lambda_k^{\otimes r}
\big(
\tilde{\pazocal M}^{(kr)}(\mathbf x^r)
\big).
\ee
The corresponding i.i.d.\ reference channel is
\bb
\pazocal W_k
\coloneqq
\Lambda_k\circ \pazocal M^{\otimes k}.
\ee
Then, $\big(\tilde{\pazocal{W}}_k^{(r)}\big)_{r\geq 1}$ is an almost-i.i.d.\ source along
\bb
    \pazocal{W}_k\coloneqq\pazocal{D}_k\circ \pazocal{M}^{\otimes k}.
\ee
Indeed, we have
\bb\label{eq:150}
    \frac{1}{r}\big\|\tilde{\pazocal{W}}_k^{(r)}-{\pazocal{W}}_k^{\times r}\big\|_{\,\clubsuit, r}&=\max_{\mathbf{x}^r}\frac{1}{r}\big\|\tilde{\pazocal{W}}_k^{(r)}(\mathbf{x}^r)-{\pazocal{W}}_k^{\times r}(\mathbf{x}^r)\big\|_{W_1^r}\\
    &=\max_{\mathbf{x}^r}\frac{1}{r}\big\|\pazocal{D}^{\otimes r}\big(\tilde{\pazocal{M}}^{(kr)}(\mathbf{x}^r)\big)-\pazocal{D}^{\otimes r}\big(\pazocal{M}^{\otimes kr}(\mathbf{x}^r)\big)\big\|_{W_1^r}\\
    &\leq \max_{\mathbf{x}^r}\frac{1}{r}\big\|\tilde{\pazocal{M}}^{(kr)}(\mathbf{x}^r)-\pazocal{M}^{\otimes kr}(\mathbf{x}^r)\big\|_{W_1^r}\\
    &\leq \max_{x^{kr}}\frac{1}{r}\big\|\tilde{\pazocal{M}}^{(kr)}(x^{kr})-\pazocal{M}^{\otimes kr}(x^{kr})\big\|_{W_1^{kr}}\\
    &= k\times \frac{1}{kr}\big\|\tilde{\pazocal{M}}^{(kr)}-\pazocal{M}^{\otimes kr}\big\|_{\,\clubsuit,kr}=o(1) \quad \text{as} \quad r\to\infty,
\ee
\errata{where the first inequality follows since the measurement channel \(\Lambda_k^{\otimes r}\) is contractive for the \(W_1\) distance~\cite[Proposition 3]{De_Palma_2021}:
\bb
\big\|
\Lambda_k^{\otimes r}(\rho)
-
\Lambda_k^{\otimes r}(\sigma)
\big\|_{W_1^r}
\le
\|\rho-\sigma\|_{W_1^r},
\ee
and the second inequality follows since from Lemma~\ref{lem:kr} we have $\left\|\,\cdot\,\right\|_{W_1^r} \le \left\|\,\cdot\,\right\|_{W_1^{kr}}$.
}

Now, by Proposition~\ref{prop:achiev_cl}, 
\bb
\label{eq:CtildeW}
    C(\tilde{\pazocal{W}_k})\geq C(\pazocal{W}_k)=\max_{p_{X^k}}I(X^k:Y_k)=\max_{\rho_{X^kB_k}}I(X^k:Y_k)_{({\rm Id}\otimes \errata{\Lambda_k})(\rho),}
\ee
\errata{where $Y_k$ is the output of $\pazocal{W}_k$ with random input $X^k\sim p_{X^k}$, and the second maximisation is restricted to be over states of the form
$$
\rho_{X^kB_k}
=
\sum_{x^k}
p(x^k)\,
\ketbra{x^k}
\otimes
\pazocal M^{\otimes k}(x^k),
$$
induced by an input distribution \(p_{X^k}\) through the fixed channel \(\pazocal M^{\otimes k}\).} \errata{Note that any code for the induced classical channel \(\tilde{\pazocal W}_k^{(r)}\), operating on \(r\) blocks of size \(k\), naturally induces a code for the channel \(\tilde{\pazocal M}^{(kr)}\). Consequently, the corresponding communication rate is divided by a factor \(k\). Hence, for all $k\geq 1$, we have}
\bb\label{eq:note}
     C(\tilde{\pazocal{M}})\geq \frac 1k  C(\tilde{\pazocal{W}_k}).
\ee

Let $(\pazocal{E}_k,\pazocal{D}_k)$ be a sequence of $(M_k,\epsilon_k)$-codes
for $\pazocal{M}^{\otimes k}$ asymptotically achieving the capacity as communication rate:
\bb
    \pazocal{E}_k:[M_k]\to \mathcal{X}^k\qquad\pazocal{D}_k:\mathcal{D}(\mathcal{H}^{\otimes k})\to \errata{\mathcal{P}([M_k])}
\ee
and
\begin{align}
    \label{eq:rate}
    &\lim_{k\to \infty}\frac{\log M_k}{k}=C(\pazocal{M}),\\ 
    \label{eq:error}
    &\errata{\lim_{k\to\infty}\epsilon_k=\lim_{k\to\infty}\frac 1 {M_k}\sum_{m=1}^{M_k}\mathbb{P}\big(m\neq m'\sim \pazocal{D}_k\circ\pazocal{M}^{\otimes k}\circ \pazocal{E}_k(m)\big)=0.}
\end{align}
\errata{Now we want to show that, by choosing $\Lambda_k\coloneqq\pazocal{D}_k$ and combining~\eqref{eq:CtildeW} with~\eqref{eq:note}, we can obtain the desired lower bound \mbox{$C(\tilde{\pazocal{M}})\geq C(\pazocal{M})$}.}
Let $Z_k$ be a uniform random variable on $[M_k]$, and let $\hat Z_k\coloneqq \big(\pazocal{D}_k\circ \pazocal{M}^{\otimes k}\circ \pazocal{E}_k\big)(Z_k)$, which takes values in $\mathcal{Y}_k=[M_k]$; then the joint probability distribution $Q_k(m,m')=\mathbb{P}(Z_k=m,\hat Z_k=m')$ satisfies
\bb
    \lim_{k\to\infty}\frac 12\big\|Q_k-R^{(M_k)}\big\|_1&\leqt{(a)} \lim_{k\to\infty}\PP{(Z_k,\hat Z_k)\sim Q_k}\big((Z_k,\hat Z_k)\neq (Z_k,Z_k)\big)\\
    &= \lim_{k\to\infty}\PP{(Z_k,\hat Z_k)\sim Q_k}\big(\hat Z_k\neq Z_k\big)\eqt{(b)}\lim_{k\to\infty} \epsilon_k=0,
\ee
where $R^{(M_k)}(m,m')\coloneqq\frac{\delta_{m\,m'}}{M_k}$ is the maximally correlated distribution, and
\begin{itemize}
    \item in (a) we have leveraged the coupling lemma $\displaystyle{\frac 12 \|P-Q\|_1=\inf_{\substack{X\sim P\\ Y\sim Q}}\mathbb{P}(X\neq Y)}$,
    \item \vspace{-1.2em}  in (b) we have recalled~\eqref{eq:error}.
\end{itemize} 
\errata{Since \(Z_k\) is uniform on \([M_k]\) and \(\mathbb P(\hat Z_k\neq Z_k)=\epsilon_k\), Fano's inequality yields
\bb
H(Z_k|\hat Z_k)
\le
h_2(\epsilon_k)+\epsilon_k\log(M_k-1).
\ee
Hence,
\bb
\big|
I(Z_k:\hat Z_k)-\log M_k
\big|
\le
h_2(\epsilon_k)+\epsilon_k\log(M_k-1).
\ee
After division by \(k\), the right-hand side vanishes because \(\epsilon_k\to0\) and \(\log M_k/k\to C(\pazocal M)\) as $k \to \infty$:
\bb
\label{eq:IZk}
    \lim_{k\to\infty} \frac{1}{k}I(Z_k:\hat Z_k)_{Q_k}\eqt{(i)}\lim_{k\to\infty} \frac{1}{k}I(Z_k:\hat Z_k)_{R^{(M_k)}}\eqt{(ii)}\lim_{k\to\infty}\frac{\log M_k}{k}\eqt{(iii)}C(\pazocal{M}),
\ee
}

where (i) is the asymptotic continuity of the Shannon entropy, in (ii) we have recalled that the mutual information of a uniform random variable over a set is the logarithm of the cardinality of the set, and in (iii) we have used~\eqref{eq:rate}. Without loss of generality, we can assume that $\pazocal{E}_k$ is injective, i.e. $\pazocal{E}_k(m)\neq \pazocal{E}_k(m')$ whenever $m\neq m'$, \errata{since, if two messages share the same codeword, they cannot both be decoded reliably with small average error.} Choosing $p(x^k)=\mathbb{P}(\pazocal{E}(Z_k)=x^k)$, we have
\bb
    ({\rm Id}\otimes \pazocal{D}_k)(\rho_{X^kB^k})&= \sum_{x^k}p(x^k)\ketbra{x^k}\otimes (\pazocal{D}_k\circ \pazocal{M}^{\otimes k})(\ketbra {x^k})\\
    &=\sum_{m,m'\in [M_k]}\mathbb{P}\big(Z_k=m,\hat Z_k=m'\big)\ketbra{\pazocal{E}_k(m)}\otimes \ketbra{m'}.
\ee
Since $\pazocal{E}_k$ is invertible on its image, we conclude that
\bb\label{eq:IZk2}
    I(Z_k:\hat Z_k)_{Q_k}=I(X^k:Y_k)_{({\rm Id}\otimes \pazocal{D}_k)(\rho)}
\ee
\errata{For each fixed block size \(k\), Proposition~\ref{prop:achiev_cl} yields coding schemes for the induced classical process \(\tilde{\pazocal W}_k\) whose rates approach \(C(\pazocal W_k)\) in the limit \(r\to\infty\). One then sends \(k\to\infty\) by choosing a sequence of \(k\)-block codes for \(\pazocal M^{\otimes k}\) whose rates asymptotically achieve \(C(\pazocal M)\):} 
\bb
      C(\tilde{\pazocal{M}})\geq \lim_{k\to \infty}\frac 1k C(\tilde{\pazocal{W}_k})\geq C(\pazocal{M}).
\ee
This completes the proof.
\end{proof}

\subsubsection{Lifting to the fully quantum case}

\errata{Let \(\tilde{\pazocal N}=(\tilde{\pazocal N}^{(n)})_n\) be an almost i.i.d.\ process along a quantum channel \(\pazocal N_{A\to B}\). Fix \(k\ge1\), and let \(\pazocal E_k:\mathcal X\to\mathcal D(\mathcal H_A^{\otimes k})\) be a generic encoding function with $\mathcal{X}$ being an arbitrary input space of finite size. For each \(r\ge1\), set \(n=kr\) and define the induced classical-quantum process
\bb
\tilde{\pazocal M}^{(r)}
\coloneqq
\tilde{\pazocal N}^{(kr)}\circ \pazocal E_k^{\otimes r}.
\ee
Then \(\tilde{\pazocal M}=(\tilde{\pazocal M}^{(r)})_r\) is an almost i.i.d.\ process along the classical-quantum channel
\bb
\pazocal M_k
\coloneqq
\pazocal N^{\otimes k}\circ\pazocal E_k.
\ee
Indeed, by Lemma~\ref{lem:kr}, we can upper bound
\bb
    \frac 1r\|\tilde{\pazocal{M}}^{(r)}-\pazocal M_k^{\otimes r}\|_{\clubsuit, r}&=\frac 1r\sup_{x^r}\|\;\tilde{\pazocal N}^{(kr)}\circ \pazocal E_k^{\otimes r}(x^r)-(\pazocal N^{\otimes k})^{\otimes r}\circ\pazocal E_k^{\otimes r}(x^r)\,\|_{W_1^r}\\
    &\leq \frac 1r\sup_{\rho\in\mathcal{D}(\mathcal{H}_A^{\otimes kr})}\|\;\tilde{\pazocal N}^{(kr)}(\rho)-(\pazocal N^{\otimes k})^{\otimes r}(\rho)\,\|_{W_1^r}\\
    &\leq \frac 1r\sup_{\rho\in\mathcal{D}(\mathcal{H}_A^{\otimes kr})}\|\;\tilde{\pazocal N}^{(kr)}(\rho)-(\pazocal N^{\otimes k})^{\otimes r}(\rho)\,\|_{W_1^{kr}}\\
    &=k\times \frac 1{kr}\|\tilde{\pazocal N}^{(kr)}-\pazocal{N}^{\otimes kr}\|_{\clubsuit,kr}
\ee
By Proposition~\ref{prop:achiev_cq},
\bb
C(\tilde{\pazocal M})
\ge
C(\pazocal M_k)
=
\chi(\pazocal N^{\otimes k}\circ\pazocal E_k).
\ee
Moreover, any code for \(\tilde{\pazocal M}^{(r)}\) with \(M_r\) messages naturally induces a code for \(\tilde{\pazocal N}^{(kr)}\) with the same number of messages. Hence, a communication rate \(\frac1r\log M_r\) for \(\tilde{\pazocal M}\) corresponds to the rate \(\frac1{kr}\log M_r\) for \(\tilde{\pazocal N}\). Therefore,
\bb
C(\tilde{\pazocal N})
\ge
\frac1k\,C(\tilde{\pazocal M})
\ge
\frac1k\,
\chi(\pazocal N^{\otimes k}\circ\pazocal E_k).
\ee
By arbitrariness of $k\geq 1$ and of $\pazocal{E}_k$, we get
\bb
    C(\tilde{\pazocal{M}})\geq \lim_{k\to\infty}\frac 1k\sup_{\pazocal{E}_k}\chi(\pazocal{N}^{\otimes k}\circ \pazocal{E}_k)=C(\pazocal{N}).
\ee
This argument concludes the achievability part of Theorem~\ref{thm:channels}.}

\subsection{Reliability function}\label{sec:reliability}
We have just proved that, for Wasserstein almost i.i.d.\ channels, the classical capacity is robust. Is this also the case for the reliability function? A simple example in the fully classical setting serves as a counterexaple. Let $\pazocal{I}_2:\{0,1\}\to \{0,1\}$ be the binary noiseless channel. Since the communication at any rate $0<r\leq 1$ can be made errorless, the reliability function $E_R(r,\pazocal{I}_2)$ is constantly equal to $+\infty$ in $0<r\leq 1$. Calling $\pazocal{R}_0$ the classical channel that outputs the symbol $0$ with probability $1$, consider the sequence of channels
\bb
    \tilde{\pazocal{I}}_2^{(n)}\coloneqq \left(1-\frac 1n\right)\pazocal{I}_2^{\times n}+\frac{1}{n}\pazocal{R}_0^{\times n}, \qquad n\geq 1.
\ee
This is actually an almost i.i.d.\ process, as it satisfies
\bb
    \big\|\tilde{\pazocal{I}}_2^{(n)}-\pazocal{I}_2^{\times n}\big\|_\clubsuit = \frac{1}{n}\|\tilde{\pazocal{I}}_2^{\times n}-\pazocal{R}_0^{\times n}\|_\clubsuit\leq 1,
\ee
whence
\bb
    \lim_{n\to\infty}\frac 1n\big\|\tilde{\pazocal{I}}_2^{(n)}-\pazocal{I}_2^{\times n}\big\|_\clubsuit=0.
\ee
By Theorem~\ref{thm:channels}, it is possible to communicate at rate $r=1$ via $\tilde{\pazocal{I}}$ with asymptotically vanishing error, but it is not difficult to see that, for any arbitrary rate $0<r\leq 1$, no code can communicate with exponentially vanishing error probability. Indeed, let $(\pazocal{E}_n,\pazocal{D}_n)_n$ be a sequence of codes with asymptotical rate $r$ for $\tilde{\pazocal{I}}$, with 
\bb
    \pazocal{E}_n:[M_n]\to \{0,1\}^n\qquad \text{and}\qquad \pazocal{D}_n:\{0,1\}^n\to [M_n].
\ee
Then, calling $0^n$ be the sequence of length $n$ having all zeroes, let $m_n\coloneq \pazocal{D}(0^n)$. The channel $\tilde{\pazocal{I}}^{(n)}_2$, regardless of the input, with probability \errata{at least} $1/n$ outputs the sequence $0^n$, which gets decoded into $m_n$. The average probability of error can then be lower bounded as
\bb
    \frac 1 {M_n}\sum_{m=1}^{M_n}\mathbb P \big(m\neq \pazocal{D}_n\circ \tilde{\pazocal{I}}^{(n)}_2\circ \pazocal{E}_n(m)\big)\geq \frac 1 {M_n}\sum_{m=1}^{M_n} \frac{\id_{m\neq m_n}}{n}=\frac{M_n-1}{M_n}\times \frac 1n.
\ee
This immediately implies that
\bb
    E(r,\tilde{\pazocal{I}})\leq \liminf_{n\to\infty}-\frac{1}{n}\log\left(\frac{M_n-1}{M_n} \times \frac 1n\right)=0.
\ee
Therefore, in this example there is an infinite gap between the reliability function of the i.i.d.\ channel and the one of the almost i.i.d.\ process. \bigskip

With the previous counterexample in mind, by carefully looking at~\eqref{eq:p_tilde}, we can grasp what might go wrong with a Wasserstein almost i.i.d.\ process in the reliability function. More precisely, in this upper bound that, of course, might not be optimal, but still instructive, we notice that
\begin{itemize}
    \item the term (a) in~\eqref{eq:p_tilde} reproduces the exponential decay of the error probability in the i.i.d.\ setting, up to a multiplicative correction which does not affect the reliability function when $\eta\to 0$;
    \item the term (b) in~\eqref{eq:p_tilde} indicates that, with small probability, the number of symbols getting spoilt by the non-i.i.d.\ nature of the channel might be larger than the chosen threshold $n\eta$; this yields an additive correction to the bound on the error probability.
\end{itemize}
As a consequence,
\begin{itemize}
    \item if (b) does not decay exponentially fast, then the lower bound error exponent is zero, as the additive correction to the i.i.d.\ case dominates;
    \item if (b) is absent or decays exponentially with a sufficiently large rate, then the reliability function is unaffected in the limit $\eta\to 0$;
   
\end{itemize}
It is not hard to imagine that the reason beyond the counterexample we have provided above is strictly connected to the first point. We leave as an open problem the investigation of the potential settings in which the almost i.i.d.\ behaviour of sequences of channel does not affect the reliability function.

\section{Conclusion and open questions}

\errata{
In this paper, we studied three paradigmatic operational tasks in quantum Shannon theory -- hypothesis testing, data compression, and channel coding -- beyond the idealised i.i.d.\ regime, focusing instead on more realistic almost i.i.d.\ resources. Our results show that different notions of almost i.i.d.\ structure are not operationally equivalent: depending on the nature of the perturbation, the asymptotic rates of information-theoretic protocols may collapse, remain stable, or even increase when suitably tailored strategies exploit the non-i.i.d.\ structure of the resource.

At the same time, we identified classes of almost i.i.d.\ sources and channels for which robust protocols exist. These protocols achieve the same asymptotic rates and error exponents as in the ideal i.i.d.\ setting, despite 
not relying on a detailed knowledge of the precise form of the perturbation.

Our analysis focused on asymptotic notions of almost i.i.d.\ structure. An important open direction is to develop a corresponding finite-blocklength or one-shot theory, capable of quantitatively capturing the operational effects of deviations from the i.i.d.\ regime at finite system sizes.}

\subsection*{Acknowledgements} FG would like to express his gratitude to Franco Flandoli for answering some questions that arose during his probability course at SNS with an inspiring detour into classical transportation distances. FG and LL acknowledge financial support from the European Union (ERC StG ETQO, Grant Agreement no.\ 101165230). ND is supported by the Engineering and Physical Sciences Research Council [Grant Ref: EP/Y028732/1].
GDP has been supported by the UNA EUROPA SeedFunding project QUANTUMUnaE (CUP J37G25000380006).
GDP is a member of the ``Gruppo Nazionale per la Fisica Matematica (GNFM)'' of the ``Istituto Nazionale di Alta Matematica ``Francesco Severi'' (INdAM)''.

\bibliography{biblio}

\end{document}

%% file: macros.tex
\usepackage{newpxtext,newpxmath}

\let\coloneqq\relax
\let\eqqcolon\relax

\usepackage[utf8]{inputenc}
\usepackage{amsthm}
\usepackage{amssymb}
\usepackage{amsmath}
\usepackage{bbold}
\usepackage{bbm}
\usepackage[pdftex, backref=page]{hyperref}
\usepackage{braket}
\usepackage{dsfont}
\usepackage{mathdots}
\usepackage{mathtools}
\usepackage{enumerate}
\usepackage[shortlabels]{enumitem}
\usepackage{csquotes}
\usepackage{stmaryrd}
\usepackage[cal=boondox]{mathalfa}
\usepackage{graphicx}
\usepackage{stackengine}
\usepackage{scalerel}
\usepackage{tensor}       
\usepackage{array}
\usepackage{makecell}
\newcolumntype{x}[1]{>{\centering\arraybackslash}p{#1}}
\usepackage{tikz}
\usepackage{pgfplots}
\usetikzlibrary{shapes.geometric, shapes.misc, positioning, arrows, arrows.meta, decorations.pathreplacing, decorations.pathmorphing, patterns, angles, quotes, calc}
\usepackage{booktabs}
\usepackage{xfrac}
\usepackage{siunitx}
\usepackage{centernot}
\usepackage{comment}
\usepackage{chngcntr}
\usepackage{caption}
\usepackage{subcaption}

\newtheorem{thm}{Theorem}
\newtheorem*{thm*}{Theorem}
\newtheorem{prop}[thm]{Proposition}
\newtheorem*{prop*}{Proposition}
\newtheorem{lemma}[thm]{Lemma}
\newtheorem*{lemma*}{Lemma}

\newtheorem*{cor*}{Corollary}

\newtheorem*{cj*}{Conjecture}
\newtheorem{Def}[thm]{Definition}
\newtheorem*{Def*}{Definition}

\newtheorem*{question*}{Question}

\newtheorem*{problem*}{Problem}

\makeatletter
\def\thmhead@plain#1#2#3{%
  \thmname{#1}\thmnumber{\@ifnotempty{#1}{ }\@upn{#2}}%
  \thmnote{ {\the\thm@notefont#3}}}
\let\thmhead\thmhead@plain
\makeatother

\theoremstyle{definition}
\newtheorem{rem}[thm]{Remark}

\newcommand{\bb}{\begin{equation}\begin{aligned}\hspace{0pt}}
\newcommand{\bbb}{\begin{equation*}\begin{aligned}}
\newcommand{\ee}{\end{aligned}\end{equation}}
\newcommand{\eee}{\end{aligned}\end{equation*}}
\newcommand*{\coloneqq}{\mathrel{\vcenter{\baselineskip0.5ex \lineskiplimit0pt \hbox{\scriptsize.}\hbox{\scriptsize.}}} =}
\newcommand*{\eqqcolon}{= \mathrel{\vcenter{\baselineskip0.5ex \lineskiplimit0pt \hbox{\scriptsize.}\hbox{\scriptsize.}}}}
\newcommand\floor[1]{\left\lfloor#1\right\rfloor}

\newcommand{\eqt}[1]{\stackrel{\mathclap{\scriptsize \mbox{#1}}}{=}}
\newcommand{\leqt}[1]{\stackrel{\mathclap{\scriptsize \mbox{#1}}}{\leq}}

\newcommand{\geqt}[1]{\stackrel{\mathclap{\scriptsize \mbox{#1}}}{\geq}}
\newcommand{\ketbra}[1]{\ket{#1}\!\!\bra{#1}}

\newcommand{\sumno}{\sum\nolimits}

\newcommand{\e}{\varepsilon}
\renewcommand{\epsilon}{\varepsilon}

\newcommand{\id}{\mathds{1}}

\newcommand{\N}{\mathds{N}}

\newcommand{\E}{\mathds{E}}

\DeclareMathOperator{\Tr}{Tr}

\DeclareMathAlphabet{\pazocal}{OMS}{zplm}{m}{n}

\DeclareMathOperator{\pr}{Pr}
\DeclareMathOperator{\supp}{supp}

\DeclareMathOperator{\stein}{Stein}

\newcommand{\HH}{\mathcal{H}}
\newcommand{\T}{\pazocal{T}}

\newcommand{\EE}{\pazocal{E}}
\newcommand{\MM}{\pazocal{M}}
\newcommand{\D}{\mathcal{D}}

\newcommand{\XX}{\mathcal{X}}

\newcommand{\PP}{\pazocal{P}}

\newcommand{\lsmatrix}{\left(\begin{smallmatrix}}
\newcommand{\rsmatrix}{\end{smallmatrix}\right)}

\newcommand{\deff}[1]{\emph{#1}}
\newcommand{\rel}[3]{#1\big(#2\,\big\|\,#3\big)}

\stackMath
\newcommand\xxrightarrow[2][]{\mathrel{%
  \setbox2=\hbox{\stackon{\scriptstyle#1}{\scriptstyle#2}}%
  \stackunder[5pt]{%
    \xrightarrow{\makebox[\dimexpr\wd2\relax]{$\scriptstyle#2$}}%
  }{%
   \scriptstyle#1\,%
  }%
}}

\newcommand{\tendsn}[1]{\xxrightarrow[\! n\rightarrow \infty\!]{#1}}

\stackMath

\makeatletter
\newcommand*\rel@kern[1]{\kern#1\dimexpr\macc@kerna}
\newcommand*\widebar[1]{%
  \begingroup
  \def\mathaccent##1##2{%
    \rel@kern{0.8}%
    \overline{\rel@kern{-0.8}\macc@nucleus\rel@kern{0.2}}%
    \rel@kern{-0.2}%
  }%
  \macc@depth\@ne
  \let\math@bgroup\@empty \let\math@egroup\macc@set@skewchar
  \mathsurround\z@ \frozen@everymath{\mathgroup\macc@group\relax}%
  \macc@set@skewchar\relax
  \let\mathaccentV\macc@nested@a
  \macc@nested@a\relax111{#1}%
  \endgroup
}

\counterwithin*{equation}{part}
\counterwithin*{thm}{part}
\counterwithin*{figure}{part}

\tikzset{meter/.append style={draw, inner sep=10, rectangle, font=\vphantom{A}, minimum width=30, line width=.8, path picture={\draw[black] ([shift={(.1,.3)}]path picture bounding box.south west) to[bend left=50] ([shift={(-.1,.3)}]path picture bounding box.south east);\draw[black,-latex] ([shift={(0,.1)}]path picture bounding box.south) -- ([shift={(.3,-.1)}]path picture bounding box.north);}}}
\tikzset{roundnode/.append style={circle, draw=black, fill=gray!20, thick, minimum size=10mm}}
\tikzset{squarenode/.style={rectangle, draw=black, fill=none, thick, minimum size=10mm}}

\definecolor{Blues5seq1}{RGB}{239,243,255}
\definecolor{Blues5seq2}{RGB}{189,215,231}
\definecolor{Blues5seq3}{RGB}{107,174,214}
\definecolor{Blues5seq4}{RGB}{49,130,189}
\definecolor{Blues5seq5}{RGB}{8,81,156}

\definecolor{Greens5seq1}{RGB}{237,248,233}
\definecolor{Greens5seq2}{RGB}{186,228,179}
\definecolor{Greens5seq3}{RGB}{116,196,118}
\definecolor{Greens5seq4}{RGB}{49,163,84}
\definecolor{Greens5seq5}{RGB}{0,109,44}

\definecolor{Reds5seq1}{RGB}{254,229,217}
\definecolor{Reds5seq2}{RGB}{252,174,145}
\definecolor{Reds5seq3}{RGB}{251,106,74}
\definecolor{Reds5seq4}{RGB}{222,45,38}
\definecolor{Reds5seq5}{RGB}{165,15,21}

\allowdisplaybreaks

\usepackage[most,breakable]{tcolorbox}
\newenvironment{boxedthm}[1]%
	{\expandafter\ifstrequal\expandafter{#1}{orange}{\begin{tcolorbox}[colback=red!15,colframe=orange!15,breakable,enhanced]}{\begin{tcolorbox}[colback=Blues5seq1,colframe=Blues5seq5,breakable,enhanced]}}%
	{\end{tcolorbox}}

%% file: almost.pdf_tex
\begingroup%
  \makeatletter%
  \providecommand\color[2][]{%
    \errmessage{(Inkscape) Color is used for the text in Inkscape, but the package 'color.sty' is not loaded}%
    \renewcommand\color[2][]{}%
  }%
  \providecommand\transparent[1]{%
    \errmessage{(Inkscape) Transparency is used (non-zero) for the text in Inkscape, but the package 'transparent.sty' is not loaded}%
    \renewcommand\transparent[1]{}%
  }%
  \providecommand\rotatebox[2]{#2}%
  \newcommand*\fsize{\dimexpr\f@size pt\relax}%
  \newcommand*\lineheight[1]{\fontsize{\fsize}{#1\fsize}\selectfont}%
  \ifx\svgwidth\undefined%
    \setlength{\unitlength}{380.17526702bp}%
    \ifx\svgscale\undefined%
      \relax%
    \else%
      \setlength{\unitlength}{\unitlength * \real{\svgscale}}%
    \fi%
  \else%
    \setlength{\unitlength}{\svgwidth}%
  \fi%
  \global\let\svgwidth\undefined%
  \global\let\svgscale\undefined%
  \makeatother%
  \begin{picture}(1,0.4066788)%
    \lineheight{1}%
    \setlength\tabcolsep{0pt}%
    \put(0,0){\includegraphics[width=\unitlength,page=1]{almost.pdf}}%
    \put(0.09707351,0.38659857){\color[rgb]{0,0,0}\makebox(0,0)[lt]{\lineheight{1.25}\smash{\begin{tabular}[t]{l}weakly almost i.i.d.\end{tabular}}}}%
    \put(-0.00060762,0.06378345){\color[rgb]{0,0,0}\makebox(0,0)[lt]{\lineheight{1.25}\smash{\begin{tabular}[t]{l}Wasserstein almost i.i.d.\end{tabular}}}}%
    \put(0.77726561,0.04703612){\color[rgb]{0,0,0}\makebox(0,0)[lt]{\lineheight{1.25}\smash{\begin{tabular}[t]{l}MSR almost i.i.d.\end{tabular}}}}%
    \put(0,0){\includegraphics[width=\unitlength,page=2]{almost.pdf}}%
    \put(0.4079873,0.30348001){\color[rgb]{0,0,0}\makebox(0,0)[lt]{\lineheight{1.25}\smash{\begin{tabular}[t]{l}$\tilde p_n$\end{tabular}}}}%
    \put(0.48244395,0.23545102){\color[rgb]{0,0,0}\makebox(0,0)[lt]{\lineheight{1.25}\smash{\begin{tabular}[t]{l}$\xi_n$\end{tabular}}}}%
    \put(0.29213611,0.24237575){\color[rgb]{0,0,0}\makebox(0,0)[lt]{\lineheight{1.25}\smash{\begin{tabular}[t]{l}$\Psi_n$\end{tabular}}}}%
    \put(0,0){\includegraphics[width=\unitlength,page=3]{almost.pdf}}%
    \put(0.40225365,0.0001737){\color[rgb]{0,0,0}\makebox(0,0)[lt]{\lineheight{1.25}\smash{\begin{tabular}[t]{l}trace distance almost i.i.d.\end{tabular}}}}%
    \put(0,0){\includegraphics[width=\unitlength,page=4]{almost.pdf}}%
    \put(0.62009634,0.17025699){\color[rgb]{0,0,0}\makebox(0,0)[lt]{\lineheight{1.25}\smash{\begin{tabular}[t]{l}$\eta_n$\end{tabular}}}}%
    \put(0.68957729,0.18282651){\color[rgb]{0,0,0}\makebox(0,0)[lt]{\lineheight{1.25}\smash{\begin{tabular}[t]{l}$\gamma_n$\end{tabular}}}}%
    \put(0.5249737,0.16852607){\color[rgb]{0,0,0}\makebox(0,0)[lt]{\lineheight{1.25}\smash{\begin{tabular}[t]{l}$\tilde\sigma_n$\end{tabular}}}}%
  \end{picture}%
\endgroup%

%% file: almost_iid_channel.pdf_tex
\begingroup%
  \makeatletter%
  \providecommand\color[2][]{%
    \errmessage{(Inkscape) Color is used for the text in Inkscape, but the package 'color.sty' is not loaded}%
    \renewcommand\color[2][]{}%
  }%
  \providecommand\transparent[1]{%
    \errmessage{(Inkscape) Transparency is used (non-zero) for the text in Inkscape, but the package 'transparent.sty' is not loaded}%
    \renewcommand\transparent[1]{}%
  }%
  \providecommand\rotatebox[2]{#2}%
  \newcommand*\fsize{\dimexpr\f@size pt\relax}%
  \newcommand*\lineheight[1]{\fontsize{\fsize}{#1\fsize}\selectfont}%
  \ifx\svgwidth\undefined%
    \setlength{\unitlength}{505.26555579bp}%
    \ifx\svgscale\undefined%
      \relax%
    \else%
      \setlength{\unitlength}{\unitlength * \real{\svgscale}}%
    \fi%
  \else%
    \setlength{\unitlength}{\svgwidth}%
  \fi%
  \global\let\svgwidth\undefined%
  \global\let\svgscale\undefined%
  \makeatother%
  \begin{picture}(1,0.20652522)%
    \lineheight{1}%
    \setlength\tabcolsep{0pt}%
    \put(0,0){\includegraphics[width=\unitlength,page=1]{almost_iid_channel.pdf}}%
    \put(0.02099201,0.12608123){\color[rgb]{0,0,0}\makebox(0,0)[lt]{\lineheight{1.25}\smash{\begin{tabular}[t]{l}\textit{$m$}\end{tabular}}}}%
    \put(0.38911074,0.12608123){\color[rgb]{0,0,0}\makebox(0,0)[lt]{\lineheight{1.25}\smash{\begin{tabular}[t]{l}\textit{$m'$}\end{tabular}}}}%
    \put(0,0){\includegraphics[width=\unitlength,page=2]{almost_iid_channel.pdf}}%
    \put(0.21642781,0.07220858){\color[rgb]{0,0,0}\makebox(0,0)[lt]{\lineheight{1.25}\smash{\begin{tabular}[t]{l}\textit{$\vdots$}\end{tabular}}}}%
    \put(0.09139638,0.11064769){\color[rgb]{0,0,0}\makebox(0,0)[lt]{\lineheight{1.25}\smash{\begin{tabular}[t]{l}\textit{$\pazocal{E}$}\end{tabular}}}}%
    \put(0.33287698,0.11219104){\color[rgb]{0,0,0}\makebox(0,0)[lt]{\lineheight{1.25}\smash{\begin{tabular}[t]{l}\textit{$\pazocal{D}$}\end{tabular}}}}%
    \put(0.20825211,0.16775173){\color[rgb]{0,0,0}\makebox(0,0)[lt]{\lineheight{1.25}\smash{\begin{tabular}[t]{l}\textit{$\pazocal{N}$}\end{tabular}}}}%
    \put(0.20825211,0.11728317){\color[rgb]{0,0,0}\makebox(0,0)[lt]{\lineheight{1.25}\smash{\begin{tabular}[t]{l}\textit{$\pazocal{N}$}\end{tabular}}}}%
    \put(0.20825211,0.02525227){\color[rgb]{0,0,0}\makebox(0,0)[lt]{\lineheight{1.25}\smash{\begin{tabular}[t]{l}\textit{$\pazocal{N}$}\end{tabular}}}}%
    \put(0,0){\includegraphics[width=\unitlength,page=3]{almost_iid_channel.pdf}}%
    \put(0.51366649,0.12608123){\color[rgb]{0,0,0}\makebox(0,0)[lt]{\lineheight{1.25}\smash{\begin{tabular}[t]{l}\textit{$m$}\end{tabular}}}}%
    \put(0.89069204,0.12608123){\color[rgb]{0,0,0}\makebox(0,0)[lt]{\lineheight{1.25}\smash{\begin{tabular}[t]{l}\textit{$\tilde m'$}\end{tabular}}}}%
    \put(0,0){\includegraphics[width=\unitlength,page=4]{almost_iid_channel.pdf}}%
    \put(0.58407062,0.11064769){\color[rgb]{0,0,0}\makebox(0,0)[lt]{\lineheight{1.25}\smash{\begin{tabular}[t]{l}\textit{$\pazocal{E}$}\end{tabular}}}}%
    \put(0.83445811,0.11219104){\color[rgb]{0,0,0}\makebox(0,0)[lt]{\lineheight{1.25}\smash{\begin{tabular}[t]{l}\textit{$\pazocal{D}$}\end{tabular}}}}%
    \put(0.69275435,0.10752521){\color[rgb]{0,0,0}\makebox(0,0)[lt]{\lineheight{1.25}\smash{\begin{tabular}[t]{l}\textit{$\tilde{\pazocal{N}}^{(n)}$}\end{tabular}}}}%
  \end{picture}%
\endgroup%

%% file: almosthptest.pdf_tex
\begingroup%
  \makeatletter%
  \providecommand\color[2][]{%
    \errmessage{(Inkscape) Color is used for the text in Inkscape, but the package 'color.sty' is not loaded}%
    \renewcommand\color[2][]{}%
  }%
  \providecommand\transparent[1]{%
    \errmessage{(Inkscape) Transparency is used (non-zero) for the text in Inkscape, but the package 'transparent.sty' is not loaded}%
    \renewcommand\transparent[1]{}%
  }%
  \providecommand\rotatebox[2]{#2}%
  \newcommand*\fsize{\dimexpr\f@size pt\relax}%
  \newcommand*\lineheight[1]{\fontsize{\fsize}{#1\fsize}\selectfont}%
  \ifx\svgwidth\undefined%
    \setlength{\unitlength}{264.66356413bp}%
    \ifx\svgscale\undefined%
      \relax%
    \else%
      \setlength{\unitlength}{\unitlength * \real{\svgscale}}%
    \fi%
  \else%
    \setlength{\unitlength}{\svgwidth}%
  \fi%
  \global\let\svgwidth\undefined%
  \global\let\svgscale\undefined%
  \makeatother%
  \begin{picture}(1,0.42782669)%
    \lineheight{1}%
    \setlength\tabcolsep{0pt}%
    \put(0,0){\includegraphics[width=\unitlength,page=1]{almosthptest.pdf}}%
    \put(-0.00232742,0.24546257){\color[rgb]{0,0,0}\makebox(0,0)[lt]{\lineheight{1.25}\smash{\begin{tabular}[t]{l}$H_0$\end{tabular}}}}%
    \put(0,0){\includegraphics[width=\unitlength,page=2]{almosthptest.pdf}}%
    \put(-0.00130829,0.15543593){\color[rgb]{0,0,0}\makebox(0,0)[lt]{\lineheight{1.25}\smash{\begin{tabular}[t]{l}$H_1$\end{tabular}}}}%
    \put(0,0){\includegraphics[width=\unitlength,page=3]{almosthptest.pdf}}%
    \put(0.62364518,0.19910666){\color[rgb]{0,0,0}\makebox(0,0)[lt]{\lineheight{1.25}\smash{\begin{tabular}[t]{l}$?$ \end{tabular}}}}%
    \put(0.75470197,0.19669676){\color[rgb]{0,0,0}\makebox(0,0)[lt]{\lineheight{1.25}\smash{\begin{tabular}[t]{l}$E_n$\end{tabular}}}}%
    \put(0,0){\includegraphics[width=\unitlength,page=4]{almosthptest.pdf}}%
    \put(0.94002579,0.2385344){\color[rgb]{0,0,0}\makebox(0,0)[lt]{\lineheight{1.25}\smash{\begin{tabular}[t]{l}$0$\end{tabular}}}}%
    \put(0.94002579,0.15918986){\color[rgb]{0,0,0}\makebox(0,0)[lt]{\lineheight{1.25}\smash{\begin{tabular}[t]{l}$1$\end{tabular}}}}%
    \put(0,0){\includegraphics[width=\unitlength,page=5]{almosthptest.pdf}}%
    \put(0.40108951,0.25045365){\color[rgb]{0,0,0}\makebox(0,0)[lt]{\lineheight{1.25}\smash{\begin{tabular}[t]{l}\textbf{\textit{$\rho^{\otimes n}$}}\end{tabular}}}}%
    \put(0.39872795,0.15953126){\color[rgb]{0,0,0}\makebox(0,0)[lt]{\lineheight{1.25}\smash{\begin{tabular}[t]{l}\textbf{\textit{$\sigma^{\otimes n}$}}\end{tabular}}}}%
    \put(0,0){\includegraphics[width=\unitlength,page=6]{almosthptest.pdf}}%
    \put(0.27788447,0.24527529){\color[rgb]{0,0,0}\makebox(0,0)[lt]{\lineheight{1.25}\smash{\begin{tabular}[t]{l}$\cdots$\end{tabular}}}}%
    \put(0,0){\includegraphics[width=\unitlength,page=7]{almosthptest.pdf}}%
    \put(0.27788447,0.1545942){\color[rgb]{0,0,0}\makebox(0,0)[lt]{\lineheight{1.25}\smash{\begin{tabular}[t]{l}$\cdots$\end{tabular}}}}%
    \put(0,0){\includegraphics[width=\unitlength,page=8]{almosthptest.pdf}}%
    \put(-0.00232742,0.37585732){\color[rgb]{0,0,0}\makebox(0,0)[lt]{\lineheight{1.25}\smash{\begin{tabular}[t]{l}$\tilde{H}_0$\end{tabular}}}}%
    \put(0,0){\includegraphics[width=\unitlength,page=9]{almosthptest.pdf}}%
    \put(0.27788447,0.37562959){\color[rgb]{0,0,0}\makebox(0,0)[lt]{\lineheight{1.25}\smash{\begin{tabular}[t]{l}$\cdots$\end{tabular}}}}%
    \put(0,0){\includegraphics[width=\unitlength,page=10]{almosthptest.pdf}}%
    \put(0.40108951,0.3751402){\color[rgb]{0,0,0}\makebox(0,0)[lt]{\lineheight{1.25}\smash{\begin{tabular}[t]{l}\textbf{\textit{$\rho_n$}}\end{tabular}}}}%
    \put(0,0){\includegraphics[width=\unitlength,page=11]{almosthptest.pdf}}%
    \put(-0.00232742,0.02446756){\color[rgb]{0,0,0}\makebox(0,0)[lt]{\lineheight{1.25}\smash{\begin{tabular}[t]{l}$\tilde{H}_1$\end{tabular}}}}%
    \put(0,0){\includegraphics[width=\unitlength,page=12]{almosthptest.pdf}}%
    \put(0.27788447,0.02424006){\color[rgb]{0,0,0}\makebox(0,0)[lt]{\lineheight{1.25}\smash{\begin{tabular}[t]{l}$\cdots$\end{tabular}}}}%
    \put(0.40108951,0.02375068){\color[rgb]{0,0,0}\makebox(0,0)[lt]{\lineheight{1.25}\smash{\begin{tabular}[t]{l}\textbf{\textit{$\sigma_n$}}\end{tabular}}}}%
    \put(0,0){\includegraphics[width=\unitlength,page=13]{almosthptest.pdf}}%
  \end{picture}%
\endgroup%

%% file: data_compression.pdf_tex
\begingroup%
  \makeatletter%
  \providecommand\color[2][]{%
    \errmessage{(Inkscape) Color is used for the text in Inkscape, but the package 'color.sty' is not loaded}%
    \renewcommand\color[2][]{}%
  }%
  \providecommand\transparent[1]{%
    \errmessage{(Inkscape) Transparency is used (non-zero) for the text in Inkscape, but the package 'transparent.sty' is not loaded}%
    \renewcommand\transparent[1]{}%
  }%
  \providecommand\rotatebox[2]{#2}%
  \newcommand*\fsize{\dimexpr\f@size pt\relax}%
  \newcommand*\lineheight[1]{\fontsize{\fsize}{#1\fsize}\selectfont}%
  \ifx\svgwidth\undefined%
    \setlength{\unitlength}{418.27269547bp}%
    \ifx\svgscale\undefined%
      \relax%
    \else%
      \setlength{\unitlength}{\unitlength * \real{\svgscale}}%
    \fi%
  \else%
    \setlength{\unitlength}{\svgwidth}%
  \fi%
  \global\let\svgwidth\undefined%
  \global\let\svgscale\undefined%
  \makeatother%
  \begin{picture}(1,0.16967385)%
    \lineheight{1}%
    \setlength\tabcolsep{0pt}%
    \put(0.99913589,0.07055804){\color[rgb]{0,0,0}\makebox(0,0)[rt]{\lineheight{1.25}\smash{\begin{tabular}[t]{r}\textit{$X_1,  \dots, X_n$}\end{tabular}}}}%
    \put(0,0){\includegraphics[width=\unitlength,page=1]{data_compression.pdf}}%
    \put(0.09799362,0.03366514){\color[rgb]{0,0,0}\makebox(0,0)[lt]{\lineheight{1.25}\smash{\begin{tabular}[t]{l}\textit{$X_1\quad X_2\quad \cdots \quad X_n$}\end{tabular}}}}%
    \put(0,0){\includegraphics[width=\unitlength,page=2]{data_compression.pdf}}%
    \put(0.09799362,0.12332095){\color[rgb]{0,0,0}\makebox(0,0)[lt]{\lineheight{1.25}\smash{\begin{tabular}[t]{l}\textit{$X_1\quad X_2\quad \cdots \quad X_n$}\end{tabular}}}}%
    \put(0,0){\includegraphics[width=\unitlength,page=3]{data_compression.pdf}}%
    \put(0.49889748,0.0719187){\color[rgb]{0,0,0}\makebox(0,0)[lt]{\lineheight{1.25}\smash{\begin{tabular}[t]{l}\textit{$\pazocal{E}_n$}\end{tabular}}}}%
    \put(0,0){\includegraphics[width=\unitlength,page=4]{data_compression.pdf}}%
    \put(0.57420706,0.09343609){\color[rgb]{0,0,0}\makebox(0,0)[lt]{\lineheight{1.25}\smash{\begin{tabular}[t]{l}\textit{$m\in[M]$}\end{tabular}}}}%
    \put(0,0){\includegraphics[width=\unitlength,page=5]{data_compression.pdf}}%
    \put(0.74275626,0.0719187){\color[rgb]{0,0,0}\makebox(0,0)[lt]{\lineheight{1.25}\smash{\begin{tabular}[t]{l}\textit{$\pazocal{D}_n$}\end{tabular}}}}%
    \put(0,0){\includegraphics[width=\unitlength,page=6]{data_compression.pdf}}%
    \put(-0.00241944,0.12332095){\color[rgb]{0,0,0}\makebox(0,0)[lt]{\lineheight{1.25}\smash{\begin{tabular}[t]{l}(a)\end{tabular}}}}%
    \put(-0.00241944,0.03007891){\color[rgb]{0,0,0}\makebox(0,0)[lt]{\lineheight{1.25}\smash{\begin{tabular}[t]{l}(b)\end{tabular}}}}%
  \end{picture}%
\endgroup%

%% file: q_data_comp.pdf_tex
\begingroup%
  \makeatletter%
  \providecommand\color[2][]{%
    \errmessage{(Inkscape) Color is used for the text in Inkscape, but the package 'color.sty' is not loaded}%
    \renewcommand\color[2][]{}%
  }%
  \providecommand\transparent[1]{%
    \errmessage{(Inkscape) Transparency is used (non-zero) for the text in Inkscape, but the package 'transparent.sty' is not loaded}%
    \renewcommand\transparent[1]{}%
  }%
  \providecommand\rotatebox[2]{#2}%
  \newcommand*\fsize{\dimexpr\f@size pt\relax}%
  \newcommand*\lineheight[1]{\fontsize{\fsize}{#1\fsize}\selectfont}%
  \ifx\svgwidth\undefined%
    \setlength{\unitlength}{540.32881129bp}%
    \ifx\svgscale\undefined%
      \relax%
    \else%
      \setlength{\unitlength}{\unitlength * \real{\svgscale}}%
    \fi%
  \else%
    \setlength{\unitlength}{\svgwidth}%
  \fi%
  \global\let\svgwidth\undefined%
  \global\let\svgscale\undefined%
  \makeatother%
  \begin{picture}(1,0.26187784)%
    \lineheight{1}%
    \setlength\tabcolsep{0pt}%
    \put(0,0){\includegraphics[width=\unitlength,page=1]{q_data_comp.pdf}}%
    \put(0.45065554,0.1271209){\color[rgb]{0,0,0}\makebox(0,0)[lt]{\lineheight{1.25}\smash{\begin{tabular}[t]{l}\textit{$\pazocal{E}_n$}\end{tabular}}}}%
    \put(0,0){\includegraphics[width=\unitlength,page=2]{q_data_comp.pdf}}%
    \put(0.5394909,0.14377768){\color[rgb]{0,0,0}\makebox(0,0)[lt]{\lineheight{1.25}\smash{\begin{tabular}[t]{l}\textit{$\mathcal{K}_n$}\end{tabular}}}}%
    \put(0,0){\includegraphics[width=\unitlength,page=3]{q_data_comp.pdf}}%
    \put(0.63942854,0.1271209){\color[rgb]{0,0,0}\makebox(0,0)[lt]{\lineheight{1.25}\smash{\begin{tabular}[t]{l}\textit{$\pazocal{D}_n$}\end{tabular}}}}%
    \put(0,0){\includegraphics[width=\unitlength,page=4]{q_data_comp.pdf}}%
    \put(0.22427142,0.16332664){\color[rgb]{0,0,0}\makebox(0,0)[lt]{\lineheight{1.25}\smash{\begin{tabular}[t]{l}\textit{$\color{Violet}\rho^{\otimes n}$}\end{tabular}}}}%
    \put(0.22427142,0.10225266){\color[rgb]{0,0,0}\makebox(0,0)[lt]{\lineheight{1.25}\smash{\begin{tabular}[t]{l}\textit{$\color{Violet}\rho_n$}\end{tabular}}}}%
    \put(0.2131669,0.24660935){\color[rgb]{0,0,0}\makebox(0,0)[lt]{\lineheight{1.25}\smash{\begin{tabular}[t]{l}\textit{$\Psi_{A^nR}$}\end{tabular}}}}%
    \put(0.05698276,0.19888895){\color[rgb]{0,0,0}\makebox(0,0)[lt]{\lineheight{1.25}\smash{\begin{tabular}[t]{l}\textit{$R$}\end{tabular}}}}%
    \put(0.05420671,0.16557587){\color[rgb]{0,0,0}\makebox(0,0)[lt]{\lineheight{1.25}\smash{\begin{tabular}[t]{l}\textit{$A^n$}\end{tabular}}}}%
    \put(0.05698276,0.06008444){\color[rgb]{0,0,0}\makebox(0,0)[lt]{\lineheight{1.25}\smash{\begin{tabular}[t]{l}\textit{$R$}\end{tabular}}}}%
    \put(0.05420671,0.09339753){\color[rgb]{0,0,0}\makebox(0,0)[lt]{\lineheight{1.25}\smash{\begin{tabular}[t]{l}\textit{$A^n$}\end{tabular}}}}%
    \put(0.22149529,0.01064169){\color[rgb]{0,0,0}\makebox(0,0)[lt]{\lineheight{1.25}\smash{\begin{tabular}[t]{l}\textit{$\tilde{\Psi}_{A^nR}$}\end{tabular}}}}%
    \put(0.90872517,0.24660935){\color[rgb]{0,0,0}\makebox(0,0)[rt]{\lineheight{1.25}\smash{\begin{tabular}[t]{r}\textit{$\sigma_{A^nR}$}\end{tabular}}}}%
    \put(0.9173006,0.00508951){\color[rgb]{0,0,0}\makebox(0,0)[rt]{\lineheight{1.25}\smash{\begin{tabular}[t]{r}\textit{$\tilde{\sigma}_{A^nR}$}\end{tabular}}}}%
    \put(0,0){\includegraphics[width=\unitlength,page=5]{q_data_comp.pdf}}%
    \put(0.00090317,0.17823932){\color[rgb]{0,0,0}\makebox(0,0)[lt]{\lineheight{1.25}\smash{\begin{tabular}[t]{l}(a)\end{tabular}}}}%
    \put(-0.00187288,0.07552295){\color[rgb]{0,0,0}\makebox(0,0)[lt]{\lineheight{1.25}\smash{\begin{tabular}[t]{l}(b)\end{tabular}}}}%
  \end{picture}%
\endgroup%

%% file: coupling.pdf_tex
\begingroup%
  \makeatletter%
  \providecommand\color[2][]{%
    \errmessage{(Inkscape) Color is used for the text in Inkscape, but the package 'color.sty' is not loaded}%
    \renewcommand\color[2][]{}%
  }%
  \providecommand\transparent[1]{%
    \errmessage{(Inkscape) Transparency is used (non-zero) for the text in Inkscape, but the package 'transparent.sty' is not loaded}%
    \renewcommand\transparent[1]{}%
  }%
  \providecommand\rotatebox[2]{#2}%
  \newcommand*\fsize{\dimexpr\f@size pt\relax}%
  \newcommand*\lineheight[1]{\fontsize{\fsize}{#1\fsize}\selectfont}%
  \ifx\svgwidth\undefined%
    \setlength{\unitlength}{367.78260239bp}%
    \ifx\svgscale\undefined%
      \relax%
    \else%
      \setlength{\unitlength}{\unitlength * \real{\svgscale}}%
    \fi%
  \else%
    \setlength{\unitlength}{\svgwidth}%
  \fi%
  \global\let\svgwidth\undefined%
  \global\let\svgscale\undefined%
  \makeatother%
  \begin{picture}(1,0.28615852)%
    \lineheight{1}%
    \setlength\tabcolsep{0pt}%
    \put(0,0){\includegraphics[width=\unitlength,page=1]{coupling.pdf}}%
    \put(0.5366511,0.09920133){\color[rgb]{0,0,0}\makebox(0,0)[lt]{\lineheight{1.25}\smash{\begin{tabular}[t]{l}\textit{$\vdots$}\end{tabular}}}}%
    \put(0.52541919,0.23046001){\color[rgb]{0,0,0}\makebox(0,0)[lt]{\lineheight{1.25}\smash{\begin{tabular}[t]{l}\textit{$\pazocal{W}$}\end{tabular}}}}%
    \put(0.52541919,0.16112545){\color[rgb]{0,0,0}\makebox(0,0)[lt]{\lineheight{1.25}\smash{\begin{tabular}[t]{l}\textit{$\pazocal{W}$}\end{tabular}}}}%
    \put(0.52541919,0.034692){\color[rgb]{0,0,0}\makebox(0,0)[lt]{\lineheight{1.25}\smash{\begin{tabular}[t]{l}\textit{$\pazocal{W}$}\end{tabular}}}}%
    \put(0,0){\includegraphics[width=\unitlength,page=2]{coupling.pdf}}%
    \put(0.016214,0.17401207){\color[rgb]{0,0,0}\makebox(0,0)[lt]{\lineheight{1.25}\smash{\begin{tabular}[t]{l}\textit{$X^n$}\end{tabular}}}}%
    \put(0.09702074,0.15015074){\color[rgb]{0,0,0}\makebox(0,0)[lt]{\lineheight{1.25}\smash{\begin{tabular}[t]{l}\textit{$\tilde{\pazocal{W}}^{(n)}$}\end{tabular}}}}%
    \put(0.20382498,0.17401207){\color[rgb]{0,0,0}\makebox(0,0)[lt]{\lineheight{1.25}\smash{\begin{tabular}[t]{l}\textit{$\tilde{Y}^n$}\end{tabular}}}}%
    \put(0.42814244,0.17401207){\color[rgb]{0,0,0}\makebox(0,0)[lt]{\lineheight{1.25}\smash{\begin{tabular}[t]{l}\textit{$X^n$}\end{tabular}}}}%
    \put(0,0){\includegraphics[width=\unitlength,page=3]{coupling.pdf}}%
    \put(0.69671479,0.14073324){\color[rgb]{0,0,0}\makebox(0,0)[lt]{\lineheight{1.25}\smash{\begin{tabular}[t]{l}\textit{$\Phi^\ast_{X^n}$}\end{tabular}}}}%
    \put(0.6110674,0.17743921){\color[rgb]{0,0,0}\makebox(0,0)[lt]{\lineheight{1.25}\smash{\begin{tabular}[t]{l}\textit{$Y^n$}\end{tabular}}}}%
    \put(0.80275426,0.17743921){\color[rgb]{0,0,0}\makebox(0,0)[lt]{\lineheight{1.25}\smash{\begin{tabular}[t]{l}\textit{$\tilde{Y}^n$}\end{tabular}}}}%
    \put(0.33025849,0.14954142){\color[rgb]{0,0,0}\makebox(0,0)[lt]{\lineheight{1.25}\smash{\begin{tabular}[t]{l}\textit{$=$}\end{tabular}}}}%
  \end{picture}%
\endgroup%